\documentclass[onecolumn,draftclsnofoot]{IEEEtran}

\usepackage{amsmath,cite,amsfonts,amssymb,psfrag}
\usepackage{graphicx}
\usepackage{rotating}
\usepackage{citesort}

\newtheorem{theorem}{Theorem}

\newtheorem{remark}{Remark}
\newtheorem{definition}{Definition}
\newtheorem{proposition}{Proposition}
\newtheorem{lemma}{Lemma}
\newtheorem{corollary}{Corollary}
\newtheorem{proof}{Proof}

\begin{document}
\title{Coding With Action-dependent Side Information and Additional
  Reconstruction Requirements}

\author{
\authorblockN{Kittipong Kittichokechai, Tobias J. Oechtering, and Mikael Skoglund\\}
\authorblockA{KTH Royal Institute of Technology, Stockholm, Sweden\\
School of Electrical Engineering and the ACCESS Linnaeus Center\\
Email: kki@kth.se, oech@kth.se, and skoglund@ee.kth.se
}
}
\maketitle
%%%%%%%%%%%%%%%%%%%%%%%%%%%%%%%%%%%%%%%%%%%%%%%%%%%%%%%%%%%%%%%%%%%%%%%%%%%%%%%%%%%%%%%%%%%%%%%%%%%%%%%%%%%%%%%%%%%%%%%%%%
\begin{abstract}
  Constrained lossy source coding and channel coding with side
  information problems which extend the classic Wyner--Ziv and
  Gel'fand--Pinsker problems are considered. Inspired by applications
  in sensor networking and control, we first consider lossy source
  coding with two-sided partial side information where the
  quality/availability of the side information can be influenced by a
  cost-constrained action sequence. A decoder reconstructs a source
  sequence subject to the distortion constraint, and at the same time,
  an encoder is additionally required to be able to estimate the
  decoder's reconstruction. Next, we consider the channel coding
  ``dual'' where the channel state is assumed to depend on the action
  sequence, and the decoder is required to decode both the transmitted
  message and channel input reliably.

  Implications on the fundamental limits of communication in discrete
  memoryless systems due to the additional reconstruction constraints
  are investigated. Single-letter expressions for the
  rate-distortion-cost function and channel capacity for the
  respective source and channel coding problems are derived. The dual
  relation between the two problems is discussed. Additionally, based
  on the two-stage coding structure and the additional reconstruction
  constraint of the channel coding problem, we discuss and give an
  interpretation of the \emph{two-stage coding} condition which
  appears in the channel capacity expression. Besides the rate
  constraint on the message, this condition is a necessary and
  sufficient condition for reliable transmission of the channel input
  sequence over the channel in our ``two-stage'' communication
  problem. It is also shown in one example that there exists a case
  where the two-stage coding condition can be active in computing the
  capacity, and it thus can actively restrict the set of capacity
  achieving input distributions.
\end{abstract}

%%%%%%%%%%%%%%%%%%%%%%%%%%%%%%%%%%%%%%%%%%%%%%%%%%%%%%%%%%%%%%%%%%%%%%%%%%%%%%%%%%%%%%%%%%%%%%%%%%%%%%%%%%%%%%%%%%%%%%%%%%
\section{Introduction}

The problems of source coding with side information and channel coding
with state information have received considerable attention due to
their broad set of applications, e.g., in high-definition television
where the noisy analog version of the TV signal is the side
information at the receiver, in cognitive radio where the secondary
user has knowledge of the message to be transmitted by the primary
user, or in digital watermarking where the host signal plays a role of
state information available at the transmitter
\cite{Cover2002},\cite{Keshet2007}. In \cite{WynerZiv} Wyner and Ziv
considered rate-distortion coding for a source with side information
available at the receiver, while the problem of coding for channels
with noncausal state information available at the transmitter was
solved by Gel'fand and Pinsker in \cite{GelfandPinsker}.  In practice,
the transmitter and/or the receiver may not have full knowledge of the
channel state information. Heegard and El Gamal in \cite{Heegard1983}
studied the channel with rate-limited noncausal state information
available at the encoder and/or the decoder. Further, Cover and Chiang
in \cite{Cover2002} provided a unifying framework to characterize
channel capacity and rate-distortion functions for systems with
two-sided partial state information, and they also discuss aspects of
duality between the source and channel coding problems.

In this work we consider source and channel coding with two-sided
partial side/state information where the side/state information can be
influenced by other nodes in the system. Such side/state information
is termed as \emph{action-dependent} side/state information
\cite{Weissman2010},\cite{Permuter2011}. Weissman studied first a
problem of coding for a channel with action-dependent state
\cite{Weissman2010}, and the source coding dual was investigated by
Permuter and Weissman \cite{Permuter2011} where a node in the system
can take action to influence the quality/availability of the side
information. This novel action-dependent coding framework introduces
new interesting features to the general system model, involving
cost-constrained communication and interaction among nodes, and is
therefore highly relevant to many applications including sensor
networking and control, and multistage coding for memories
\cite{Weissman2010},\cite{Permuter2011}. Additional work on coding
with action includes \cite{Asnani2011} where it is natural to consider
action probing as a means for channel state acquisition, and in
\cite{Chia2011},\cite{Ahmadi2011c} where the problem of source coding
with action-dependent side information is extended to the
multi-terminal case.

In addition, we are interested in the recently introduced problem of
lossy source coding with side information under the additional
requirement that the sender should be able to locally produce an exact
copy of the receiver's reconstruction. This requirement was introduced
and termed the \emph{common reconstruction} (CR) constraint by
Steinberg \cite{Steinberg2009}. The general case of additional
reconstruction subject to the distortion constraint was later studied
in \cite{Lapidoth2011}. The channel coding dual is also investigated
in the context of information embedding by Sumszyk and Steinberg in
\cite{Sumszyk2009} where the decoder is interested in decoding both an
embedded message and a stegotext signal. There, it is shown that if
the objective is to decode only the message and the stegotext (channel
input signal), then decoding the message and the channel state first
and then re-encoding the channel input is suboptimal. As with
action-dependent coding, also the framework of additional
reconstruction requirements provides new useful features of
simultaneous signal transmission in the general system model. Recent
works on common reconstruction in multi-terminal information theoretic
problems include \cite{Timo2010pre}, \cite{Ahmadi2011pre}. Some
closely related works on additional signal reconstruction include
\cite{Kim2008}, \cite{Willems}.

In the present work we unify the problems of action-dependence
and common reconstruction constraints by studying   source
and channel coding with action-dependent partial side information
known noncausally at the encoder and the decoder, and with additional
reconstruction constraints.
%%%
The constrained source coding problem is an extension of Wyner--Ziv
lossy source coding where the encoder is additionally required to
estimate the decoder's reconstruction reliably and the available
two-sided partial side information depends on a cost-constrained
action sequence. This setting captures the problem of simultaneously
controlling the quality of the decoder's reconstruction via the
action-dependent side information, and monitoring the resulting
performance via common reconstruction. As a motivating example,
consider a closed-loop control system. Assuming that there exists a
coding scheme which satisfies the CR constraint, an observer/encoder
having knowledge about the reconstruction at a controller/decoder
will have the possibility to compensate for possible impact of state
reconstruction distortion and thus achieve better control performance
in future time instants. The unified system modeled with both
action-dependent side information and the CR constraint can also be
viewed as a resource-efficient system, i.e., the quality of side
information can be adjusted on demand and the control objective can be
achieved more efficiently due to the knowledge of the controller's
reconstruction at the observer.
%%%
On the other hand, the constrained channel coding dual is an extension
of the Gel'fand-Pinsker problem where the channel state is allowed to
depend on an action sequence and the decoder is additionally required
to reconstruct the channel input signal reliably. This setting
captures the idea of simultaneously transmitting the message and the
channel input sequence reliably over the channel. To be consistent
with the terminology used in \cite{Sumszyk2009}, we refer to the
reconstruction constraint as the \emph{reversible input} (RI)
constraint. This setup is for example relevant in a data storage
problem where a user is interested in both decoding the embedded
message and in tracing what has been written in the previous
stages. It may also be relevant in a wireless networking scenario
where knowing the channel input signal can enable interference
mitigation at some node in the network.

In this work, we characterize fundamental limits of discrete
memoryless systems, and discuss the implication of additional
reconstruction constraints. An investigation on the dual relationship
between the problems is also of interest. We note that different kinds
of duality between various source and channel coding problems with
side information (SI) have been recognized earlier. For example,
several works have discussed duality between the Wyner--Ziv and
Gel'fand--Pinsker problems
\cite{Cover2002},\cite{Pradhan2003},\cite{Barron2003},\cite{Gupta2011}. Our
definition of duality simply follows the notion of ``formula'' duality
in \cite{Cover2002}. Although it is not based on a strict definition
like in other work, it is appealing that one might be able to
anticipate the optimal solution of a new problem from its dual
problem.
\footnotetext[1]{After submission, we got aware of two recent works \cite{Choudhuri2012pre},\cite{Zaidi2012pre} in which a similar two-stage coding condition appears in a similar fashion as an extra constraint resulted from the additional reconstruction requirements in the two-stage communication setting.}

Our source and channel coding problems are ``dually'' formulated,
i.e., an encoder in one problem has the same functionality as a
decoder in the other problem. However, there are some fundamental
differences in their operational structure. As we will show, the
source coding setup requires causal processing at the encoder for
compressing the source using action-dependent side information, while
in the case of channel coding the channel decoder can observe the
channel output and the channel state information noncausally. In
addition, the channel coding scenario requires sequential two-stage
processing at the encoder in generating an action-dependent state
sequence and then a channel input sequence. When we impose an
additional constraint on decoding a signal generated in the later
stage (channel input $X^{n}$) at the decoder, an extra condition,
apart from the rate constraint, is needed. This leads us to the
conclusion that formula duality between our problems does not hold. We
term the new condition which appears in the channel coding problem the
\emph{two-stage coding condition}\footnotemark[1] since it arises essentially from the
two-stage operational structure of the setting that requires the
channel input reconstruction. In addition to the rate constraint, we
show that the two-stage coding condition is a necessary and sufficient
condition for reliable transmission of the channel input signal over
the channel in our two-stage communication problem. We also discuss
different aspects of the presence of the two-stage coding condition in
the channel capacity problem, based on operational, source coding, and
channel coding perspectives. Finally, we show in one of the examples
that there exists a case where the two-stage coding condition can be
active when computing the capacity, and it can thus actively restrict
the set of capacity achieving input distributions. The material in
this paper was presented in part in \cite{Kittichokechai2010a},
\cite{Kittichokechai2010b}, \cite{Kittichokechai2011a}, and
\cite{Kittichokechai2011c}.

The remaining parts of the paper are organized as follows. In Section
\ref{sec:sourcecoding1} we formulate the problem of source coding with
action-dependent two-sided partial SI and CR constraint. We derive a
closed-form expression for the rate-distortion-cost function. Other
related results and a binary example illustrating an implication of
common reconstruction constraint on the rate-distortion-cost tradeoff
are given. The channel coding dual is presented in Section
\ref{sec:channelcoding2}, where the channel capacity is found in a
form with the two-stage coding condition. In this section we also
present other related results as well as an example showing that the
two-stage coding condition can be active in some cases. We discuss the
presence of the two-stage coding condition as well as the dual
relations among the related problems in Section
\ref{sec:discussion}. The conclusion is provided in Section
\ref{sec:conclusion}.

\textit{Notation}: We denote the discrete random variables, their
corresponding realizations or deterministic values, and their
alphabets by the upper case, lower case, and calligraphic letters,
respectively. The term $X_{m}^{n}$ denotes the sequence
$\{X_{m},\ldots,X_{n}\}$ when $m\leq n$, and the empty set
otherwise. Also, we use the shorthand notation $X^{n}$ for
$X_{1}^{n}$. The term $X^{n\setminus i}$ denotes the set
$\{X_{1},\ldots,X_{i-1},X_{i+1},\ldots,X_{n}\}$. Cardinality of the
set $\mathcal{X}$ is denoted by $|\mathcal{X}|$. Finally, we use
$X-Y-Z$ to denote a Markov chain formed by the joint distribution of
$(X,Y,Z)$ that is factorized as
$P_{X,Y,Z}(x,y,z)=P_{X,Y}(x,y)P_{Z|Y}(z|y)$ or
$P_{X,Y,Z}(x,y,z)=P_{X|Y}(x|y)P_{Y,Z}(y,z)$.
%%%%%%%%%%%%%%%%%%%%%%%%%%%%%%%%%%%%%%%%%%%%%%%%%%%%%%%%%%%%%%%%%%%%%%%%%%%%%%%%%%%%%%%%%%%%%%%%%%%%%%%%%%%%%%%%%%%%%%%%%%
\section{Source Coding with Action-dependent Side Information and CR
  Constraint}\label{sec:sourcecoding1}

In this section we study source coding with action-dependent side
information and CR constraint as depicted in
Fig.~\ref{fig:SystemModel1}. The side information is generated based
on the source and cost-constrained action sequences, and are given at
both encoder and decoder. The decoder reconstructs the source sequence
subject to the distortion constraint. Meanwhile, the encoder is
required to locally produce an exact copy of the decoder's
reconstruction. This scenario captures the idea of simultaneously
controlling the quality of the decoder's reconstruction via
action-dependent side information, and monitoring the decoder's
reconstruction via common reconstruction. Our setup can be considered
as a combination of Permuter and Weissman's source coding with side
information ``vending machine'' \cite{Permuter2011} and Steinberg's
coding and common reconstruction \cite{Steinberg2009}.

In the following, we present the problem formulation, characterize the
main result which is the rate-distortion-cost function of the setting,
and also present some other related results. Finally, a binary example
is given to illustrate an implication of the common reconstruction on
the rate-distortion-cost tradeoff.

\subsection{Problem Formulation and Main Results}
\begin{figure}[]
    \centering
    \psfrag{x}[][][0.8]{$X^{n}$}
    \psfrag{enc}[][][0.8]{Encoder}
    \psfrag{w}[][][0.8]{$w$, rate $R$}
   % \psfrag{t}[][][0.8]{$t\in [1,2^{nR_{2}}]$}
    \psfrag{act}[][][0.8]{Action}
    \psfrag{en}[][][0.8]{encoder}
    \psfrag{Aw}[][][0.8]{$A^{n}$}
    \psfrag{Ps1s2xa}[][][0.8]{$P_{S_{e},S_{d}|X,A}$}
    \psfrag{s1s2}[][][0.8]{$S^{n}=(S_{e}^{n},S_{d}^{n})$}
    \psfrag{s}[][][0.8]{$S^{n}$}
    \psfrag{s1}[][][0.8]{$S_{e}^{n}$}
    \psfrag{s2}[][][0.8]{$S_{d}^{n}$}
    \psfrag{k1}[][][0.8]{$l_{e}$}
    \psfrag{k2}[][][0.8]{$l_{d}$}
    \psfrag{dec}[][][0.8]{Decoder}
    \psfrag{de}[][][0.8]{decoder}
    \psfrag{xtil}[][][0.8]{$\tilde{X}^{n}, E[d(X^{n},\tilde{X}^{n})]\lesssim D$}
    \psfrag{xhat}[][][0.8]{$\hat{X}^{n}= \psi(X^{n},S_{e}^{n},A^{n}), \mathrm{Pr}(\hat{X}^{n} \neq \tilde{X}^{n})\rightarrow 0$}
    \includegraphics[width=10cm]{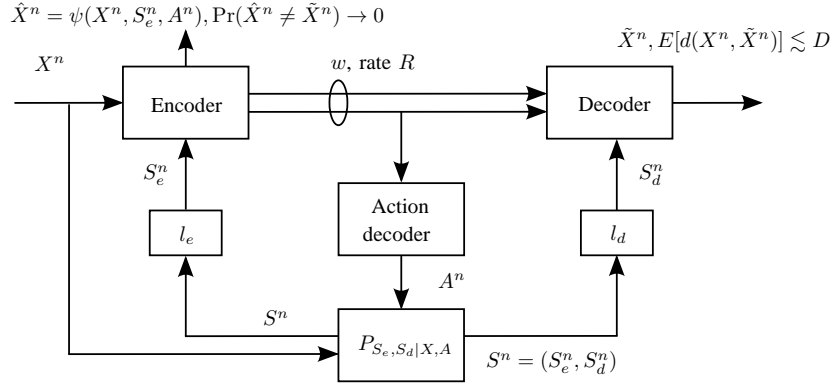}
     \caption{Rate distortion with action-dependent partial side information and CR constraint.}\label{fig:SystemModel1}
\end{figure}

We consider finite alphabets for the source, action, side information,
and reconstruction sets, i.e., $\mathcal{X}$, $\mathcal{A}$,
$\mathcal{S}_{e}$, $\mathcal{S}_{d}$, and $\hat{\mathcal{X}}$ are
finite. Let $X^{n}$ be a source sequence of length $n$ with
i.i.d.~elements according to $P_{X}$. Given a source sequence $X^{n}$,
an encoder generates an index representing the source sequence and
sends it over a noise-free, rate-limited link to an action decoder and
a source decoder. An action sequence is then selected based on the
index. With input $(X^{n},A^{n})$ whose current symbols do not depend
on the previous channel output, the side information
$(S_{e}^{n},S_{d}^{n})$ is generated as an output of the memoryless
channel with transition
probability \[P_{S_{e}^{n},S_{d}^{n}|X^{n},A^{n}}(s_{e}^{n},s_{d}^{n}|x^{n},a^{n})=\prod_{i=1}^{n}P_{S_{e},S_{d}|X,A}(s_{e,i},s_{d,i}|x_{i},a_{i}).\]

The side information is then mapped to the partial side information
for the encoder and the decoder by the mappings
$l_{e}^{(n)}(S_{e}^{n},S_{d}^{n})=S_{e}^{n}$ and
$l_{d}^{(n)}(S_{e}^{n},S_{d}^{n})=S_{d}^{n}$. Next, the encoder uses
knowledge about $S_{e}^{n}$ to generate another index and sends it to
the source decoder. Given the indices and the side information
$S_{d}^{n}$ the source decoder reconstructs the source sequence as
$\tilde{X}^{n}$. On the other hand, the encoder also estimates the
decoder's reconstruction as $\hat{X}^{n}$.

\begin{definition}
  An $(|\mathcal{W}^{(n)}|,n)$-code for a memoryless source with
  partially known two-sided action-dependent side information and a CR constraint consists of the following functions:\\
  an encoder 1 \[f_{1}^{(n)}: \mathcal{X}^{n} \rightarrow
  \mathcal{W}_1^{(n)}, \mathcal{W}_1^{(n)} =
  \{1,2,\ldots,|\mathcal{W}_1^{(n)}| \},\] an action
  decoder \[g_{a}^{(n)}: \mathcal{W}_1^{(n)} \rightarrow
  \mathcal{A}^{n},\] an encoder 2 \[f_{2}^{(n)}: \mathcal{X}^{n}\times
  \mathcal{S}_e^{n} \rightarrow \mathcal{W}_2^{(n)},
  \mathcal{W}_2^{(n)} = \{1,2,\ldots,|\mathcal{W}_2^{(n)}| \},\] a
  source decoder \[g^{(n)}: \mathcal{W}_1^{(n)}\times
  \mathcal{W}_2^{(n)}\times \mathcal{S}_{d}^{n} \rightarrow
  \hat{\mathcal{X}}^{n},\] and a CR mapper \[\psi^{(n)}:
  \mathcal{X}^{n}\times \mathcal{S}_e^{n} \times \mathcal{A}^{n}
  \rightarrow \mathcal{\hat{X}}^{n},\] where $|\mathcal{W}^{(n)}| =
  |\mathcal{W}_1^{(n)}| \cdot |\mathcal{W}_2^{(n)}|$.

  Let $d: \mathcal{X} \times \hat{\mathcal{X}} \rightarrow [0,\infty)$
  and $\Lambda: \mathcal{A} \rightarrow [0,\infty)$ be the bounded
  single-letter distortion and cost measures. The average distortion
  between a length-$n$ source sequence and its reconstruction at the
  decoder, and the average cost are defined as
\begin{align*}
     &E\left[d^{(n)}\big(X^{n},\tilde{X}^{n}\big)\right] \triangleq \frac{1}{n}E\left[\sum_{i=1}^{n}d(X_{i},\tilde{X}_{i})\right], \\
     & \qquad E\left[\Lambda^{(n)}(A^{n})\right] \triangleq \frac{1}{n}E\left[\sum_{i=1}^{n}\Lambda(A_{i})\right],
\end{align*}
where $d^{(n)}(\cdot)$ and $\Lambda^{(n)}(\cdot)$ are the distortion and cost functions, respectively.

The average probability of error in estimating the decoder's reconstruction sequence is defined by
\[P^{(n)}_{\text{CR}} = \mathrm{Pr}\big(\psi^{(n)}(X^{n},S_{e}^{n},A^{n})\neq g^{(n)}(f_{1}^{(n)}(X^{n}),f_{2}^{(n)}(X^{n},S_e^{n}),S_{d}^{n})\big).\]
\end{definition}

\begin{definition} A rate-distortion-cost triple $(R,D,C)$ is said to be \emph{achievable} if for any $\delta >0$, there exists for all sufficiently large $n$ an
$(|\mathcal{W}^{(n)}|,n)$-code such that $\frac{1}{n}\log|\mathcal{W}^{(n)}| \leq R+\delta$,
\[E\big[d^{(n)}(X^{n},\tilde{X}^{n})\big] \leq D+\delta, \qquad
E\big[\Lambda^{(n)}(A^{n})\big] \leq C+\delta,\ \mbox{and} \qquad
P^{(n)}_{\text{CR}}  \leq \delta.\]
The \emph{rate-distortion-cost function} $R_{\text{ac,cr}}(D,C)$ is the infimum of the achievable rates at distortion level $D$ and cost $C$.
\end{definition}

\begin{theorem}\label{theoremRateDist}
  The rate-distortion-cost function for the source with a CR constraint
  and action-dependent partial side information available at the
  encoder and the decoder is given by
\begin{equation}\label{eq:rateDistor1}
    R_{\text{ac,cr}}(D,C) = \min [I(X;A) + I(\hat{X};X,S_{e}|A) - I(\hat{X};S_{d}|A)],
\end{equation}
where the joint distribution of $(X,A,S_{e},S_{d},\hat{X})$ is of the form
\begin{equation*}
 P_{X}(x)P_{A|X}(a|x)P_{S_{e},S_{d}|X,A}(s_{e},s_{d}|x,a)P_{\hat{X}|X,S_{e},A}(\hat{x}|x,s_{e},a)
\end{equation*}
and the minimization is over all $P_{A|X}$ and $ P_{\hat{X}|X,S_{e},A}$ subject to \[E\big[d(X,\hat{X})\big] \leq D,\qquad E[\Lambda(A)] \leq C.\]
%$U$ is the auxiliary random variable with $|\mathcal{U}|\leq |\mathcal{A}||\mathcal{X}|+3$.
\end{theorem}

\begin{proof}
  The proof follows similar arguments as in \cite{Permuter2011} with
  some modifications in which we extend the SI-channel transition
  probability to the two-sided SI $P_{S_{e},S_{d}|X,A}$, and consider
  the additional CR constraint at the encoder as in
  \cite{Steinberg2009}. In the following, we give a sketch of the
  achievability proof. An action codebook $\{a^{n}\}$ of size
  $2^{n(I(X;A)+ \delta_{\epsilon})}$ is generated i.i.d. $\sim
  P_{A}$. For each $a^{n}$ another codebook $\{\hat{x}^{n}\}$ of size
  $2^{(n(I(\hat{X};X,S_e|A)+ \delta_{\epsilon}))}$ is generated
  i.i.d.~$\sim P_{\hat{X}|A}$. These codewords are then distributed at
  random into
  $2^{n(I(\hat{X};X,S_e|A)-I(\hat{X};S_d|A)+2\delta_{\epsilon})}$
  equal-sized bins (see Fig. \ref{fig:binningSC}). Given the source
  sequence $x^{n}$ the encoder in the first step uses $n(I(X;A)+
  \delta_{\epsilon})$ bits to transmit an index representing the
  action codeword $a^{n}$ which is jointly typical with $x^{n}$ to the
  decoder. Then the action-dependent SI is generated based on $x^{n}$
  and $a^{n}$. Given $x^{n},s_{e}^{n}$ and previously chosen $a^{n}$,
  the encoder in the second step uses another
  $n(I(\hat{X};X,S_e|A)-I(\hat{X};S_d|A)+2\delta_{\epsilon})$ bits to
  communicate the bin index of the jointly typical codeword
  $\hat{x}^{n}$. In addition, the encoder produces this jointly
  typical $\hat{x}^{n}$ as an estimate of the decoder's
  reconstruction. Given the identity of $a^{n}$, the bin index of
  $\hat{x}^{n}$, and the side information $s_{d}^{n}$, the decoder
  will find with high probability the unique codeword $\hat{x}^{n}$ in
  its bin that is jointly typical with $s_{d}^{n}$ and
  $a^{n}$. Finally, the decoder reconstructs
  $\tilde{x}^{n}=\hat{x}^{n}$. For completeness, we provide the
  detailed achievability proof and converse proof in Appendix
  \ref{sec:proofsourcecoding}.
\end{proof}

    \begin{figure}[]
    %\vspace{1.3cm}
    \centering
    \psfrag{e}[][][0.75]{$2^{n(I(\hat{X};S_{d}|A)-\delta_{\epsilon})}$}
    \psfrag{dot}[][][0.75]{$\cdots$}
    \psfrag{x}[][][0.75]{a codeword $\hat{X}^{n}$}
    \psfrag{p1}[][][0.75]{$2^{n(I(\hat{X};X,S_{e}|A)-I(\hat{X};S_{d}|A)+2\delta_{\epsilon})}$ bins}
    \includegraphics[width=9cm]{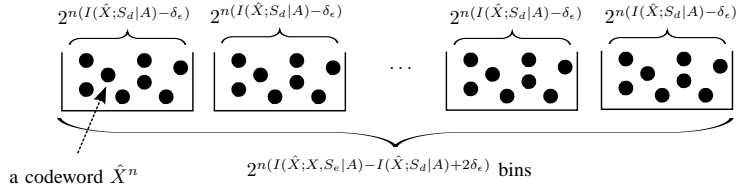}
    %\small Binning for the achievability: for each codeword $a^{n}$, a codebook $\{\hat{x}^{n}\}$ of size $2^{n(I(\hat{X};X,S_{e}|A)+ \epsilon)}$ is generated \\ i.i.d. each $\sim P_{\hat{X}|A}$. Then they are distributed uniformly into $2^{n(I(X;Y|A)-I(X;S_{e}|A)-2\epsilon)}$ equal-sized bins.
    \caption{Binning for the achievability: for each codeword $a^{n}$, a codebook $\{\hat{x}^{n}\}$ of size $2^{n(I(\hat{X};X,S_{e}|A)+ \delta_{\epsilon})}$ is generated i.i.d. each $\sim P_{\hat{X}|A}$. Then they are distributed uniformly into $2^{n(I(\hat{X};X,S_{e}|A)-I(\hat{X};S_{d}|A)+2\delta_{\epsilon})}$ equal-sized bins.} \label{fig:binningSC}
    \end{figure}

\begin{remark}
We can also express $R_{\text{ac,cr}}(D,C)$ in \eqref{eq:rateDistor1} as
\begin{align}\label{eq:alterR2}
  R_{\text{ac,cr}}(D,C) %&= \min[I(X;A)+I(U;X,S_{e}|A)-I(U;S_{d}|A)] \nonumber \\
   &=  \min[I(X;A)+H(\hat{X}|A)-H(\hat{X}|A,X,S_{e}) -H(\hat{X}|A)+H(\hat{X}|A,S_{d})] \nonumber\\
   &\overset{(*)}{=} \min[I(X;A)-H(\hat{X}|A,X,S_{e},S_{d})+H(\hat{X}|A,S_{d})] \nonumber \\
   &= \min[I(X;A)+I(\hat{X};X,S_{e}|A,S_{d})],
\end{align}
where $(*)$ follows from the Markov chain $\hat{X}-(X,A,S_{e})-S_{d}$ and the minimization is over the same distribution as in \eqref{eq:rateDistor1}.
\end{remark}

\begin{lemma}\label{lemma:convexRDC}
The rate-distortion-cost function $R_{\text{ac,cr}}(D,C)$ given in \eqref{eq:rateDistor1} and \eqref{eq:alterR2} is a non-increasing convex function of distortion $D$ and cost $C$.
\end{lemma}

\begin{proof}
Proof is given in Appendix \ref{sec:prooflemma}.
\end{proof}

\subsection{Other Results}\label{sec:specialcase1}
In the following, we provide some connecting conclusions which help
develop our understanding and also relate our main result to other
known results in the literature. We consider the case where the common
reconstruction constraint is omitted and then our setting recovers
the source coding with action-dependent SI setup of
\cite{Permuter2011}. On the other hand, if we have no control over the
SI, then our setting simply recovers source coding with common
reconstruction \cite{Steinberg2009}. We might also consider a special
case where side information at the encoder or the decoder is
absent. The result in this case can be derived straightforwardly by
setting the SI to be a constant value.

\begin{proposition}\label{theoremRateDist2}
When the additional CR constraint is omitted, the rate-distortion-cost function for the source with action-dependent partial side information available at the encoder and the decoder (no CR) is given by
\begin{equation}\label{eq:rateDistor2}
    R_{\text{ac}}(D,C) = \min [I(X;A) + I(U;X,S_{e}|A) - I(U;S_{d}|A)],
\end{equation}
where the joint distribution of $(X,A,S_{e},S_{d},U)$ is of the form
\begin{equation*}
 P_{X}(x)P_{A|X}(a|x)P_{S_{e},S_{d}|X,A}(s_{e},s_{d}|x,a)P_{U|X,S_{e},A}(u|x,s_{e},a)
\end{equation*}
and the minimization is over all $P_{A|X}, P_{U|X,S_{e},A}$ and $\tilde{g}: \mathcal{U}\times \mathcal{S}_{d} \rightarrow \mathcal{\hat{X}}$ subject to \[E\big[d\big(X,\tilde{g}(U,S_{d})\big)\big] \leq D,\qquad E[\Lambda(A)] \leq C,\]
and $U$ is the auxiliary random variable with $|\mathcal{U}|\leq |\mathcal{A}||\mathcal{X}|+3$.
\end{proposition}

\begin{proof}
The rate-distortion-cost function in this case can be derived along the lines of Theorem \ref{theoremRateDist}. The achievability proof is a straightforward modification of that of Theorem \ref{theoremRateDist} where the codeword $U^{n}$ is used instead of $X^{n}$ and the decoding function $\tilde{g}$ is introduced (similarly as in the Wyner-Ziv problem). The converse proof is given in Appendix \ref{sec:proofISIT}.
\end{proof}

\begin{corollary} \label{sec:specialcase2}
For a special case where the side information at the encoder is absent, the rate-distortion-cost function for the source with action-dependent side information available at the decoder (and CR constraint) can be derived as a special case of Proposition \ref{theoremRateDist2} (Theorem \ref{theoremRateDist}) by setting $S_{e}$ to a constant value.
\end{corollary}

\begin{corollary}
For the case where the side information at the encoder is absent and we have no control over the SI at the decoder, i.e., the action alphabet size is one, the rate-distortion function for the source with CR constraint is given by
\begin{equation}\label{eq:C}
    R_{\text{cr}}(D) = \min[I(\hat{X};X|S_{d})],
\end{equation}
where the joint distribution of $(X,S_{d},\hat{X})$ is of the form
\begin{equation*}
 P_{X}(x)P_{S_{d}|X}(s_{d}|x)P_{\hat{X}|X}(\hat{x}|x)
\end{equation*}
and the minimization is over all $ P_{\hat{X}|X}$ subject to
$E\big[d(X,\hat{X})\big] \leq D$.  Note that this result recovers
Theorem 1 in Steinberg's coding and common reconstruction
\cite{Steinberg2009}.
\end{corollary}

Since the action sequence is taken based on a rate-limited link which
is part of the total rate from the encoder to the decoder (see
Fig. \ref{fig:SystemModel1}), in some cases, we might be interested in
characterizing the individual rate constraint in the form of a rate
region. Here we consider the same setting as in
Fig. \ref{fig:SystemModel1}, but we assume that the rate on the link
used for generating the action sequence is denoted by $R_{1}$, and the
remaining rate from the encoder to the decoder is denoted by $R_{2}$.
\begin{corollary}
The rate-distortion-cost region is given by the set of all $(R_{1},R_{2},D,C)$ satisfying
\begin{align*}
    R_{1} &\geq I(X;A)\\
    R_{1}+ R_{2} &\geq I(X;A) + I(\hat{X};X,S_{e}|A,S_{d})\\
    D &\geq E[d(X,\hat{X})]\\
    C &\geq E[\Lambda(A)],
\end{align*}
where the joint distribution of $(X,A,S_{e},S_{d},\hat{X})$ is of the form
\begin{equation*}
 P_{X}(x)P_{A|X}(a|x)P_{S_{e},S_{d}|X,A}(s_{e},s_{d}|x,a)P_{\hat{X}|X,S_{e},A}(\hat{x}|x,s_{e},a).
\end{equation*}
Note that the result is related to the successive refinement rate-distortion region where we might consider the action sequence as a reconstruction sequence in the first stage, and the refinement stage involves the side information available at the encoder and the decoder $(S_{e},S_{d})$. We also note that the rate-distortion-cost function in Theorem \ref{theoremRateDist} is simply a constraint on the total rate $R=R_{1}+R_{2}$ for a given distortion $D$ and cost $C$.
\end{corollary}
\begin{proof}
The proof is a modification of that of Theorem \ref{theoremRateDist} where we consider instead the individual rate constraints. More specifically, the achievable scheme of Theorem \ref{theoremRateDist} is modified so that the index $W_{1}$ is split into two independent parts $(W_{1,1},W_{1,2})$, and the action sequence is selected based on only $W_{1,1}$. In the converse, the sum-rate constraint is the same as in the converse proof of Theorem \ref{theoremRateDist}, while the constraint on $R_{1}$ can be derived straightforwardly using the techniques from the point-to-point lossy source coding.
\end{proof}

\subsection{Binary Example}\label{sec:example}
We will show an example of the rate-distortion-cost function for the special case considered in Corollary \ref{sec:specialcase2} where the SI at the encoder is absent. Our example is a combination of examples in \cite{Permuter2011}  and \cite{Steinberg2009} which are based on the Wyner-Ziv example \cite{WynerZiv} and illustrate nicely the expected behavior of the rate-distortion function due to the implication of  action-dependent side information with cost \cite{Permuter2011} and common reconstruction constraint \cite{Steinberg2009}.

We consider a given source and side information distribution $P_{X}$,
$P_{S_{d}|X,A}$. We assume binary action $A \in \mathcal{A}=\{0,1\}$
with $A=1$ corresponding to observing the side information symbol and
$A=0$ to not observing it. We assume that an observation has a unit
cost, i.e., $\Lambda(A)=A$ and $E[\Lambda(A)]=P_{A}(1)=C$. We note
that the second mutual information term in \eqref{eq:alterR2}
neglecting $S_e$ corresponds to the CR rate-distortion function
\cite[eq.(8)]{Steinberg2009} conditioned on $A$. Let $D_{i}$ be the
contribution to the average distortion given $A=i$, $i=0,1$, i.e.,
$(1-C)D_{0}+CD_{1}=D$. Thus, the specialization of Theorem
\ref{theoremRateDist} for this case gives
\begin{align}\label{eq:Rac,cr}
  &R_{\text{ac,cr}}(D,C) =
  \min_{P_{A|X},P_{A}(1)=C,(1-C)D_{0}+CD_{1}=D}I(X;A) +(1-C)\cdot
  R(P_{X|A=0},D_{0}) + C \cdot R_{\text{cr}}(P_{X,S_{d}|A=1},D_{1}),
\end{align}
where $R(P_{X},D)$ denotes the rate-distortion function of the source
$P_{X}$ without side information and $R_{\text{cr}}(P_{X,S_{d}},D)$ denotes
the CR rate-distortion function defined in \cite{Steinberg2009} when
source and side information are jointly distributed according to
$P_{X,S_{d}}$.

It is interesting to compare $R_{\text{ac,cr}}(D,C)$ to the
rate-distortion-cost function of the case without the CR constraint
$R_{\text{ac}}(D,C)$ (a special case of \eqref{eq:rateDistor2} when
neglecting $S_e$) to see how much we have to ``pay" for satisfying the
additional CR constraint. In this case
\begin{align}\label{eq:Rac}
&R_{\text{ac}}(D,C) = \min_{P_{A|X},P_{A}(1)=C,(1-C)D_{0}+CD_{1}=D}I(X;A) +(1-C) \cdot R(P_{X|A=0},D_{0})  + C \cdot R_{\text{wz}}(P_{X,S_{d}|A=1},D_{1}),
\end{align}
where $R_{\text{wz}}(P_{X,S_{d}},D)$ denotes the Wyner-Ziv rate-distortion function when source and side information are jointly distributed according to $P_{X,S_{d}}$.
We note that the difference between \eqref{eq:Rac,cr} and \eqref{eq:Rac} is only in their last terms.

Let us consider a binary symmetric source, a binary reconstruction,
and a symmetric side information channel when actions are taken to
observe the side information. That is,
$\mathcal{X}=\mathcal{\hat{X}}=\mathcal{S}_{d}=\{0,1\}$, where $X$ is
distributed according to Bernoulli(1/2), and the side information
$S_{d}$ is given as an output of a binary symmetric channel with input
$X$ and crossover probability $p_{0}$ when $A=1$. The Hamming distance
is considered as the distortion measure.

In \cite[Example 1]{Steinberg2009} the author computes the CR rate-distortion function for this source,
\[R_{\text{cr}}(P_{X,S_{d}|A=1},D) = h(p_{0}\star D)-h(D), \quad 0\leq D\leq 1/2,\]
where $h(\cdot)$ is the binary entropy function and $p_{0}\star D \triangleq p_{0}(1-D)+(1-p_{0})D$.
As known from \cite{WynerZiv} the Wyner-Ziv rate-distortion function for this source is given by
\[R_{\text{wz}}(P_{X,S_{d}|A=1},D) = \inf_{\theta,\beta}\big[\theta\big(h(p_{0}\star \beta)-h(\beta)\big)\big],\]
for $ 0\leq D \leq p_{0}$, where the infimum is with respect to all $\theta, \beta$, where $0 \leq \theta \leq 1$ and $ 0 \leq \beta \leq p_{0}$ such that
$D = \theta\beta + (1-\theta)p_{0}$.
In addition, we know that $R(P_{X|A=0},D)=1-h(D)$ for this source \cite{InfoBook}.
\begin{figure}[]
    \centering
    \includegraphics[trim = 15 0 0 0,width=9.3cm]{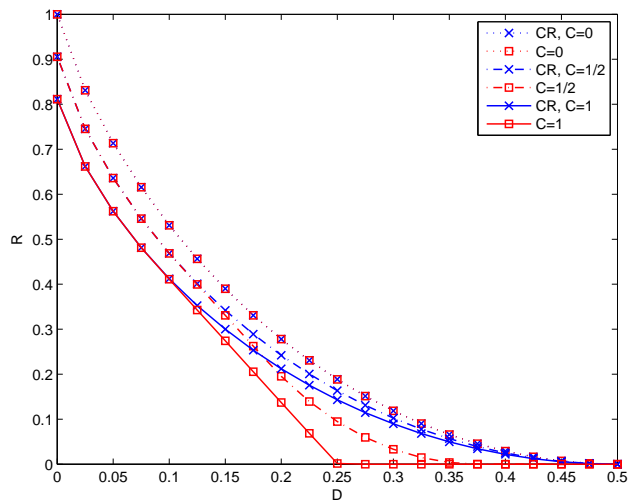}
    \vspace{-0.7cm}
     \caption{Rate-distortion curves for the binary symmetric source
       with common reconstruction and action-dependent side
       information available at the decoder. The markers $\times$
       correspond to the cases with CR constraint; the markers
       $\square$ correspond to the cases without the CR constraint. The different line styles correspond to different costs (dotted $C=0$, dashed-dotted $C=1/2$, and solid $C=1$).}\label{fig:example}
\vspace{-0.15cm}
\end{figure}

Using these results, we can compute \eqref{eq:Rac,cr} and
\eqref{eq:Rac}, and compare $R_{\text{ac,cr}}(D,C)$ and $R_{\text{ac}}(D,C)$ to
illustrate the consequences of enforcing the CR constraint. For a
given $C=0,1/2,$ and $1$, and $p_{0}=1/4$, we plot the rate-distortion
tradeoffs in Fig. \ref{fig:example}.  The plot shows that there is a
rate penalty when the CR constraint is required. This penalty changes
according to an action-cost as shown by the gap between
$R_{\text{ac,cr}}(D,C)$ and $R_{\text{ac}}(D,C)$ for different costs. Also, with the
additional CR constraint, there is a tradeoff between the action-cost
used for generating $S^{n}_{d}$ and the minimum rate one can compress
to achieve a desired distortion level. That is, ``spending" too much
on generating the SI for the decoder can negatively influence the
common reconstruction capability of the encoder, and thus affect the
minimum rate required to compress the source.

%%%%%%%%%%%%%%%%%%%%%%%%%%%%%%%%%%%%%%%%%%%%%%%%%%%%%%%%%%%%%%%%%%%%%%%%%%%%%%%%%%%%%%%%%%%%%%%%%%%%%%%%%%%%%%%%%%%%%%%%%%
\section{Channel Coding with Action-dependent State and Reversible Input}\label{sec:channelcoding2}

In this section, we consider channel coding with action-dependent
state, where the state is known partially and noncausally at the
encoder and the decoder as depicted in Fig.
\ref{fig:SystemModel4}. In addition to decoding the message, the
channel input $X^{n}$ is reconstructed with arbitrarily small error
probability at the decoder. The corresponding reconstructed signal is
termed \emph{reversible input}. This setup captures the idea of
simultaneously transmitting both the message and channel input
sequence reliably over the channel. Our setup can be considered as a
combination of Weissman's channel with action-dependent state
\cite{Weissman2010}, and Sumszyk and Steinberg's information embedding
with reversible stegotext \cite{Sumszyk2009}. It is also closely
related to the problems of reversible information embedding
\cite{Willems} and state amplification \cite{Kim2008}.

In the following, we present the problem formulation, characterize the
main result which is the capacity of a discrete memoryless channel,
and also present some other related results. The channel capacity is
given as a solution to a constrained optimization problem with a
constraint on the set of input distributions. We term this constraint
the \emph{two-stage coding} condition since it arises essentially from the
two-stage structure of the encoding as well as the additional
reconstruction constraint of a signal generated in the second
stage. Also, we show in one example that such a constraint can be
active in some cases, i.e., it actively restricts the set of capacity
achieving input distributions, and when it is active, it will be
satisfied with equality. This two-stage coding condition will be
discussed further in Section \ref{sec:discussion}.

\subsection{Problem Formulation and Main Results}

\begin{figure}[]
    \centering
    \psfrag{m}[][][0.7]{$M$}
    \psfrag{act}[][][0.8]{Action}
    \psfrag{enc}[][][0.8]{encoder}
    \psfrag{Am}[][][0.7]{$A^{n}(M)$}
    \psfrag{Ps1s2}[][][0.75]{$P_{S_{e},S_{d}|A}$}
    \psfrag{s1s2}[][][0.7]{$S^{n}=(S_{e}^{n},S_{d}^{n})$}
    \psfrag{s1}[][][0.7]{$S_{e}^{n}$}
    \psfrag{s2}[][][0.7]{$S_{d}^{n}$}
    \psfrag{ch}[][][0.8]{Channel}
    \psfrag{Pyxs}[][][0.75]{$P_{Y|X,S}$}
    \psfrag{x}[][][0.7]{$X^{n}$}
    \psfrag{y}[][][0.7]{$Y^{n}$}
    \psfrag{k1}[][][0.7]{$l_{e}$}
    \psfrag{k2}[][][0.7]{$l_{d}$}
    \psfrag{dec}[][][0.8]{Decoder}
    \psfrag{mhat}[][][0.7]{$\hat{M},\hat{X}^{n}$}
    %\psfrag{xhat}[][][0.7]{$\hat{X}^{n}$}
    \includegraphics[width=12cm]{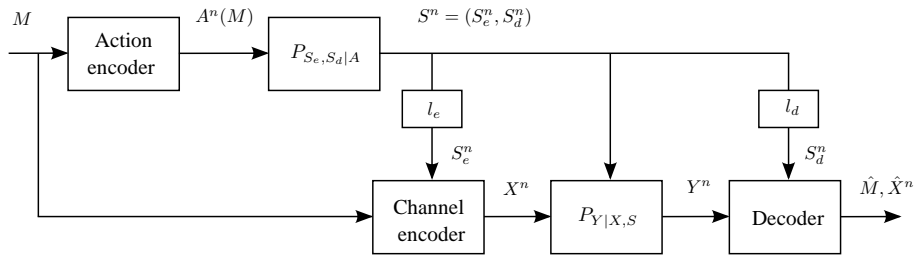}
    %\resizebox{12 cm}{!}{\epsfbox{pic.eps}}
    %\resizebox{3.0in}{!}{\includegraphics{pic.eps}}
    %\vspace{-.2cm}
     \caption{Channel with action-dependent state information and reversible channel input.}\label{fig:SystemModel4}
 %  \vspace{-.3cm}
    \centering
\end{figure}

Let $n$ denote the block length and $\mathcal{A},
\mathcal{S}_{e},\mathcal{S}_{d}, \mathcal{X}$, and $\mathcal{Y}$ be
finite sets. The system consists of two encoders, namely, an action
encoder and a channel encoder, and one decoder. A message $M$ chosen
uniformly from the set
$\mathcal{M}^{(n)}=\{1,2,\ldots,|\mathcal{M}^{(n)}|\}$ is given to
both encoders. An action sequence $A^{n}$ is chosen based on the
message $M$ and is the input to the \emph{state information channel},
described by a triple $(\mathcal{A},P_{S_{e},S_{d}|A},\mathcal{S}_{e}
\times \mathcal{S}_{d})$, where $\mathcal{A}$ is the action alphabet,
$\mathcal{S}_{e}$ and $\mathcal{S}_{d}$ are the state alphabets, and
$P_{S_{e},S_{d}|A}$ is the transition probability from $\mathcal{A}$
to $(\mathcal{S}_{e}\times \mathcal{S}_{d})$. The channel state
$S^{n}=(S_{e}^{n},S_{d}^{n})$ is mapped to the partial state
information for the encoder and the decoder by the mappings
$l_{e}^{(n)}(S_{e}^{n},S_{d}^{n})=S_{e}^{n}$ and
$l_{d}^{(n)}(S_{e}^{n},S_{d}^{n})=S_{d}^{n}$. The input to the
\emph{state-dependent channel} is denoted by $X^{n}$. This channel is
described by a quadruple $(\mathcal{X},\mathcal{S}_{e} \times
\mathcal{S}_{d},P_{Y|X,S_{e},S_{d}},\mathcal{Y})$, where $\mathcal{X}$
is the input alphabet, $\mathcal{Y}$ is the output alphabet and
$P_{Y|X,S_{e},S_{d}}$ is the transition probability from $(\mathcal{X}
\times \mathcal{S}_{e} \times \mathcal{S}_{d})$ to $\mathcal{Y}$. The
decoder, which might be considered as two separate decoders, i.e., a
message decoder and a channel input decoder, decodes the message and
the channel input based on channel output $Y^{n}$ and state
information $S_{d}^{n}$. We assume that both state information and
state-dependent channels are discrete memoryless and used without
feedback with transition probabilities,
\[P_{S_{e}^{n},S_{d}^{n}|A^{n}}(s_{e}^{n},s_{d}^{n}|a^{n})=\prod_{i=1}^{n}P_{S_{e},S_{d}|A}(s_{e,i},s_{d,i}|a_{i}),\]
\[P_{Y^{n}|X^{n},S_{e}^{n},S_{d}^{n}}(y^{n}|x^{n},s_{e}^{n},s_{d}^{n})=\prod_{i=1}^{n}P_{Y|X,S_{e},S_{d}}(y_{i}|x_{i},s_{e,i},s_{d,i}).\]

\begin{definition}
An $(|\mathcal{M}^{(n)}|,n)$ code for the channels $P_{S_{e},S_{d}|A}$ and $P_{Y|X,S_{e},S_{d}}$ consists of the following functions:\\
an action encoder
\[f_{a}^{(n)}: \mathcal{M}^{(n)} \rightarrow \mathcal{A}^{n},\]
a channel encoder
\[f^{(n)}: \mathcal{M}^{(n)}\times \mathcal{S}_{e}^{n} \rightarrow \mathcal{X}^{n},\]
a message decoder
\[g_{m}^{(n)} : \mathcal{Y}^{n}\times \mathcal{S}_{d}^{n} \rightarrow \mathcal{M}^{(n)},\]
and a channel input decoder
\[g_{x}^{(n)} : \mathcal{Y}^{n}\times \mathcal{S}_{d}^{n} \rightarrow \mathcal{X}^{n}.\]

The average probabilities of error in decoding the message $M$ and the channel input $X^{n}$ are defined by
\[P^{(n)}_{m,e} = \frac{1}{|\mathcal{M}^{(n)}|}\sum_{m,s_{e}^{n},s_{d}^{n},\\y^{n}:g^{(n)}_{m}(y^{n},s_{d}^{n}) \neq m} p(y^{n}|f^{(n)}(m,s_{e}^{n}),s_{e}^{n},s_{d}^{n}) \cdot p(s_{e}^{n},s_{d}^{n}|f^{(n)}_{a}(m)),\]
\[P^{(n)}_{x,e} = \frac{1}{|\mathcal{M}^{(n)}|}\sum_{m,s_{e}^{n},s_{d}^{n},\\y^{n}:g^{(n)}_{x}(y^{n},s_{d}^{n}) \neq f^{(n)}(m,s_{e}^{n})} p(y^{n}|f^{(n)}(m,s_{e}^{n}),s_{e}^{n},s_{d}^{n})\cdot p(s_{e}^{n},s_{d}^{n}|f^{(n)}_{a}(m)).\]

\end{definition}

\begin{definition}
A rate $R$ is said to be \emph{achievable} if for any $\delta >0$ there exists for all sufficiently large $n$ an $(|\mathcal{M}^{(n)}|,n)$-code such that
$ \frac{1}{n}\log|\mathcal{M}^{(n)}|  \geq R-\delta, P^{(n)}_{m,e} \leq \delta,$ and $P^{(n)}_{x,e} \leq \delta$. The \emph{capacity} of the channel is the supremum of all achievable rates.
\end{definition}

\begin{theorem}\label{theoremCapa}
The capacity of channels with action-dependent state available noncausally to the encoder and the decoder and with reversible input at the decoder shown in Fig. \ref{fig:SystemModel4} is given by
\begin{align}
     C = \max[ I(A,X;Y,S_{d}) - I(X;S_{e}|A)], \label{eq:C}
\end{align}
where the joint distribution of $(A,S_{e},S_{d},X,Y)$ is of the form
\begin{align}
&P_{A}(a)P_{S_{e},S_{d}|A}(s_{e},s_{d}|a)P_{X|A,S_{e}}(x|a,s_{e})P_{Y|X,S_{e},S_{d}}(y|x,s_{e},s_{d}) \nonumber
\end{align}
and the maximization is over all $P_{A}$ and $P_{X|A,S_{e}}$
such that
\begin{align}
    0 \leq I(X;Y,S_{d}|A) - I(X;S_{e}|A).
\end{align}
\end{theorem}

\begin{proof}
  We prove achievability by showing that any rate $R < C$ is
  achievable, i.e., for any $\delta>0$, there exists for all
  sufficiently large $n$ an $(|\mathcal{M}^{(n)}|, n)$ code with
  $\frac{1}{n}\log|\mathcal{M}^{(n)}| \geq R-\delta$, and average
  probabilities of error $P^{(n)}_{m,e} \leq \delta$, and
  $P^{(n)}_{x,e} \leq \delta$.  The proof of achievability uses random
  coding and joint typicality decoding. Conversely, we show that given
  any sequence of $(|\mathcal{M}^{(n)}|, n)$ codes with
  $\frac{1}{n}\log|\mathcal{M}^{(n)}| \geq R-\delta_{n}$,
  $P^{(n)}_{m,e} \leq \delta_{n}$, and $P^{(n)}_{x,e} \leq
  \delta_{n}$, then $R \leq C$.  The proof of the converse uses Fano's
  inequality and properties of the entropy function.

  The achievability proof follows arguments in \cite{Weissman2010}
  with a modification in which we use the channel input codeword
  $x^{n}$ directly instead of the auxiliary codeword. In the
  following, we give a sketch of the achievability proof. An action
  codebook $\{a^{n}\}$ of size $2^{n(I(A;Y,S_d)-\delta_{\epsilon})}$
  is generated i.i.d. $\sim P_{A}$. For each $a^{n}$, another
  codebook $\{x^{n}\}$ of size $2^{n(I(X;Y,S_d|A)-\delta_{\epsilon})}$
  is generated i.i.d. $\sim P_{X|A}$. Then the codewords are
  distributed uniformly into
  $2^{n(I(X;Y,S_d|A)-I(X;S_e|A)-2\delta_{\epsilon})}$ equal-sized bins
  (see Fig. \ref{fig:binningCC}). Given the message $m=(m_1,m_2)$, the
  action codeword $a^{n}(m_1)$ is selected. Then the channel states
  $(s_{e}^{n},s_{d}^{n})$ are generated as an output of the memoryless
  channel with transition probability
  $P_{S_{e}^{n},S_{d}^{n}|A^{n}}(s_{e}^{n},s_{d}^{n}|a^{n})=\prod_{i=1}^{n}P_{S_{e},S_{d}|A}(s_{e,i},s_{d,i}|a_{i})$.
  The encoder looks for $x^{n}$ that corresponds to $m_1$ and is in
  the bin $m_2$ such that it is jointly typical with the selected
  $a^{n}$ and $s_{e}^{n}$. For sufficiently large $n$, with
  arbitrarily high probability, there exists such a codeword because
  there are approximately $2^{n(I(X;S_e|A)+\delta_{\epsilon})}$
  codewords in the bin. Then the selected $x^{n}$ is transmitted over
  the channel $P_{Y|X,S_{e},S_{d}}$. Given $y^{n}$ and $s_{d}^{n}$,
  the decoder in the first step looks for codeword $a^{n}$ that is
  jointly typical with $y^{n}$ and $s_{d}^{n}$. With high probability,
  it will find one and it is the one chosen by the encoder since the
  codebook size is $2^{n(I(A;Y,S_d)-\delta_{\epsilon})}$. Then, given
  the correctly decoded $m_1$, the decoder in the second step looks
  for $x^{n}$ that is jointly typical with $y^{n},s_{d}^{n}$, and
  $a^{n}$. Again, with high probability, it will find one and it is
  the one chosen by the encoder since the size of the codebook is
  $2^{n(I(X;Y,S_d|A)-\delta_{\epsilon})}$. The corresponding bin index
  is then decoded as $\hat{m}_{2}$. In total,
  $I(A;Y,S_d)+I(X;Y,S_d|A)-I(X;S_e|A)-3\delta_{\epsilon}$ bits per
  channel use can be used to transmit the message $m$ such that both
  $m$ and $x^{n}$ are decoded correctly at the decoder.  Note that the
  above coding scheme which splits the message into two parts and
  decodes them sequentially works successfully when we have a proper
  positive number of bins for codewords $x^{n}$, i.e.,
  $I(X;Y,S_d|A)-I(X;S_e|A)-2\delta_{\epsilon} > 0$. The more detailed
  achievability proof and converse proof are given in Appendix
  \ref{sec:proofchannelcoding}.
\end{proof}

    \begin{figure}[]
    %\vspace{1.3cm}
    \centering
    \psfrag{e}[][][0.75]{$2^{n(I(X;S_{e}|A)+\delta_{\epsilon})}$}
    \psfrag{dot}[][][0.75]{$\cdots$}
    \psfrag{x}[][][0.75]{a codeword $X^{n}$}
    \psfrag{p1}[][][0.75]{$2^{n(I(X;Y,S_{d}|A)-I(X;S_{e}|A)-2\delta_{\epsilon})}$ bins}
    \includegraphics[width=9cm]{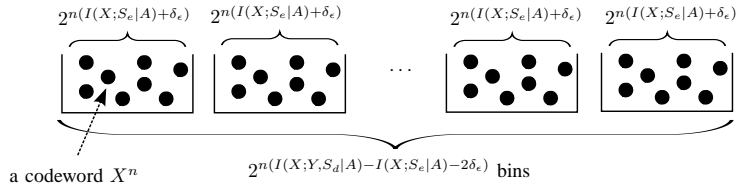}
    %\small Binning for the achievability: for each codeword $a^{n}$, a codebook $\{x^{n}\}$ of size $2^{n(I(X;Y|A)- \epsilon)}$ is generated \\ i.i.d. each $\sim P_{X|A}$. Then they are distributed uniformly into $2^{n(I(X;Y|A)-I(X;S_{e}|A)-2\epsilon)}$ equal-sized bins.
    \caption{Binning for the achievability: for each codeword $a^{n}$, a codebook $\{x^{n}\}$ of size $2^{n(I(X;Y,S_{d}|A)- \delta_{\epsilon})}$ is generated  i.i.d. each $\sim P_{X|A}$. Then they are distributed uniformly into $2^{n(I(X;Y,S_{d}|A)-I(X;S_{e}|A)-2\delta_{\epsilon})}$ equal-sized bins.} \label{fig:binningCC}
    \end{figure}

    We term the condition $I(X;Y,S_{d}|A)-I(X;S_{e}|A) \geq 0$ which
    appears in Theorem \ref{theoremCapa} the \emph{two-stage coding
    condition} since it represents the underlining sufficient condition
    for successful two-stage coding. It plays a role in restricting
    the set of input distributions in the capacity expression. It is
    also natural to wonder whether the two-stage coding condition can
    really be active or is always inactive when computing the
    capacity. In Example 1, Subsection C, we show by example that
    there exists a case where the condition is active. In the
    following results we also show that if the condition is active,
    then it is satisfied with equality, i.e., the capacity is obtained
    with $I(X;Y,S_{d}|A)-I(X;S_{e}|A) = 0$.  More details on the
    two-stage coding condition and its connection to other related
    problems will be given in Section
    \ref{sec:discussion}.

\begin{remark}\label{remark:generalCh}
It is possible to consider an action symbol as another input to the memoryless channel $P_{Y|X,S_{e},S_{d}}$. The capacity expression for this more general channel $P_{Y|X,S_{e},S_{d},A}$ remains unchanged. This can be shown by defining the new state $S'_{e}\triangleq (S_{e},A)$ and then applying the characterization in Theorem \ref{theoremCapa}.
\end{remark}

\begin{proposition}\label{prop:activewith0}
  If the two-stage coding condition is ignored, and the solution to
  the unconstrained problem in \eqref{eq:C} results in $I(X;Y,S_{d}|A)
  - I(X;S_{e}|A) < 0$ (the two-stage coding condition is active), then
  the actual channel capacity will be obtained with $I(X;Y,S_{d}|A) -
  I(X;S_{e}|A) = 0$.
\end{proposition}
\begin{proof}
  We consider a set $\mathcal{R}_{\text{mod}}$ containing pairs of rate $R$
  and dummy variable $\tilde{R} \in \mathbb{R}$ introduced for the
  two-stage coding condition, i.e.,
\begin{align*}
  \mathcal{R}_{\text{mod}} = \{ (R,\tilde{R}): 0 \leq &R \leq I(A,X;Y,S_{d})-I(X;S_{e}|A)  \triangleq I(A;Y,S_{d})+ \Delta I\\
  &\tilde{R} < I(X;Y,S_{d}|A)-I(X;S_{e}|A) \triangleq \Delta I\\
  \mbox{for some}\ \
  P_{A}(a)&P_{S_{e},S_{d}|A}(s_{e},s_{d}|a)P_{X|A,S_{e}}(x|a,s_{e})P_{Y|X,S_{e},S_{d}}(y|x,s_{e},s_{d})
  \}
\end{align*}
For each $P_{A} \in \mathcal{P}_{A}, P_{X|A,S_{e}} \in
\mathcal{P}_{X|A,S_{e}}$, we can compute a tuple $(I(A;Y,S_{d})+
\Delta I, \Delta I)$, and obtain the corresponding region as shown in
Fig. \ref{fig:region1}. We can show that the region
$\mathcal{R}_{\text{mod}}$ is convex (see Appendix
\ref{sec:proof_convex_dummy}). Then, to evaluate the region
$\mathcal{R}_{\text{mod}}$, we find the union of all regions obtained
from all possible $P_{A} \in \mathcal{P}_{A}, P_{X|A,S_{e}} \in
\mathcal{P}_{X|A,S_{e}}$. Our main task is to compute the channel
capacity so we are interested in finding the maximum rate $R$ under
the feasible value of $\Delta I$, i.e., $\Delta I \geq 0$. Since
$\mathcal{R}_{\text{mod}}$ is convex, one can show that there are only
two possible shapes of the region $\mathcal{R}_{\text{mod}}$, i.e.,
the ones where the maximum of $R$ is obtained with non-negative and
negative $\Delta I$, respectively. This is depicted in
Fig.~\ref{fig:region3}. The case $(b)$ in Fig.~\ref{fig:region3},
which is the case where the two-stage coding condition is active, is
of interest here. Since the feasible solutions have to satisfy $\Delta
I \geq 0$, we can conclude that when the two-stage coding condition is
active, the channel capacity will be obtained with $\Delta I = 0$.

    \begin{figure}[]
    %\vspace{1.3cm}
    \centering
    \psfrag{R}[][][0.75]{$R$}
    \psfrag{R'}[][][0.75]{$\tilde{R}$}
    \psfrag{d}[][][0.75]{$\Delta I$}
    \psfrag{I+d}[][][0.75]{$I(A;Y,S_{d})+ \Delta I$}
    \psfrag{0}[][][0.75]{$0$}
    \psfrag{a}[][][0.75]{$(a)$}
    \psfrag{b}[][][0.75]{$(b)$}
    \includegraphics[width=10cm]{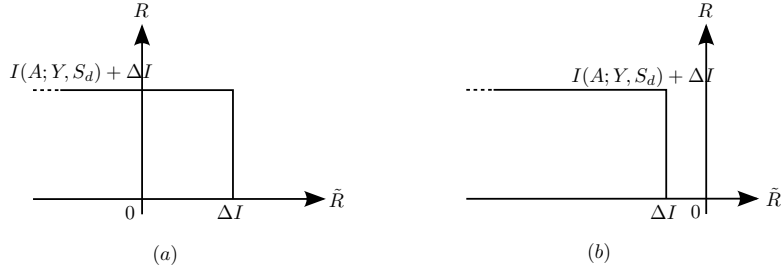}
    \caption{(a) the corresponding region with $\Delta I \geq 0$, (b) the corresponding region with $\Delta I < 0$.} \label{fig:region1}
    \end{figure}
\end{proof}

    \begin{figure}[]
    %\vspace{1.3cm}
    \centering
    \psfrag{R}[][][0.75]{$R$}
    \psfrag{R'}[][][0.75]{$\tilde{R}$}
    \psfrag{R=R'}[][][0.75]{$R=\tilde{R}$}
    \psfrag{0}[][][0.75]{$0$}
    \psfrag{a}[][][0.75]{$(a)$}
    \psfrag{b}[][][0.75]{$(b)$}
    \includegraphics[width=10cm]{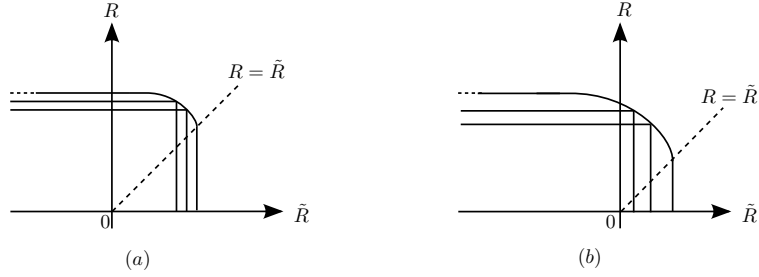}
    \caption{(a) the region $\mathcal{R}_{\text{mod}}$ where the maximum $R$ achieved with $\Delta I \geq 0$, (b) the region $\mathcal{R}_{\text{mod}}$ where the maximum $R$ achieved with $\Delta I < 0$.} \label{fig:region3}
    \end{figure}

\subsection{Other Results}
In the following, we provide some conclusions which help develop our
understanding and also relate our main result to other known results
in the literature. We consider the case where the reversible input
constraint is omitted and then our setting recovers Weissman's channel
with action-dependent state \cite{Weissman2010}. On the other hand, if
the channel state sequences are given by nature, i.i.d.~according to
some distribution, then our setting simply recovers the special case of
information embedding with reversible stegotext \cite{Sumszyk2009}. We
also consider the special case where channel state information at the
encoder or the decoder is absent. The result in this case can be
derived straightforwardly by setting the channel state variable to a
constant value. Lastly, it is also natural to consider the case where
the decoder is interested in decoding the message and the encoder's
state information instead. By this, the channel input sequence can be
retrieved based on the decoded message, the encoder's state
information, and a known deterministic encoding function. We show that
if the objective is to decode only the message and the channel input,
then decoding the message and encoder's state information first, and
then re-encoding the channel input is suboptimal.

\begin{proposition}
  When the reversible input constraint is omitted, the capacity of
  the channel with action-dependent state available noncausally to the
  encoder and the decoder is given by
\begin{align}\label{eq:lowerbound2}
     C_{M} &= \max [I(A,U;Y,S_{d})-I(U;S_{e}|A)],
\end{align}
where the joint distribution of $(A,S_{e},S_{d},U,X,Y)$ is of the form
\begin{align}
&P_{A}(a)P_{S_{e},S_{d}|A}(s_{e},s_{d}|a)P_{U|A,S_{e}}(u|a,s_{e})1_{\{X=\tilde{f}(U,S_{e})\}}P_{Y|X,S_{e},S_{d}}(y|x,s_{e},s_{d}). \nonumber %1_{\{X=\tilde{f}(m,S_{e})\}}
\end{align}
and the maximization is over $P_{A},P_{U|A,S_{e}}$ and $\tilde{f}: \mathcal{U} \times \mathcal{S}_{e} \rightarrow \mathcal{X}$, and $U$ is the auxiliary random variable with $|\mathcal{U}| \leq |\mathcal{A}||\mathcal{S}_{e}||\mathcal{X}|+1$.
\end{proposition}
\begin{proof}
  The proof follows from arguments in \cite{Weissman2010} with
  modifications such that the state $S^{n}=(S_{e}^{n},S_{d}^{n})$, and
  $(Y^{n},S_{d}^{n})$ are considered as the new channel output, and a
  set of distributions is restricted to satisfy the Markov relations
  $U-(A,S_{e})-S_{d}$ and $X-(U,S_{e})-(A,S_{d})$.
\end{proof}
\begin{corollary}\label{coroll:example}
For the special case where the state information at the decoder is absent, the capacity of the channel is given as a special case of Theorem \ref{theoremCapa} by setting $S_{d}$ to a constant value.
\end{corollary}

\begin{corollary}
  For the case where the state information at the decoder is absent
  and the channel state is given by nature, i.e., the action alphabet
  size is one, the capacity of the channel is obtained as
\begin{equation}\label{eq:C2}
    C_{\text{stegotext}} = \max [I(X;Y)-I(X;S_{e})],
\end{equation}
where the joint distribution of $(S_{e},X,Y)$ is of the form
\begin{align}
&P_{S_{e}}(s_{e})P_{X|S_{e}}(x|s_{e})P_{Y|X,S_{e}}(y|x,s_{e}) \nonumber 
\end{align}
and the maximization is over all $P_{X|S_{e}}$.  Note that this
recovers a special case of the results on information embedding with
reversible stegotext \cite{Sumszyk2009} when there is no distortion
constraint between $X^{n}$ and $S_{e}^{n}$.
\end{corollary}

Next we are looking at a related problem which later helps us interpret the two-stage coding condition. We consider a new and slightly different communication problem where the decoder is interested in decoding instead the message $M$ and the state $S_{e}^{n}$. Due to a deterministic encoding function, the channel input signal can be retrieved based on the decoded message and the encoder's state information. This communication problem has a more demanding reconstruction constraint than our main problem considered in Fig. \ref{fig:SystemModel4} since it essentially requires that the decoder can decode the message, the encoder's state, and the channel input signal, all reliably.

\begin{proposition}\label{eq:C_Se}
  Consider a new communication problem which is slightly different
  than the one considered in Fig. \ref{fig:SystemModel4} in that the
  decoder is interested in decoding the message $M$ and the state
  $S_{e}^{n}$ reliably. The capacity of such a channel is given by
\begin{align}\label{eq:cap_Se}
     C_{S_{e}} &= \max [I(A,S_{e},X;Y,S_{d})-H(S_{e}|A)],
\end{align}
where the joint distribution of $(A,S_{e},S_{d},X,Y)$ is of the form
\begin{align}
&P_{A}(a)P_{S_{e},S_{d}|A}(s_{e},s_{d}|a)P_{X|A,S_{e}}(x|a,s_{e})P_{Y|X,S_{e},S_{d}}(y|x,s_{e},s_{d}) \nonumber
\end{align}
and the maximization is over all $P_{A}$ and $P_{X|A,S_{e}}$
such that
\begin{align}
    0 \leq I(S_{e},X;Y,S_{d}|A) - H(S_{e}|A).
\end{align}
\end{proposition}
\begin{proof}
  Since decoding $M$ and $S_{e}^{n}$ implies that $X^{n}$ is also
  decoded from the deterministic encoding function, one can substitute
  $(S_{e},X)$ in place of $X$ in Theorem \ref{theoremCapa} and obtain
  the capacity. More specifically, the achievable scheme in this case
  is different from the previous case of decoding $M$ and $X^{n}$ in
  that the SI codebook is introduced and it has to ``cover'' all
  possible generated $S_{e}^{n}$ \emph{losslessly}. That is, the size
  of the SI codebook should be sufficiently large so that the encoder
  is able to find an exact $S_{e}^{n}$ from the codebook. Similarly to
  Theorem \ref{theoremCapa}, in the capacity expression, we also have
  a similar restricting condition $0 \leq
  I(S_{e},X;Y,S_{d}|A)-H(S_{e}|A)$ on the set of input
  distributions. Besides the rate constraint, this condition can be
  considered as a necessary and sufficient condition for the process
  of losslessly compressing $S_{e}^{n}$ through $X^{n}$ and
  then transmit them reliably over the channel in our two-stage
  communication problem. The detailed achievability proof and the
  converse proof are given in Appendix~\ref{sec:proof_C_se}.
\end{proof}

\begin{remark}\label{remark:C_geq_Cse}
  We know that the channel input sequence can be retrieved based on
  the decoded message, the encoder's state information, and a known
  deterministic encoding function. Therefore, it is natural to compare
  the capacity $C$ in Theorem \ref{theoremCapa} with $C_{S_{e}}$ in
  Proposition \ref{eq:C_Se}. For a given channel $P_{S_{e},S_{d}|A},
  P_{Y|X,S_{e},S_{d}}$, we have that $C \geq C_{S_{e}}$.
\end{remark}
\begin{proof}
  One can show that $I(S_{e},X;Y,S_{d}|A)-H(S_{e}|A) \leq
  I(X;Y,S_{d}|A) - I(X;S_{e}|A)$ for all joint distributions factorized
  in the form of
  $P_{A}(a)P_{S_{e},S_{d}|A}(s_{e},s_{d}|a)P_{X|A,S_{e}}(x|a,s_{e})P_{Y|X,S_{e},S_{d}}(y|x,s_{e},s_{d})$. This
  implies that $C_{S_{e}}$ is evaluated over a smaller set than that
  of $C$. In addition, one can show in a similar fashion that
  $I(A,S_{e},X;Y,S_{d})-H(S_{e}|A) \leq I(A,X;Y,S_{d}) -
  I(X;S_{e}|A)$, and thus conclude that $C \geq C_{S_{e}}$.
\end{proof}

We note that this new communication problem is closely related to the
problems of state amplification \cite{Kim2008}, and reversible
information embedding \cite{Willems}. The main difference is that, in
our setting, channel states are generated based on the action
sequence. In \cite{Kim2008} the decoder is interested in decoding the
message reliably and in decoding the encoder's state information
within a list, while in \cite{Willems}, the decoder is interested in
decoding both the message and the encoder's state information reliably. The
result in Remark \ref{remark:C_geq_Cse} is also analogous to that in
information embedding with reversible stegotext \cite{Sumszyk2009} in
which the authors show that if the objective is to decode only $M$ and
$X^{n}$, then decoding $M$ and $S_{e}^{n}$ first and re-encoding
$X^{n}$ using a deterministic encoding function is suboptimal.

\subsection{Examples}
In the following, we show two examples to illustrate the role of the
two-stage coding condition in restricting a set of input distributions
in the capacity expression. Example 1 shows that the two-stage coding
condition can be active in computing the capacity, while Example 2
shows that there also exists a case where such a condition is not
active at the optimal design.

\emph{Example 1: Memory Cell With a Rewrite Option}\\
For simplicity, let us consider a special case where $S_{d}^{n}$ is
absent as in Corollary \ref{coroll:example} and the channel is in the
more general form $P_{Y|X,S_{e},A}$ as in Remark
\ref{remark:generalCh}. We consider a binary example where
$A,X,S_{e},Y \in \{0,1\}$, and the scenario of writing on a memory
cell with a rewrite option. The first writing is done through a binary
symmetric channel with crossover probability $\delta$ (BSC($\delta$)),
input $A$, and output $S_{e}$. Then, assuming that there is a perfect
feedback of the output $S_{e}$ to the second encoder, the second
encoder has an option to rewrite on the memory or not to rewrite
(indicated by a value of $X$). If the rewrite value $X=1$ which
corresponds to ``rewrite,'' then $Y$ is given as the output of
BSC($\delta$) with input $A$ (rewrite using the old input). If $X=0$
which corresponds to ``no rewrite,'' we simply get $Y=S_{e}$. In this
case the decoder is interested in decoding both the embedded message
and the rewrite signal. See Fig.~\ref{fig:rewrite} for an illustration
of this rewrite channel.

\begin{figure}[]
    \centering
    \psfrag{m}[][][0.7]{$M$}
    \psfrag{act}[][][0.7]{Action}
    \psfrag{enc1}[][][0.8]{Encoder 1}
    \psfrag{enc2}[][][0.8]{Encoder 2}
    \psfrag{Am}[][][0.7]{$A^{n}(M)$}
    \psfrag{Ps1s2}[][][0.75]{$P_{S_{e}|A}$}
    \psfrag{s1}[][][0.7]{$S_{e}^{n}$}
    \psfrag{ch}[][][0.8]{Channel}
    \psfrag{Pyxs}[][][0.75]{$P_{Y|X,S_{e},A}$}
    \psfrag{x}[][][0.7]{$X^{n}$}
    \psfrag{y}[][][0.7]{$Y^{n}$}
    \psfrag{a}[][][0.5]{$A$}
    \psfrag{se}[][][0.5]{$S_{e}$}
    \psfrag{y1}[][][0.5]{$Y$}
    \psfrag{p}[][][0.5]{$\delta$}
    \psfrag{1-p}[][][0.5]{$1-\delta$}
    \psfrag{x=0}[][][0.5]{$X=0$}
    \psfrag{x=1}[][][0.5]{$X=1$}
    \psfrag{0}[][][0.5]{$0$}
    \psfrag{1}[][][0.5]{$1$}
    \psfrag{dec}[][][0.8]{Decoder}
    \psfrag{mhat,}[][][0.7]{$\hat{M},$}
    \psfrag{xhat}[][][0.7]{$\hat{X}^{n}$}
    \includegraphics[width=11.5cm]{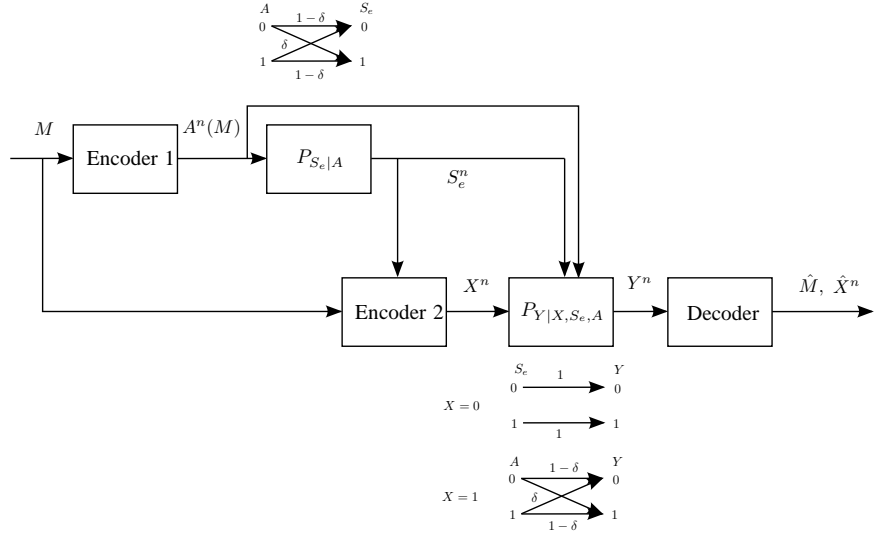}
    %\resizebox{12 cm}{!}{\epsfbox{pic.eps}}
    %\resizebox{3.0in}{!}{\includegraphics{pic.eps}}
    %\vspace{-.2cm}
     \caption{Two-stage writing on a memory cell with a rewrite option.}\label{fig:rewrite}
 %  \vspace{-.3cm}
    \centering
\end{figure}

From Theorem \ref{theoremCapa} and Remark \ref{remark:generalCh}, we know that the capacity of this channel is given by
\begin{align}
     C = \max[ I(A,X;Y) - I(X;S_{e}|A)],
\end{align}
where the joint distribution of $(A,S_{e},X,Y)$ is of the form
\begin{align}
&P_{A}(a)P_{S_{e}|A}(s_{e}|a)P_{X|A,S_{e}}(x|a,s_{e})P_{Y|X,S_{e},A}(y|x,s_{e},a) \nonumber
\end{align}
and the maximization is over all $P_{A}$ and $P_{X|A,S_{e}}$
such that
\begin{align}
    0 \leq I(X;Y|A) - I(X;S_{e}|A).
\end{align}

Letting $A\sim \mbox{Bernoulli}(p_{a})$, and
\begin{align}
    &p(x=0|s_{e}=0,a=0)=p\\
    &p(x=0|s_{e}=0,a=1)=q\\
    &p(x=0|s_{e}=1,a=0)=r\\
    &p(x=0|s_{e}=1,a=1)=s.
\end{align}
By straightforward manipulation, we get
\begin{align*}
&H(Y) = h\big((1-\delta)(1-p_{a})(1-\delta+\delta p) + \delta p_{a}(q+\delta- \delta q) + \delta(1-\delta)(1-r+rp_{a}-sp_{a})\big), \\
&H(Y|A) = (1-p_{a})h((1-\delta)(p+(1-p)(1-\delta)+(1-r)\delta)) + p_{a}h(\delta (q+(1-q)\delta +(1-s)(1-\delta))),\\
&-H(Y|A,X)-I(X;S_{e}|A) \\
& \qquad = [(1-p)(1-\delta)(1-p_{a})+ (1-r)\delta (1-p_{a})]\cdot [h(\frac{(1-p)(1-\delta)(1-p_{a})}{(1-p)(1-\delta)(1-p_{a})+ (1-r)\delta (1-p_{a})})-h(\delta)]\\
          &\qquad \qquad  + [(1-q)\delta p_{a}+ (1-s)(1-\delta)p_{a}]\cdot[h(\frac{(1-q)\delta p_{a}}{(1-q)\delta p_{a}+ (1-s)(1-\delta)p_{a}})-h(\delta)]  - h(\delta),
\end{align*}
then
\begin{align*}
     C = &\max_{p_{a},p,q,r,s \in [0,1]}\Big[H(Y)-H(Y|A,X)-I(X;S_{e}|A) \Big]
\end{align*}
subject to
\begin{align*}
     0 & \leq H(Y|A)-H(Y|A,X)-I(X;S_{e}|A).
\end{align*}

By performing numerical optimization with $\delta=0.1$, we obtain that
the capacity of the channel equals to $0.5310$ bits per channel
use. The optimal (capacity achieving) input distributions in this case
are those in which $X-A-S_{e}$ forms a Markov chain, i.e., $p=r,q=s$,
and in the end $P_{a}$ is the only remaining optimization variable. We
note that if we instead neglect the restriction on the maximization
domain and solve the unconstrained optimization problem, we would
obtain the maximum value of $0.6690$ which is strictly larger than the
actual capacity. Therefore, this example shows that there exists a
case where the two-stage coding condition is active. In fact, the
corresponding two-stage coding condition in this case is satisfied
with equality as expected from Proposition \ref{prop:activewith0}.

\emph{Example 2: Inactive Two-stage Coding Condition}\\
In other cases the two-stage coding condition in the capacity
expression might not be active. One trivial example is when
$S_{e}-A-S_{d}$ forms a Markov chain for the action-dependent state
channel $P_{S_{e},S_{d}|A}$, and $Y-(X,S_{d})-S_{e}$ forms a Markov
chain for the state-dependent channel $P_{Y|X,S_{e},S_{d}}$. In this
case, it can be shown that for any joint distribution
$\big(P_{A}^{(1)}(a),P_{X|A,S_{e}}^{(1)}(x|a,s_{e})\big)$ such that
$I(X;Y,S_{d}|A)-I(X;S_{e}|A) < 0$, there always exists another joint
distribution $\big(P_{A}^{(2)}(a),P_{X|A,S_{e}}^{(2)}(x|a,s_{e})\big)$
which satisfies $I(X;Y,S_{d}|A)-I(X;S_{e}|A)\geq 0$ and achieves a
higher rate. One possible choice is to let
$P_{A}^{(2)}(a)=P_{A}^{(1)}(a)$ and $P_{X|A,S_{e}}^{(2)}(x|a,s_{e}) =
\sum_{s_{e}}P_{S_{e}|A}(s_{e}|a)P_{X|A,S_{e}}^{(1)}(x|a,s_{e})$. Consequently,
the maximizing input distribution in this case will result in
$I(X;Y,S_{d}|A)-I(X;S_{e}|A)\geq 0$ and the capacity of such a channel
is given by $C= \max_{P_{A},P_{X|A,S_{e}}}[I(A,X;Y,S_{d}) -
I(X;S_{e}|A)]$.

%%%%%%%%%%%%%%%%%%%%%%%%%%%%%%%%%%%%%%%%%%%%%%%%%%%%%%%%%%%%%%%%%%%%%%%%%%%%%%%%%%%%%%%%%%%%%%%%%%%%%%%%%%%%%%%%%%%%%%%%%%
\section{Discussion on the Two-stage Coding Condition and Formula Duality}\label{sec:discussion}
In this section we discuss in more detail the presence and impact of
the \emph{two-stage coding condition} in Theorem~\ref{theoremCapa},
and we also consider the potential dual relations between the source
coding and channel coding problems in Section \ref{sec:sourcecoding1}
and \ref{sec:channelcoding2}.

\subsection{Two-stage Coding Condition}
\subsubsection{Operational Coding View}
As can be seen in our achievable scheme, the condition
$I(X;Y,S_{d}|A)
-I(X;S_{e}|A)>0$ represents a tradeoff in the size of the
codebook $\{x^{n}\}$ conditioned on the action sequence. Our coding
scheme involves random binning, and in order to encode/decode
successfully we need to ensure that there is a proper positive number
of bins to satisfy both encoding and decoding requirements, based on
joint typicality. More specifically, the decoder is interested in
decoding both the message (partly carried in the action codeword and
partly as a bin index of $x^n$) and the codeword $x^n$ itself. From
the analysis of the error probability (see
Appendix~\ref{sec:proofchannelcoding}), this additional restriction on
the number of bins ($I(X;Y,S_{d}|A)-I(X;S_{e}|A) > 0$) arises in part
from the error event where only the message that is conveyed in the
action codeword is decoded correctly, but not the codeword
$x^{n}$. Since the action sequence carries information about the same
message that is carried by the codeword $x^{n}$, this additional
constraint is needed to ensure a vanishing probability of such an
error event (see also \eqref{eq:M2} that the two-stage coding
condition is the underlying constraint on the number of bins of
codewords $x^{n}$). Conversely, we also see that for any achievable
rate, it is never possible to have a joint distribution that leads to
$I(X;Y,S_{d}|A)-I(X;S_{e}|A) < 0$.

The condition can also be interpreted based on the structure of the
encoder, which involves two-stage coding (the action sequence is
selected first, then the channel input is selected based on the
action-dependent state). That is, the action sequence can be decoded
in the first stage, which in turn results in an extra constraint for
decoding the channel input in the second stage. Hence the condition
describes a causality constraint imposed by the communication
problem. This observation might be interesting for some other problems
as well.

\subsubsection{Source Coding View}
We notice that the condition $I(X;Y,S_{d}|A) - I(X;S_{e}|A) > 0$ can
be equivalently written as $H(X|Y,S_{d},A)<
H(X|S_{e},A)$. Intuitively, this tells us that for reliable
transmission of the channel input signal over the channel given that
the action is communicated, the uncertainty about $X$ that remains
after observing $Y$ and $S_{d}$ at the decoder should be less than the
uncertainty of $X$ at the transmitter. Hence the two-stage coding
condition can, as a complement to the rate constraint, be considered
as a necessary and sufficient condition for reliable transmission of
the description $X^{n}$ of the state $S_{e}^{n}$ through the channel
in our two-stage communication problem.

Alternatively, we note that in our case we do not need to reconstruct
$S^{n}_{e}$ perfectly at the decoder, i.e., information about
$S^{n}_{e}$ conveyed through $X^{n}$ over the channel is needed only
in part. We can write the condition as $I(X;S_{e}|A) < I(X;Y,S_{d}|A)$
and interpret it as a condition for \emph{lossy} transmission of
$S^{n}_{e}$ through $X^{n}$ over the channel given that $A^{n}$ is
communicated. It is then natural to compare this to the case when we
are interested in decoding $M$ and $S_{e}^{n}$, e.g., as in
Proposition~\ref{eq:C_Se}. In that case, we want to reconstruct
$S^{n}_{e}$ perfectly at the decoder; therefore, given $A^{n}$, the
necessary and sufficient condition for \emph{lossless} transmission of
$S^{n}_{e}$ through $X^{n}$ over the channel $P_{Y|X,S_{e},S_{d}}$ is
given by $H(S_{e}|A) < I(X,S_{e};Y,S_{d}|A)$.

\subsubsection{Channel Coding View}
We may also consider the condition $I(X;Y,S_{d}|A)-I(X;S_{e}|A) \geq
0$ from the point of view of connecting it to a class of cooperative
``multiple-access channels (MACs)'' with common message. Consider
therefore a slightly modified setting shown in
Fig.~\ref{fig:modified}, where there is another independent message
$W$ to be encoded at the channel encoder, and the message $M$ is a
common message for both encoders. This setting will reduce to our
original problem when the rate of message $W$ is zero. From this point
of view, the condition $I(X;Y,S_{d}|A)-I(X;S_{e}|A) \geq 0$ is in fact
a degenerate rate constraint derived from the underlying rate
constraint of message $W$ in the ``MAC'' setting.
\begin{figure}[]
    \centering
    \psfrag{m}[][][0.7]{$M$}
    \psfrag{w}[][][0.7]{$W$}
    \psfrag{act}[][][0.75]{Action}
    \psfrag{enc}[][][0.75]{encoder}
    \psfrag{Am}[][][0.7]{$A^{n}(M)$}
    \psfrag{Ps1s2}[][][0.75]{$P_{S_{e},S_{d}|A}$}
    \psfrag{s1s2}[][][0.7]{$S^{n}=(S_{e}^{n},S_{d}^{n})$}
    \psfrag{s1}[][][0.7]{$S_{e}^{n}$}
    \psfrag{s2}[][][0.7]{$S_{d}^{n}$}
    \psfrag{ch}[][][0.75]{Channel}
    \psfrag{Pyxs}[][][0.75]{$P_{Y|X,S}$}
    \psfrag{x}[][][0.7]{$X^{n}$}
    \psfrag{y}[][][0.7]{$Y^{n}$}
    \psfrag{k1}[][][0.7]{$l_{e}$}
    \psfrag{k2}[][][0.7]{$l_{d}$}
    \psfrag{dec}[][][0.75]{Decoder}
    \psfrag{mhat, xhat}[][][0.7]{$\hat{M}, \hat{X}^{n}$}
    \includegraphics[width=11cm]{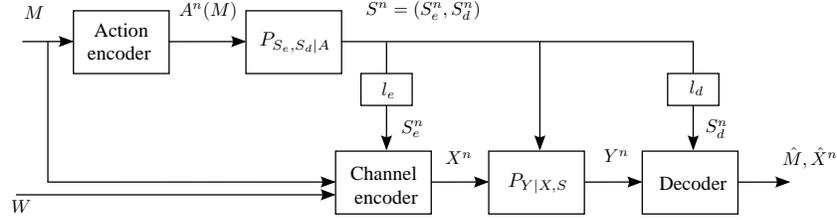}
    %\resizebox{12 cm}{!}{\epsfbox{pic.eps}}
    %\resizebox{3.0in}{!}{\includegraphics{pic.eps}}
    %\vspace{-.2cm}
    \caption[Caption for LOF]{Modified setting: a class of cooperative ``multiple-access channel (MAC)" with common message.\footnotemark[2]}\label{fig:modified}
 %  \vspace{-.3cm}
    \centering
\end{figure}

\footnotetext[2]{Based on this scenario, one can also recover special
  cases of results available for the multiple-access channel with
  common message. For example, if the encoder state information
  $S_{e}$ is assumed to be a deterministic function of $A$, then this
  modified setting will reduce to a class of MAC with common message
  and cribbing encoder, and eventually to a class of MAC with common
  message. To decode both messages and the channel input $X^n$ at the
  decoder is then equivalent to just decode messages $M$ and $W$.}

\subsection{Duality}
In this work we notice the ``dual'' relations between input-output of
elements in the source and channel coding systems as depicted in
Fig.~\ref{fig:Duality}. Similar dual relations also appear in other
related problems, as listed below.
 \begin{align*}
 \textnormal{Wyner-Ziv's source coding (SC)  (WZ,\cite{WynerZiv})} &\leftrightarrow   \textnormal{Gel'fand-Pinsker's  channel coding (CC)  (GP,\cite{GelfandPinsker})}  \\
  \textnormal{Permuter-Weissman's SC with action (PW,\cite{Permuter2011}) }  &\leftrightarrow \textnormal{Weissman's CC with action (W,\cite{Weissman2010})}  \\
   \textnormal{Steinberg's SC with CR (S,\cite{Steinberg2009})} &\leftrightarrow  \textnormal{Sumszyk-Steinberg's CC with RI (SS,\cite{Sumszyk2009})} \\
   \textnormal{Section \ref{sec:sourcecoding1}} &\overset{(\sharp)}{\leftrightarrow} \textnormal{Section \ref{sec:channelcoding2}}
\end{align*}

\begin{figure}[h]
    \centering
    \psfrag{m}[][][0.7]{$M$}
    \psfrag{act}[][][0.8]{Action}
    \psfrag{enc}[][][0.8]{Encoder}
    \psfrag{Am}[][][0.7]{$A^{n}(M)$}
    \psfrag{Aw}[][][0.7]{$A^{n}(W)$}
     \psfrag{t,w}[][][0.7]{$(T,W)$}
      \psfrag{w}[][][0.7]{$W$}
    \psfrag{se}[][][0.7]{$S_{e}^{n}$}
    \psfrag{sd}[][][0.7]{$S_{d}^{n}$}
    \psfrag{channel}[][][0.9]{Channel Coding}
    \psfrag{source}[][][0.9]{Source Coding}
    \psfrag{x}[][][0.7]{$X^{n}$}
    \psfrag{y}[][][0.7]{$Y^{n}$}
    \psfrag{dec}[][][0.8]{Decoder}
    \psfrag{mhat}[][][0.7]{$\hat{M}$}
    \psfrag{xhat}[][][0.7]{$\hat{X}^{n}$}
    \psfrag{xhathat}[][][0.7]{$\psi\big(X^{n},S_e^{n},A^{n}(W)\big)=\psi^{*}(X^{n},S_e^{n})=\hat{X}^{n}$}
    \psfrag{inputhat}[][][0.7]{$g_{x}(Y^{n},S_d^{n})=\hat{X}^{n}$}
    \includegraphics[width=8.5cm]{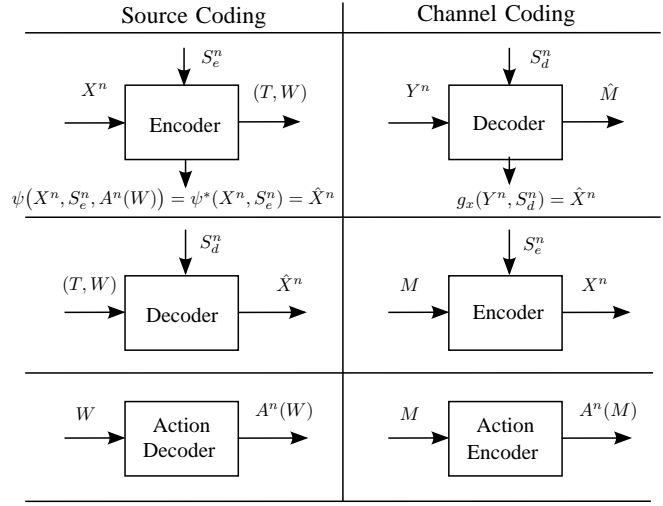}
    \caption{Duality between the source coding with action-dependent side information and common reconstruction  (Fig. \ref{fig:SystemModel1})  and channel coding with action-dependent states and reversible input (Fig. \ref{fig:SystemModel4}). } \label{fig:Duality}
 %  \vspace{-.3cm}
\end{figure}

As stated before in the introduction part, we are interested in investigating  \emph{formula duality} of a set of problems \cite{Cover2002}. Table \ref{table:1} below summarizes the rate-distortion(-cost) function and the channel capacity expressions of the interested problems, neglecting the optimization variables (input probability distribution).\\
\begin{table}[h!]
\caption{Rate-distortion(-cost) function and channel capacity.} \label{table:1}
\renewcommand{\arraystretch}{1.3}
\begin{tabular}{|c|c|c|}
  \hline
  Problems & Rate-distortion-cost function & Channel capacity \\
  \hline
  WZ and GP& $R_{\text{WZ}}(D)=\min [I(U;X) - I(U,S_{d})]$ & $C_{\text{GP}}= \max[I(U;Y)-I(U;S_{e})]$\\
  PW and W& $R_{\text{PW}}(D,C)=\min [I(A;X) + I(U;X|A) - I(U,S_{d}|A)]$ & $C_{\text{W}}= \max[I(A;Y) + I(U;Y|A)-I(U;S_{e}|A)]$ \\
  S and SS & $R_{\text{S}}(D)= \min[I(\hat{X};X)-I(\hat{X};S_{d})]$ & $C_{\text{SS}}= \max[I(X;Y)-I(X;S_{e})]$\\
  Sec. II and III & $R(D,C) = \min[I(A;X)+I(\hat{X};X,S_{e}|A)-I(\hat{X};S_{d}|A)]$ & $C = \max_{p*}[I(A;Y,S_{d})+I(X;Y,S_{d}|A)-I(X;S_{e}|A)]$\\
  \hline
\end{tabular}
\end{table}

As in \cite{Cover2002} we can recognize the formula duality of the
rate-distortion(-cost) function and the channel capacity by the
following correspondence,
\begin{align*}
\textnormal{Rate-distortion-cost} &\leftrightarrow  \textnormal{Channel capacity}\\
 %R(D,C) &\leftrightarrow  C\\
  \textnormal{minimization} &\leftrightarrow \textnormal{maximization} \\
  X (\textnormal{source symbol}) &\leftrightarrow Y (\textnormal{received symbol}) \\
  \hat{X} (\textnormal{decoded symbol}) &\leftrightarrow X (\textnormal{transmitted symbol}) \\
  S_{e} (\textnormal{state at the encoder})  &\leftrightarrow S_{d} (\textnormal{state at the decoder})\\
  S_{d} (\textnormal{state at the decoder})  &\leftrightarrow S_{e} (\textnormal{state at the encoder})\\
  U (\textnormal{auxiliary})&\leftrightarrow U (\textnormal{auxiliary}) \\
  A (\textnormal{action}) &\leftrightarrow A (\textnormal{action}).
\end{align*}

We see that the first three cases in Table~\ref{table:1} are obvious
from the expressions of the rate-distortion(-cost) function and the
channel capacity, while the last duality (Secs.~II and III) does not
hold in general due to the fundamental differences in the source and
channel coding problems. We now give reasons based on the dual roles
of the encoder/decoder in the source coding problem and the
decoder/encoder in the channel coding problem.

The first reason that the last duality does not hold in general is the
presence of the two-sided side/state information. That is, at the
\emph{encoder} in the source coding setup, the processing is
sequential, i.e., the action-dependent side information is generated
first and then the side information $S_{e}^{n}$ is used in compressing
the source sequence. However, this sequential processing is not
required in the \emph{decoding} process of the channel coding problem
since the state information for both encoder and decoder are generated
in the beginning, and both $Y^{n}$ and $S_{d}^{n}$ are available at
the decoder noncausally. The effect of this fundamental difference can
be seen from the difference in the terms $I(A,X)$ and $I(A;Y,S_{d})$
in the rate-distortion-cost function and channel capacity expressions
in Table~\ref{table:1}.

The second reason is the additional reconstruction constraint imposed on
the two communication problems. First consider the channel coding
problem where we require to decode as well the channel input sequence
(reversible input constraint). In our problem the \emph{encoder} has a
causal structure; that is, $S_{e}^{n}$ is generated first, then
followed by $X^{n}$. When we require to decode $X^{n}$ which is the
signal generated in the second stage, the two-stage coding condition,
apart from the rate constraint, is necessary to ensure reliable
transmission of the channel input $X^{n}$. In fact, it plays a role in
restricting the set of capacity achieving input distributions marked
by $p*$ in Table~\ref{table:1}. Now we consider the source coding
counterpart where we require the encoder to estimate the decoder's
reconstruction (common reconstruction constraint). Although there
seems to be a similar two-stage structure in the \emph{decoder}
of this setup, the two-stage coding condition is not relevant
here. This is because the common reconstruction is performed in the
beginning at the encoder side and the identity of action sequence is
in fact known at both sides due to the noiseless link between the
encoder and the decoder.

\section{Conclusion} \label{sec:conclusion} In this paper we studied a
class of problems that extend Wyner-Ziv source coding and
Gel'fand-Pinsker channel coding with action-dependent side
information. The extension involves having two-sided action-dependent
partial SI, and also enforcing additional reconstruction
constraints. In the source coding problem, we solved the
rate-distortion-cost function for the memoryless source with two-sided
action-dependent partial SI and common reconstruction, while in the
channel coding problem, the capacity of the discrete memoryless
channel with two-sided action-dependent state and reversible input is
derived under the two-stage coding condition. In fact, this two-stage
coding condition arises from the additional reconstruction constraint
and the causal structure of the setup, i.e., the channel input signal
to be reconstructed is generated in the second stage
transmission. Besides the message rate constraint, it can be
considered as a necessary and sufficient condition for reliable
transmission of channel input signal over the channel given that the
action is communicated. An intuitive interpretation derived from its
expression is that uncertainty about the channel input remaining at
the receiver after observing the channel output and the decoder's
state information should be less than that at the transmitter.

We were also interested in investigating the formula duality between
rate-distortion-cost function and channel capacity of the source and
channel coding problems. Although our extended problems seem to retain
the dual structure seen in Wyner-Ziv and Gel'fand-Pinsker problems,
they are not dual in general. In fact, there is ``operational
mismatch'' caused by enforcing causality in parts of the system.  For
example, the two-sided SI in the source coding problem requires a
sequential encoding process, while in the channel coding problem the
channel output and state information are available noncausally to the
decoder. Moreover, when we require additional reconstruction of the
channel input in the channel coding problem, the two-stage coding
condition is needed due to the causal structure of the encoder where
the channel encoder has to wait for the state to be generated based on
the action sequence.

We find it interesting to note that the two-stage coding condition
which appears in the capacity expression can be active, as shown in
one example. This is, however, not surprising since the condition can
also be seen as a degenerate rate constraint of the underlying rate
constraint in a cooperative MAC setup (see Section
\ref{sec:discussion}, part A).
%%%
We notice that by imposing an additional reconstruction constraint on
that related problem, we are still able to derive a closed form
solution. This leads us to believe that it might be possible to
consider other (possibly open) network information theory problems
with additional reconstruction constraints, and be able to derive the
closed-form solutions. In addition, if we obtain a similar two-stage
coding condition in the solution, we might be able to find a class of
channels of which the capacity can be achieved with the input
distribution that results in an inactive two-stage coding
condition. This can provide some insights into the role of the
additional reconstruction constraint in some communication channels,
and should be considered as a topic for future work.
%%%%%%%%%%%%%%%%%%%%%%%%%%%%%%%%%%%%%%%%%%%%%%%%%%%%%%%%%%%%%%%%%%%%%%%%%%%%%%%%%%%%%%%%%%%%%%%%%%%%%%%%%%%%%%%%%%%%%%%%%%
\appendices
\section{Proof of Lemma \ref{lemma:convexRDC}} \label{sec:prooflemma}
Since the domain size of minimization in \eqref{eq:rateDistor1} or \eqref{eq:alterR2} increases with $D$ and $C$, $R_{\text{ac,cr}}(D,C)$ is non-increasing in $D$ and $C$. For convexity, we consider two distinct points $(R_{i},D_{i},C_{i}),\ i=1,2$, which lie on the boundary of $R_{\text{ac,cr}}(D,C)$. Suppose $\big(P^{(i)}_{A|X},P^{(i)}_{\hat{X}|X,S_{e},A}\big)$, $i=1,2$, \ achieve these respective points, i.e.,
\begin{align*}
  R_{i} &= R_{\text{ac,cr}}(D_{i},C_{i}) = I^{(i)}(X;A)+I^{(i)}(\hat{X};X,S_{e}|A_{i},S_{d}), \ i=1,2,
\end{align*}
where $I^{(i)}(\cdot)$ denotes the mutual information associated with $P^{(i)}_{A|X}$ and $P^{(i)}_{\hat{X}|X,S_{e},A}$.

Let $Q \in \{1,2\}$ be a random variable independent of $X$ and conditionally independent of $(S_{e},S_{d})$ given $(X,A)$, with $P_{Q}(1)=1-P_{Q}(2)=\lambda$, $0 \leq \lambda \leq 1$. Then we have the joint distribution
\begin{align*}
  &P_{Q,X,A,S_{e},S_{d},\hat{X}}(q,x,a,s_{e},s_{d},\hat{x})= P_{Q}(q)P_{X}(x)P_{A|X,Q}(a|x,q)P_{S_{e},S_{d}|X,A}(s_{e},s_{d}|x,a)P_{\hat{X}|X,S_{e},A,Q}(\hat{x}|x,s_{e},a,q),
\end{align*}
where $P_{A|X,Q}(a|x,q) \triangleq P^{(q)}_{A|X}(a|x)$ and $P_{\hat{X}|X,S_{e},A,Q}(\hat{x}|x,s_{e},a,q) \triangleq P^{(q)}_{\hat{X}|X,S_{e},A}(\hat{x}|x,s_{e},a)$ for $q=1,2$.

Consider now the marginal distribution (averaged over $Q$)
\begin{align*}
  &P_{X,A,S_{e},S_{d},\hat{X}}(x,a,s_{e},s_{d},\hat{x})= \sum_{q=1,2}P_{Q}(q)P_{X}(x)P_{A|X,Q}(a|x,q)P_{S_{e},S_{d}|X,A}(s_{e},s_{d}|x,a)P_{\hat{X}|X,S_{e},A,Q}(\hat{x}|x,s_{e},a,q),
\end{align*}
which is associated with the sum of mutual information terms $I(X;A)+I(\hat{X};X,S_{e}|A,S_{d})$. It follows that
\begin{align*}
&I(X;A)+I(\hat{X};X,S_{e}|A,S_{d}) \\
&= I(X;A,\hat{X},S_{d})-I(X;S_{d}|A)+I(\hat{X};S_{e}|X,A,S_{d})\\
&= H(X)-H(X|A,\hat{X},S_{d})-H(S_{d}|A)+H(S_{d}|X,A)+H(S_{e}|X,A,S_{d})-H(S_{e}|X,A,S_{d},\hat{X})\\
&\overset{(*)}{=} H(X|Q)-H(X|A,\hat{X},S_{d})-H(S_{d}|A)+H(S_{d}|X,A,Q)+H(S_{e}|X,A,S_{d},Q)-H(S_{e}|X,A,S_{d},\hat{X})\\
& \leq H(X|Q)-H(X|A,\hat{X},S_{d},Q)-H(S_{d}|A,Q)+H(S_{d}|X,A,Q)+H(S_{e}|X,A,S_{d},Q)-H(S_{e}|X,A,S_{d},\hat{X},Q)\\
& = I(X;A,\hat{X},S_{d}|Q)-I(X;S_{d}|A,Q)+I(\hat{X};S_{e}|X,A,S_{d},Q)\\
&= I(X;A|Q)+I(\hat{X};X,S_{e}|A,S_{d},Q)\\
& = \lambda [I^{(1)}(X;A)+I^{(1)}(\hat{X};X,S_{e}|A,S_{d})] + (1-\lambda)[I^{(2)}(X;A)+I^{(2)}(\hat{X};X,S_{e}|A,S_{d})],
\end{align*}
where $(*)$ follows from $X\perp Q$ and the Markov chain $(S_{e},S_{d})-(X,A)-Q$.

Consider also the average distortion and cost (averaged over $Q$),
\begin{align*}
D &= E\big[d\big(X,\hat{X}\big)\big] = \lambda E^{(1)}\big[d\big(X,\hat{X}\big)\big] +(1-\lambda)E^{(2)}\big[d\big(X,\hat{X}\big)\big] =\lambda D_{1} + (1-\lambda)D_{2}\\
 \mbox{and} \quad C &= E[\Lambda(A)] =   \lambda E^{(1)}[\Lambda(A)] +(1-\lambda)E^{(2)}[\Lambda(A)] = \lambda C_{1} + (1-\lambda)C_{2}.
\end{align*}
Then, by the definition of the rate-distortion-cost function $R_{\text{ac,cr}}(D,C)$, it follows that
\begin{align*}
  &R_{\text{ac,cr}}\big(\lambda D_{1} + (1-\lambda)D_{2},\lambda C_{1} + (1-\lambda)C_{2}\big) = R_{\text{ac,cr}}(D,C)\\
  %&= \min [I(X;A)+I(\hat{X};X,S_{e}|A,S_{d})]\\
   &\leq  I(X;A)+I(\hat{X};X,S_{e}|A,S_{d})\\
& \leq \lambda [I^{(1)}(X;A)+I^{(1)}(\hat{X};X,S_{e}|A,S_{d})] + (1-\lambda)[I^{(2)}(X;A)+I^{(2)}(\hat{X};X,S_{e}|A,S_{d})]\\
   &=  \lambda R_{\text{ac,cr}}(D_{1},C_{1}) + (1-\lambda)R_{\text{ac,cr}}(D_{2},C_{2}).
\end{align*}
%where $I(X;Q)=0$ is due to the independence between $Q$ and $X$.%
Thus, we have shown that $R_{\text{ac,cr}}(D,C)$ is a non-increasing convex function of $D$ and $C$. \hfill$\blacksquare$

%%%%%%%%%%%%%%%%%%%%%%%%%%%%%%%%%%%%%%%%%%%%%%%%%%%%%%%%%%%%%%%%%%%%%%%%%%%%%%%%%%%%%%%%%%%%%%%%%%%%%%%%%%%%%%%%%%%%%%%%%%
\section{Proof of Theorem \ref{theoremRateDist}}\label{sec:proofsourcecoding}
\subsection{Achievability Proof of Theorem \ref{theoremRateDist}}

The proof follows from a standard random coding argument where we use the definitions and properties of $\epsilon$-typicality as in \cite{ElGamalKim}, i.e., the set of $\epsilon$-typical sequence for $\epsilon >0$ with respect to $P_{X}(\cdot)$ is denoted by
\begin{align}\label{eq:Typical}
    T_{\epsilon}^{(n)}(X)&=\Big\{ x^{n}\in \mathcal{X}^{n}: \left|\frac{1}{n}N(a|x^{n})-P_{X}(a)\right| \leq \epsilon P_{X}(a), \ \mbox{for all} \ a \in \mathcal{X} \Big\},
\end{align}
where $N(a|x^{n})$ is the number of occurrences of $a$ in the sequence $x^{n}$.

\textit{Codebook Generation}: Fix $P_{A|X},P_{\hat{X}|X,S_e,A}$. Let $\mathcal{W}_1^{(n)} =\{1,2,\ldots,|\mathcal{W}_1^{(n)}|\}$, $\mathcal{W}_2^{(n)} =\{1,2, \ldots,|\mathcal{W}_2^{(n)}|\}$, and $\mathcal{V}^{(n)} =\{1,2, \ldots,|\mathcal{V}^{(n)}|\}$. For all $w_1\in \mathcal{W}_1^{(n)}$ the action codewords $a^{n}(w_1)$ are generated i.i.d. each according to $\prod_{i=1}^{n}
    P_{A}(a_{i})$ and for each $w_1 \in \mathcal{W}_1^{(n)}$
    $|\mathcal{W}_2^{(n)}||\mathcal{V}^{(n)}|$ codewords $\{\hat{x}^{n}(w_1,w_2,v)\}_{w_2\in \mathcal{W}_2^{(n)},v\in \mathcal{V}^{(n)}}$ are generated i.i.d. each according to $\prod_{i=1}^{n}P_{\hat{X}|A}\big(\hat{x}_{i}|a_{i}(w_1)\big)$. The codebooks are then revealed to the encoder, the action decoder, and the decoder. Let $0< \epsilon_{0}<\epsilon_{1}<\epsilon < 1$.

    \textit{Encoding}: Given a source realization $x^{n}$ the encoder
    first looks for the smallest $w_1 \in \mathcal{W}_1^{(n)}$ such
    that $a^{n}(w_1)$ is jointly typical with $x^{n}$. Then the
    channel states are generated as outputs of the memoryless channel
    with transition probability
    $P_{S_{e}^{n},S_{d}^{n}|A^{n}}(s_{e}^{n},s_{d}^{n}|a^{n})=\prod_{i=1}^{n}P_{S_{e},S_{d}|A}(s_{e,i},s_{d,i}|a_{i})$,
    and the encoder in the second stage looks for the smallest $w_2
    \in \mathcal{W}_2^{(n)}$ and $v \in \mathcal{V}^{(n)}$ such that
    $\big(x^{n},\hat{x}^{n}(w_1,w_2,v),s_{e}^{n},a^{n}(w_1)\big)\in
    T_{\epsilon_{1}}^{(n)}(X,\hat{X},S_{e},A)$. If successful, the
    encoder produces $\hat{x}^{n}(w_1,w_2,v)$ as a common
    reconstruction at the encoder and transmits indices $(w_1,w_2)$ to
    the decoder. If not successful, the encoder transmits $w_1=1,
    w_2=1$ and produces $\hat{x}^{n}(1,1,1)$.

\textit{Decoding}: Given the indices $w_1$ and $w_2$, and the side information $s_{d}^{n}$ the decoder reconstructs $\tilde{x}^{n}=\hat{x}^{n}(w_1,w_2,\tilde{v})$ if there exists a unique $\tilde{v}\in \mathcal{V}^{(n)}$ such that $\big(s_{d}^{n},\hat{x}^{n}(w_1,w_2,\tilde{v}),a^{n}(w_1)\big)\in T_{\epsilon}^{(n)}(S_{d},\hat{X},A)$. Otherwise, the decoder puts out $\tilde{x}^{n}=\hat{x}^{n}(w_1,w_2,1)$.

\textit{Analysis of Probability of Error}: Let $(W_1,W_2,V)$ denote the corresponding indices of the chosen codewords $A^{n}$ and $\hat{X}^{n}$ at the encoder. We define the ``error" events as follows.
\begin{equation*}
\begin{split}
      &\mathcal{E}_{0} = \big\{X^{n}\notin T_{\epsilon_{0}}^{(n)}(X)\big\} \\
      &\mathcal{E}_{1a} = \big\{(X^{n},A^{n}(w_1))\notin T_{\epsilon_{1}}^{(n)}(X,A)\  \mbox{for all} \ w_1 \in  \mathcal{W}_1^{(n)} \big\} \\
    &\mathcal{E}_{1b} = \big\{(X^{n},A^{n}(W_1),S_{e}^{n},S_{d}^{n})\notin T_{\epsilon_{1}}^{(n)}(X,A,S_{e},S_{d}) \big\} \\
      &\mathcal{E}_{2} = \big\{\big(X^{n},\hat{X}^{n}(W_1,w_2,v),S_{e}^{n},A^{n}(W_1)\big)\notin
    T_{\epsilon_{1}}^{(n)}(X,\hat{X},S_{e},A)\ \mbox{for all} \ (w_2,v) \in  \mathcal{W}_2^{(n)}\times\mathcal{V}^{(n)}\big\} \\
      &\mathcal{E}_{3} = \big\{\big(S_{d}^{n},\hat{X}^{n}(W_1,W_2,V),A^{n}(W_1)\big)\notin T_{\epsilon}^{(n)}(S_{d},\hat{X},A)\big\} \\
      &\mathcal{E}_{4} = \big\{\big(S_{d}^{n},\hat{X}^{n}(W_1,W_2,\tilde{v}),A^{n}(W_1)\big) \in T_{\epsilon}^{(n)}(S_{d},\hat{X},A)\  \mbox{for some}\ \tilde{v}\in  \mathcal{V}^{(n)}, \tilde{v}\neq V  \big\}.
\end{split}
\end{equation*}
The total ``error" probability is bounded by
\begin{align*}
    \mathrm{Pr}(\mathcal{E}) &\leq \mathrm{Pr}(\mathcal{E}_{0}) + \mathrm{Pr}(\mathcal{E}_{1a}\cap \mathcal{E}_{0}^{c})+ \mathrm{Pr}(\mathcal{E}_{1b}\cap \mathcal{E}_{1a}^{c})+ \mathrm{Pr}(\mathcal{E}_{2}\cap \mathcal{E}_{1b}^{c})+\mathrm{Pr}(\mathcal{E}_{3}\cap \mathcal{E}_{2}^{c}) +\mathrm{Pr}(\mathcal{E}_{4}),
\end{align*}
where $\mathcal{E}_{i}^{c}$ denotes the complement of the event $\mathcal{E}_{i}$.

    0) By the law of large numbers (LLN), $\mathrm{Pr}\big(X^{n}\in
    T_{\epsilon_{0}}^{(n)}(X)\big)\geq 1-\delta_{\epsilon_{0}}$. Since $\delta_{\epsilon_{0}}$ can be made arbitrarily small with increasing $n$ if $\epsilon_{0}>0$, we have $\mathrm{Pr}(\mathcal{E}_{0})\rightarrow 0$ as $n\rightarrow \infty$.

    1a) By the covering lemma \cite{ElGamalKim}, $\mathrm{Pr}(\mathcal{E}_{1a}\cap \mathcal{E}_{0}^{c})\rightarrow 0$ as
$n\rightarrow \infty$ if $\frac{1}{n}\log|\mathcal{W}_1^{(n)}|>I(X;A)+\delta_{\epsilon_{1}}$.

    1b) By the conditional typicality lemma \cite{ElGamalKim} where $(S_{d}^{n}, S_{e}^{n})$ is i.i.d. according to $\prod_{i=1}^{n}P_{S_{d},S_{e}|X,A}\big(s_{d,i},s_{e,i}|x_{i},a_{i}(w_1)\big)$, we have $\mathrm{Pr}(\mathcal{E}_{1b}\cap \mathcal{E}_{1a}^{c})\rightarrow 0$ as $n\rightarrow \infty$.

    2) Averaging over all $W_1=w_1$, by the covering lemma, where each $\hat{X}^{n}$ is drawn independently according to
    $\prod_{i=1}^{n}P_{\hat{X}|A}\big(\hat{x}_{i}|a_{i}(w_1)\big)$, we have that $\mathrm{Pr}(\mathcal{E}_{2}\cap \mathcal{E}_{1b}^{c})\rightarrow 0$ as $n\rightarrow \infty$ if $\frac{1}{n}\log|\mathcal{W}_2^{(n)}|+\frac{1}{n}\log|\mathcal{V}^{(n)}| > I(X,S_{e};\hat{X}|A)+\delta_{\epsilon_{1}}$

    3) Consider the event $\mathcal{E}_{2}^{c}$ in which there exists
    $(W_1,W_2,V)$ such that $\big(X^{n},\hat{X}^{n}(W_1,W_2,V),A^{n}(W_1),S_e^{n}\big)\in
    T_{\epsilon_{1}}^{(n)}(X,\hat{X},A,S_e)$.
    Since we have the Markov chain  $\hat{X}-(X,S_{e},A)-S_{d}$, and $S_{d}^{n}$ is distributed according to $\prod_{i=1}^{n}P_{S_{d}|X,S_{e},A}\big(s_{d,i}|x_{i},s_{e,i},a_{i}(w_1)\big)$, by using the conditional typicality lemma, we have
    \[\mathrm{Pr}\big((X^{n},\hat{X}^{n}(W_1,W_2,V),A^{n}(W_1),S_e^{n},S_{d}^{n})\in
    T_{\epsilon}^{(n)}(X,\hat{X},A,S_{e},S_{d})\big) \rightarrow 1 \ \mbox{as} \ n \rightarrow
    \infty.\]
    This implies that  $\mathrm{Pr}(\mathcal{E}_{3}\cap \mathcal{E}_{2}^{c})\rightarrow 0$ as
$n\rightarrow \infty$.

    4) Averaging over all $W_1=w_1, W_2=w_2$, and $V=v$ \cite[Ch.12, Lemma 1]{ElGamalKim}, by the packing lemma \cite{ElGamalKim} where each $\hat{X}^{n}$ is drawn independently  according to
    $\prod_{i=1}^{n}P_{\hat{X}|A}\big(\hat{x}_{i}|a_{i}\big)$, we have that $\mathrm{Pr}(\mathcal{E}_{4})\rightarrow 0$ as $n\rightarrow \infty$ if $\frac{1}{n}\log|\mathcal{V}^{(n)}| < I(\hat{X};S_{d}|A)-\delta_{\epsilon}$.

Finally, we consider the case where there is no error, i.e.,\[(X^{n},\hat{X}^{n}(W_1,W_2,V),A^{n}(W_1),S_{e}^{n},S_{d}^{n})\in
    T_{\epsilon}^{(n)}(X,\hat{X},A,S_{e},S_{d}).\] By the law of total expectation, the averaged distortion (over all codebooks $\mathfrak{C}$ containing codewords $(\hat{X}^{n},A^{n})$) is given by
     \begin{align*}
      E_{\mathfrak{C},X^{n}}[d^{(n)}(X^{n},\tilde{X}^{n})] &= \mathrm{Pr}(\mathcal{E}) \cdot E_{\mathfrak{C},X^{n}}[d^{(n)}(X^{n},\hat{X}^{n})|\mathcal{E}] +  \mathrm{Pr}(\mathcal{E}^c) \cdot E_{\mathfrak{C},X^{n}}[d^{(n)}(X^{n},\hat{X}^{n})|\mathcal{E}^c]\\
      & \leq \mathrm{Pr}(\mathcal{E}) \cdot d_{max} +  \mathrm{Pr}(\mathcal{E}^c) \cdot E_{\mathfrak{C},X^{n}}[d^{(n)}(X^{n},\hat{X}^{n})|\mathcal{E}^c],
    \end{align*}
where $d_{max}$ is assumed to be the maximal average distortion incurred by the ``error" events.

   Given $\mathcal{E}^c$, the distortion is bounded by
    \begin{align*}
      d^{(n)}(x^{n},\hat{x}^{n}) &= \frac{1}{n}\sum_{i=1}^{n}d\big(x_{i},\hat{x}_{i}\big) \\
       &= \frac{1}{n}\sum_{a,b}N(a,b|x^{n},\hat{x}^{n})d\big(a,b\big) \\
       & \overset{(*)}{\leq}  \sum_{a,b}P_{X,\hat{X}}(a,b)\left(1+\epsilon\right)d\big(a,b\big) = E\big[d\big(X,\hat{X}\big)\big]\left(1+\epsilon\right),
    \end{align*}
where $(*)$ follows from the definition in \eqref{eq:Typical}.

Therefore, we have
 \begin{align*}
      &E_{\mathfrak{C},X^{n}}[d^{(n)}(X^{n},\tilde{X}^{n})] \leq \mathrm{Pr}(\mathcal{E}) \cdot d_{max} +  \mathrm{Pr}(\mathcal{E}^c) \cdot E\big[d\big(X,\hat{X}\big)\big](1+\epsilon).
    \end{align*}
Similarly, we have for the average cost
 \begin{align*}
      &E_{\mathfrak{C}}[\Lambda^{(n)}(A^{n})] \leq \mathrm{Pr}(\mathcal{E}) \cdot c_{max} +  \mathrm{Pr}(\mathcal{E}^c) \cdot E\big[\Lambda(A)](1+\epsilon),
    \end{align*}
where $c_{max}$ is assumed to be the maximal average cost incurred by the ``error" events.

By combining the bounds on the code rates that make
$\mathrm{Pr}(\mathcal{E})\rightarrow 0$ as $n\rightarrow \infty$ and
considering the constraint
$\frac{1}{n}\log|\mathcal{W}^{(n)}|=\frac{1}{n}\log|\mathcal{W}_1^{(n)}||\mathcal{W}_2^{(n)}|
\leq R+\delta$, for any $\delta>0$, we have
\[R+\delta \geq \frac{1}{n}\log|\mathcal{W}^{(n)}| > I(X;A) +I(X,S_{e};\hat{X}|A)-I(\hat{X};S_{d}|A) +\delta'_{\epsilon}, \]
where $\delta'_{\epsilon}$ can be made arbitrarily small, i.e., $\delta'_{\epsilon}\rightarrow 0$ as $\epsilon \rightarrow 0$.

Thus, for any $\delta>0$, if $R \geq I(X;A) +I(X,S_{e};\hat{X}|A)-I(\hat{X};S_{d}|A)$, $E\big[d\big(X,\hat{X}\big)\big] \leq D$ and $E\big[\Lambda(A)] \leq C$, then we have
$\mathrm{Pr}(\mathcal{E})\rightarrow 0$ as $n\rightarrow \infty$, and for all sufficiently large $n$,
\begin{align*}
E_{\mathfrak{C},X^{n}}[d^{(n)}(X^{n},\tilde{X}^{n})] &\leq D + \delta_{\epsilon} \leq D + \delta,\\
E_{\mathfrak{C}}[\Lambda^{(n)}(A^{n})] &\leq C+ \delta_{\epsilon} \leq C + \delta.
\end{align*}

Lastly, with $\mathrm{Pr}(\mathcal{E})\rightarrow 0$ as $n\rightarrow \infty$, it follows that with high probability the decoded codeword $\tilde{X}^{n}= \hat{X}^{n}(W_1,W_2,\tilde{v}) $ at the decoder is the correct one which was chosen at the encoder. We recall the encoding process which determines the codeword $\hat{X}^{n}$ based on $X^{n},S_{e}^{n}$ and $A^{n}$, i.e., there exists a mapping $\psi^{(n)}(\cdot)$ such that $\hat{X}^{n}=\psi^{(n)}(X^{n},S_e^{n},A^{n})$. Thus, for any $\delta >0$, we can have $\mathrm{Pr}\big(\psi^{(n)}(X^{n},S_e^{n},A^{n})\neq g^{(n)}(W_1,W_2,S_{d}^{n})\big) \leq \delta$ for all sufficiently large $n$.

The average distortion, cost, and common reconstruction error probability (over all codebooks) are upper-bounded by $D+\delta, C+\delta$ and $\delta$, respectively. Therefore, there must exist at least one code such that, for sufficiently large $n$, the average distortion, cost, and common reconstruction error probability are upper-bounded by $D+\delta, C+\delta$ and $\delta$.

Thus, any $(R,D,C)$ such that we have $R \geq I(X;A) +I(X,S_{e};\hat{X}|A)-I(\hat{X};S_{d}|A)$, $E\big[d\big(X,\hat{X}\big)\big]\leq D$, and $E[\Lambda(A)]\leq C$ for some  $P_{X}(x)P_{A|X}(a|x)P_{S_{e},S_{d}|X,A}(s_{e},s_{d}|x,a)P_{\hat{X}|X,S_{e},A}(\hat{x}|x,s_{e},a)$ is achievable.
This concludes the achievability proof. \hfill$\blacksquare$

%%%%%%%%%%%%%%%%%%%%%%%%%%%%%%%%%%%%%%%%%%%%%%%%%%%%%%%%%%%%%%%%%%%%%%%%%%%%%%%%%%%%%%%%%%%%%%%
\subsection{Converse Proof of Theorem \ref{theoremRateDist}}
Let us assume the existence of a specific sequence of $(|\mathcal{W}^{(n)}|,n)$ codes
such that for $\delta_{n}>0$, $\frac{1}{n}\log|\mathcal{W}_1^{(n)}||\mathcal{W}_2^{(n)}|\leq R +\delta_{n},
\frac{1}{n}E\left[\sum_{i=1}^{n}d(X_{i},g_{i})\right] \leq D+\delta_{n},\ \frac{1}{n}E\left[\sum_{i=1}^{n}\Lambda(A_{i})\right] \leq C+\delta_{n},$
and $\mathrm{Pr}\big(\psi^{(n)}(X^{n},S_{e}^{n},A^{n})\neq g^{(n)}(W_{1},W_{2},S_{d}^{n})\big)\leq \delta_{n}$,
where $g_{i}$ denotes the $i^{th}$ symbol of $g^{(n)}(W_{1},W_{2},S_{d}^{n})$ and $\lim_{n \to \infty} \delta_{n}=0$. Then we will show that $R \geq R_{\text{ac,cr}}(D,C)$,
where $R_{\text{ac,cr}}(D,C)$ is the rate-distortion-cost function defined as
\begin{equation}\label{eq:minSumRate}
    R_{\text{ac,cr}}(D,C) = \min_{P_{A|X},P_{\hat{X}|X,S_{e},A}} [I(X;A)+I(\hat{X};X,S_{e}|A,S_{d})].
\end{equation}

With $\mathrm{Pr}\big(\psi(X^{n},S_{e}^{n},A^{n})\neq g^{(n)}(W_{1},W_{2},S_{d}^{n})\big)=\delta'_{n} \leq \delta_{n}$, and $|\hat{\mathcal{X}}|=|\tilde{\mathcal{X}}|$, the Fano inequality can be applied to bound
\begin{align}\label{eq:fano}
   & H\big(\psi^{(n)}(X^{n},S_{e}^{n},A^{n})\big|g^{(n)}(W_{1},W_{2},S_{d}^{n})\big) \leq h(\delta'_{n})+\delta'_{n}\log(|\hat{\mathcal{X}}|^{n}-1)  \triangleq n\epsilon_{n},
\end{align}
where $h(\delta'_{n})$ is the binary entropy function, and $ \epsilon_{n} \rightarrow 0$ as $\delta_{n}' \rightarrow 0$.

Then the standard properties of the entropy function give
\begin{align}\label{eq:nR}
  n(R+ \delta_{n})& \geq \log \big(|\mathcal{W}_1^{(n)}| \cdot |\mathcal{W}_2^{(n)}|\big) \geq H(W_{1},W_{2}) \nonumber \\
  & \overset{(*)}{=} H(W_{1},W_{2},A^{n}) = H(A^{n})+ H(W_{1},W_{2}|A^{n})\nonumber \\
  & \geq [H(A^{n})-H(A^{n}|X^{n},S_{e}^{n})]+[H(W_{1},W_{2}|A^{n},S_{d}^{n})-H(W_{1},W_{2}|A^{n},X^{n},S_{e}^{n},S_{d}^{n})]\nonumber \\
  & = \underbrace{H(X^{n},S_{e}^{n})-H(X^{n},S_{e}^{n}|A^{n})+ H(X^{n},S_{e}^{n}|A^{n},S_{d}^{n})}_{=P} \underbrace{-H(X^{n},S_{e}^{n}|A^{n},S_{d}^{n},W_{1},W_{2}) }_{=Q },
\end{align}
where  in $(*)$ we used the fact that $A^{n}=g_{a}^{(n)}(W_{1})$, and $g_{a}^{(n)}(\cdot)$ is the deterministic function. Further,
\begin{align}\label{eq:PP}
  P &= H(X^{n},S_{e}^{n})+ H(S_{d}^{n}|X^{n},S_{e}^{n},A^{n})-H(S_{d}^{n}|A^{n})\nonumber  \\
   & = H(X^{n}) + H(S_{e}^{n}|X^{n})+ H(S_{e}^{n},S_{d}^{n}|X^{n},A^{n})- H(S_{e}^{n}|X^{n},A^{n})-H(S_{d}^{n}|A^{n})\nonumber  \\
   &\overset{(\star)}{\geq} \sum_{i=1}^{n} H(X_{i})+ H(S_{e,i},S_{d,i}|X_{i},A_{i})-H(S_{d,i}|A_{i}) \nonumber \\
   &= \sum_{i=1}^{n} H(X_{i})+ H(S_{e,i}|X_{i},A_{i}) + H(S_{d,i}|X_{i},S_{e,i},A_{i})-H(S_{d,i}|A_{i}) \nonumber \\
   &= \sum_{i=1}^{n} H(X_{i})+ H(S_{e,i}|X_{i},A_{i}) - H(X_{i},S_{e,i}|A_{i})+H(X_{i},S_{e,i}|A_{i},S_{d,i}) \nonumber \\
   &= \sum_{i=1}^{n} I(X_{i};A_{i})+ H(X_{i},S_{e,i}|A_{i},S_{d,i}),
\end{align}
where $(\star)$ holds due to the i.i.d. property of $P_{X^{n}}$ and $P_{S_{e}^{n},S_{d}^{n}|X^{n},A^{n}}$, %And  %and that conditioning reduces entropy for the last term. And
\begin{align}\label{eq:QQ1}
  Q &= -H\big(X^{n},S_{e}^{n}|A^{n},S_{d}^{n},W_{1},W_{2},g^{(n)}(W_{1},W_{2},S_{d}^{n})\big) \nonumber \\
   &\geq -H\big(X^{n},S_{e}^{n}|A^{n},S_{d}^{n},g^{(n)}(W_{1},W_{2},S_{d}^{n})\big) \nonumber \\
   &= -H\big(\psi^{(n)}(X^{n},S_{e}^{n},A^{n}),X^{n},S_{e}^{n}|A^{n},S_{d}^{n},g^{(n)}(W_{1},W_{2},S_{d}^{n})\big) \nonumber \\ & \qquad +H(\psi^{(n)}(X^{n},S_{e}^{n},A^{n})|A^{n},S_{d}^{n},g^{(n)}(W_{1},W_{2},S_{d}^{n}),X^{n},S_{e}^{n})\nonumber \\
   &\geq -H\big(\psi^{(n)}(X^{n},S_{e}^{n},A^{n})|g^{(n)}(W_{1},W_{2},S_{d}^{n})\big)-H\big(X^{n},S_{e}^{n}|A^{n},S_{d}^{n},\psi^{(n)}(X^{n},S_{e}^{n},A^{n})\big) \nonumber \\
   &\overset{(a)}{\geq} -n\epsilon_{n}-\sum_{i=1}^{n}H(X_{i},S_{e,i}|A^{n},S_{d}^{n},\psi^{(n)}(X^{n},S_{e}^{n},A^{n}),X^{i-1},S_{e}^{i-1})\nonumber \\
   &\overset{(b)}{\geq} -n\epsilon_{n} - \sum_{i=1}^{n} H\big(X_{i},S_{e,i}|A_{i},S_{d,i},\psi^{(n)}_{i}(X^{n},S_{e}^{n},A^{n})\big),
\end{align}
where $(a)$ follows from \eqref{eq:fano} and $\psi^{(n)}_{i}(X^{n},S_{e}^{n},A^{n})$ in $(b)$ corresponds to the $i^{th}$ symbol of $\psi^{(n)}(X^{n},S_{e}^{n},A^{n})$.

Define $\hat{X}_{i}=\psi^{(n)}_{i}(X^{n},S_{e}^{n},A^{n})$. Then the Markov chain $\hat{X}^{n}-(X^{n},S_{e}^{n},A^{n})-S_{d}^{n}$ holds. Together with the memoryless property, $P_{S_{e}^{n},S_{d}^{n}|X^{n},A^{n}}(s_{e}^{n},s_{d}^{n}|x^{n},a^{n})=\prod_{i=1}^{n}P_{S_{e},S_{d}|X,A}(s_{e,i},s_{d,i}|x_{i},a_{i})$, we also have that $(S_{d}^{i-1},X^{n\setminus i},S_{e}^{n\setminus i},A^{n\setminus i},\hat{X}^{n})-(X_{i},S_{e,i},A_{i})-S_{d,i}$ forms a Markov chain.
Combining \eqref{eq:nR}-\eqref{eq:QQ1}, we have
\begin{align}\label{eq:nRSum}
   R+ \delta_{n} & \geq \frac{1}{n} \log \big(|\mathcal{W}_1^{(n)}| \cdot |\mathcal{W}_2^{(n)}|\big) \nonumber\\
   & \geq \frac{1}{n}\sum_{i=1}^{n} I(X_{i};A_{i})+ I(\hat{X}_{i};X_{i},S_{e,i}|A_{i},S_{d,i})- \epsilon_{n} \nonumber\\
   & \overset{(a)}{\geq}\frac{1}{n}\sum_{i=1}^{n} R_{\text{ac,cr}}\Big(E\big[d\big(X_{i},\hat{X}_{i}\big)\big] ,E[\Lambda(A_{i})]\Big)- \epsilon_{n} \nonumber\\
   & \overset{(b)}{\geq} R_{\text{ac,cr}}\left(\frac{1}{n}\sum_{i=1}^{n}E\big[d\big(X_{i},\hat{X}_{i}\big)\big],\frac{1}{n}\sum_{i=1}^{n}E[\Lambda(A_{i})]\right)- \epsilon_{n},
\end{align}
where $(a)$ follows from the definition of $R_{\text{ac,cr}}(D,C)$ in \eqref{eq:minSumRate}, and the fact that $\hat{X}_{i}-(X_{i},S_{e,i},A_{i})-S_{d,i}$ forms a Markov chain,
$(b)$ follows from Jensen's inequality and convexity of $R_{\text{ac,cr}}(D,C)$.

Let $\beta$ be the event that the reconstruction at the encoder is not equal to that at the decoder, i.e., $\beta=\{\psi^{(n)}(X^{n},S_{e}^{n},A^{n})\neq g^{(n)}(W_{1},W_{2},S_{d}^{n})\}$.
It then follows that
\begin{align}\label{eq:E(x,g)}
&E[d(X_{i},g_{i})] = E[d(X_{i},g_{i})|\beta^{c}]\cdot\mathrm{Pr}(\beta^{c})+ E[d(X_{i},g_{i})|\beta]\cdot\mathrm{Pr}(\beta) \overset{(\star)}{\geq} E[d(X_{i},\hat{X}_{i})|\beta^{c}]\cdot\mathrm{Pr}(\beta^{c}),
\end{align}
where $(\star)$ holds because we have $g_{i}=\hat{X}_{i}$ for given $\beta^{c}$. Thus
\begin{align}\label{eq:E(x,xhat)}
         \frac{1}{n}\sum_{i=1}^{n}E[d(X_{i},\hat{X}_{i})] &= \frac{1}{n}\sum_{i=1}^{n}E[d(X_{i},\hat{X}_{i})|\beta^{c}]\cdot\mathrm{Pr}(\beta^{c}) + E[d(X_{i},\hat{X}_{i})|\beta]\cdot\mathrm{Pr}(\beta)\nonumber\\
        &\overset{(a)}{\leq} \frac{1}{n}\sum_{i=1}^{n}E[d(X_{i},\hat{X}_{i})|\beta^{c}]\cdot\mathrm{Pr}(\beta^{c})+ \tilde{d}_{max}\delta_{n}\nonumber\\
        &\overset{(b)}{\leq} \frac{1}{n}\sum_{i=1}^{n}E[d(X_{i},g_{i})]+ \tilde{d}_{max}\delta_{n}\nonumber\\
        &\overset{(c)}{\leq} D+ (1+\tilde{d}_{max})\delta_{n},
\end{align}
where $(a)$ follows from the assumption that $\tilde{d}_{max}$ is the maximum average distortion incurred by the error event $\beta$ and that $\mathrm{Pr}\big(\psi^{(n)}(X^{n},S_{e}^{n},A^{n})\neq g^{(n)}(W_{1},W_{2},S_{d}^{n})\big) \leq \delta_{n}$, $(b)$ follows from \eqref{eq:E(x,g)}, and $(c)$ follows from the assumption that $\frac{1}{n} E\left[\sum_{i=1}^{n}d(X_{i},g_{i})\right] \leq D+\delta_{n}$ in the beginning.

Finally, we substitute \eqref{eq:E(x,xhat)} into
\eqref{eq:nRSum}. With $\lim_{n \to \infty} \delta_{n}=0$, and
$\lim_{n \to \infty} \epsilon_{n}=0$, we thus get $R \geq
R_{\text{ac,cr}}(D,C)$ by using the assumption that
$\frac{1}{n}E\left[\sum_{i=1}^{n}\Lambda(A_{i})\right] \leq
C+\delta_{n}$ and the non-increasing property of $R_{\text{ac,cr}}(D,C)$. This concludes the proof of converse. \hfill$\blacksquare$
%%%%%%%%%%%%%%%%%%%%%%%%%%%%%%%%%%%%%%%%%%%%%%%%%%%%%%%%%%%%%%%%%%%%%%%%%%%%%%%%%%%%%%%%%%%%%%%%%%%%%%%%%%%%%%%%%%%%%%%%%%

\section{Converse Proof of Proposition \ref{theoremRateDist2}}\label{sec:proofISIT}
Let $(W_{1},W_{2})\in \{1,2,\ldots,|\mathcal{W}_{1}^{(n)}|\} \times \{1,2,\ldots,|\mathcal{W}_{2}^{(n)}|\}$ denote the encoded version of $X^{n}$ where $|\mathcal{W}^{(n)}|=|\mathcal{W}_{1}^{(n)}|\cdot|\mathcal{W}_{2}^{(n)}|$.  Let us assume the existence of a specific sequence of $(|\mathcal{W}^{(n)}|,n)$ codes
such that for $\delta_{n}>0$, $\frac{1}{n}\log|\mathcal{W}^{(n)}|\leq R +\delta_{n},
\frac{1}{n}E\left[\sum_{i=1}^{n}d(X_{i},\hat{X}_{i})\right] \leq D+\delta_{n},\ \frac{1}{n}E\left[\sum_{i=1}^{n}\Lambda(A_{i})\right] \leq C+\delta_{n},$
where $\hat{X}_{i}$ denotes the $i^{th}$ symbol of $\hat{X}^{n}=g^{(n)}(W,S_{d}^{n})$ and $\lim_{n \to \infty} \delta_{n}=0$. Then we identify $U$ and $\tilde{g}:\mathcal{U} \times \mathcal{S}_{d}\rightarrow \hat{\mathcal{X}}$ and show that $R \geq R_{\text{ac}}(D,C)$,
where $R_{\text{ac}}(D,C)$ is the rate-distortion-cost function defined as
\begin{equation}\label{eq:minSumRate2}
    R_{\text{ac}}(D,C) = \min_{P_{A|X},P_{U|X,S_{e},A},\tilde{g}:\mathcal{U} \times \mathcal{S}_{d}\rightarrow \hat{\mathcal{X}}} [I(X;A)+I(U;X,S_{e}|A,S_{d})].
\end{equation}

We start bounding the rate as in \eqref{eq:nR},
\begin{align}\label{eq:nR_copied}
  n(R+ \delta_{n})& \geq \log \big(|\mathcal{W}^{(n)}| \big) \geq H(W) \nonumber \\
  & \geq  \underbrace{H(X^{n},S_{e}^{n})-H(X^{n},S_{e}^{n}|A^{n})+ H(X^{n},S_{e}^{n}|A^{n},S_{d}^{n})}_{=P} \underbrace{-H(X^{n},S_{e}^{n}|A^{n},S_{d}^{n},W) }_{=Q }
\end{align}
The term $P$ is given as in \eqref{eq:PP},
\begin{align}\label{eq:PP_copied}
  P &\geq  \sum_{i=1}^{n} I(X_{i};A_{i})+ H(X_{i},S_{e,i}|A_{i},S_{d,i})
\end{align}
and
\begin{align}\label{eq:QQ}
  Q %&= -H(X^{n},S_{e}^{n}|A^{n},S_{d}^{n},T,W) \\
   &= -H(X^{n}|A^{n},S_{d}^{n},W)- H(S_{e}^{n}|A^{n},S_{d}^{n},W,X^{n}) \nonumber \\
   %&= \sum_{i=1}^{n} -H(X_{i}|A^{n},S_{d}^{n},T,W,X^{i-1})\\
%   & \qquad - H(S_{e,i}|A^{n},S_{d}^{n},T,W,X^{n},S_{e}^{i-1}) \\
   &\geq \sum_{i=1}^{n} -H(X_{i}|A^{n},S_{d}^{n},W,X^{i-1}) - H(S_{e,i}|A^{n},S_{d}^{n},W,X^{i-1},X_{i})
   %&= \sum_{i=1}^{n} H(X_{i},S_{e,i}|A^{n},S_{d}^{n},T,W,X^{i-1})\\
   = - \sum_{i=1}^{n} H(X_{i},S_{e,i}|U_{i},A_{i},S_{d,i}),
\end{align}
where $U_{i} \triangleq(A^{n\setminus i},S_{d}^{n\setminus i},W,X^{i-1}),\ i=1,2,\ldots,n$.

Combining \eqref{eq:nR_copied}-\eqref{eq:QQ}, and letting $n\rightarrow \infty$, we have
\begin{align*}
  nR &\geq \sum_{i=1}^{n} I(X_{i};A_{i})+ I(U_{i};X_{i},S_{e,i}|A_{i},S_{d,i})\\
   & \overset{(a)}{\geq}\sum_{i=1}^{n} R_{\text{ac}}\Big(E\big[d\big(X_{i},\tilde{g}_{i}(U_{i},S_{d,i})\big)\big] ,E[\Lambda(A_{i})]\Big) \\
   & \overset{(b)}{\geq} n R_{\text{ac}}\left(\frac{1}{n}\sum_{i=1}^{n}E\big[d\big(X_{i},\tilde{g}_{i}(U_{i},S_{d,i})\big)\big],\frac{1}{n}\sum_{i=1}^{n}E[\Lambda(A_{i})]\right)\\
   %&= n\cdot R\left(E[\frac{1}{n}\sum_{i=1}^{n}d(X_{i},\hat{X}_{i})] ,E[\frac{1}{n}\sum_{i=1}^{n}\Lambda(A_{i})]\right) \\
   & \overset{(c)}{\geq} n R_{\text{ac}}(D,C),
\end{align*}
where $(a)$ follows from the definition of rate-distortion-cost function in \eqref{eq:minSumRate2} and the fact that $U_{i}-(A_{i},X_{i},S_{e,i})-S_{d,i}$ forms a Markov chain, and that $\hat{X}_{i} =g^{(n)}_{i}(W,S_{d}^{n})\triangleq \tilde{g}_{i}(U_{i},S_{d,i})$ for some $\tilde{g}_{i}(\cdot)$, $(b)$ follows from Jensen's inequality and convexity of $R_{\text{ac}}(D,C)$ which can be proved similarly as in \cite{Permuter2011} or Lemma \ref{lemma:convexRDC}, and $(c)$ follows from the non-increasing property of $R_{\text{ac}}(D,C)$, $\frac{1}{n}E\big[\sum_{i=1}^{n}d(X_{i},\hat{X}_{i})\big] \leq D+\delta_{n}$, and $\frac{1}{n}E\big[\sum_{i=1}^{n}\Lambda(A_{i})\big] \leq C+\delta_{n}$.% which have to hold for any $\delta>0$.

For the bound on the cardinality of the set of $U$, it can be shown by using the support lemma \cite{CsiszarBook} that $\mathcal{U}$ should have $|\mathcal{A}||\mathcal{X}|-1$ elements to preserve $P_{A,X}$, plus four more for $I(U;X,S_{e}|A)$, $I(U;S_{d}|A)$, the distortion, and the cost constraints.
This finally concludes the proof. \hfill$\blacksquare$

%%%%%%%%%%%%%%%%%%%%%%%%%%%%%%%%%%%%%%%%%%%%%%%%%%%%%%%%%%%%%%%%%%%%%%%%%%%%%%%%%%%%%%%%%%%%%%%%%%%%%%%%%%%%%%%%%%%%%%%%%%
\section{Proof of Theorem \ref{theoremCapa}}\label{sec:proofchannelcoding}
\subsection{Achievability Proof of Theorem \ref{theoremCapa}}

Similarly to the previous achievability proof in Theorem \ref{theoremRateDist}, the proof follows from a standard random coding argument where we use the definition and properties of  $\epsilon$-typicality as in \cite{ElGamalKim}. %, i.e., the set of $\epsilon$-typical sequence for $\epsilon >0$ with respect to $P_{X}(\cdot)$ is denoted by
%\begin{equation*}%\label{eq:Typical}
%    T_{\epsilon}^{(n)}(X)=\left\{ x^{n}\in \mathcal{X}^{n}: \left|\frac{1}{n}N(x|x^{n})-P_{X}(x)\right| \leq \epsilon P_{X}(x) \right\},
%\end{equation*}
%where $N(x|x^{n})$ is the number of occurrences of $x$ in $x^{n}$.
We use the technique of rate splitting, i.e., the message $M$ of rate $R$ is split into two messages $M_{1}$ and $M_{2}$ of rates $R_{1}$ and $R_{2}$. Two-stage coding is then considered, i.e., a first stage for communicating the identity of the action sequence, and a second stage for communicating the identity of $X^{n}$ based on the known action sequence. 

For given channels with transition probabilities
$P_{S_{e},S_{d}|A}(s_{e},s_{d}|a)$ and $P_{Y|X,S_{e},S_{d}}(y|x,s_{e},s_{d})$ we can assign the joint
probability to any random vector $(A,X,S_{e})$ by
\begin{align*}
P_{A,S_{e},S_{d},X,Y}(a,s_{e},s_{d},x,y)%&= P_{A,S_{e},X}(a,s_{e},x)P_{Y|X,S_{e}}(y|x,s_{e}) \\
  & =P_{A}(a)P_{S_{e},S_{d}|A}(s_{e},s_{d}|a)P_{X|A,S_{e}}(x|a,s_{e})P_{Y|X,S_{e},S_{d}}(y|x,s_{e},s_{d})
\end{align*}

\textit{Codebook Generation}: Fix $P_{A}$ and $P_{X|A,S_{e}}$. Let $\mathcal{M}_{1}^{(n)}=\{1,2,\ldots,|\mathcal{M}_{1}^{(n)}|\}$, $\mathcal{M}_{2}^{(n)}=\{1,2,\ldots,|\mathcal{M}_{2}^{(n)}|\}$ and $\mathcal{J}^{(n)}=\{1,2,\ldots,|\mathcal{J}^{(n)}|\}$. For all $m_{1}\in \mathcal{M}_{1}^{(n)}$, randomly generate $a^{n}(m_{1})$ i.i.d. according to $\prod_{i=1}^{n}P_{A}(a_{i})$. For each $m_{1}\in \mathcal{M}_{1}^{(n)}$, generate $|\mathcal{M}_{2}^{(n)}|\cdot |\mathcal{J}^{(n)}|$ codewords, $\{x^{n}(m_{1},m_{2},j)\}_{m_{2} \in \mathcal{M}_{2}^{(n)}, j \in \mathcal{J}^{(n)}}$ i.i.d. each according to $\prod_{i=1}^{n}P_{X|A}(x_{i}|a_{i}(m_1))$. Then the codebooks are revealed to the action encoder, the channel encoder and the decoder. Let $0 < \epsilon_{0} <\epsilon_{1}<\epsilon < 1$.

\textit{Encoding}: Given the message $m=(m_{1},m_{2})\in \mathcal{M}^{(n)}$, the action codeword $a^{n}(m_{1})$ is chosen and the channel state information $(s_{e}^{n},s_{d}^{n})$ is generated as an output of the memoryless channel $P_{S_{e}^{n},S_{d}^{n}|A^{n}}(s_{e}^{n},s_{d}^{n}|a^{n})=\prod_{i=1}^{n}P_{S_{e},S_{d}|A}(s_{e,i},s_{d,i}|a_{i})$.
The encoder looks for the smallest value of $j \in \mathcal{J}^{(n)}$ such that $\big(s_{e}^{n},a^{n}(m_{1}),x^{n}(m_{1},m_{2},j)\big)\in T_{\epsilon_{1}}^{(n)}(S_{e},A,X)$. If no such $j$ exists, set $j=1$. The channel input sequence is then chosen to be $x^{n}(m_{1},m_{2},j)$.

\textit{Decoding}:  Upon receiving $y^{n}$ and $s_{d}^{n}$, the decoder in the first step looks for the smallest $\tilde{m}_{1} \in \mathcal{M}_{1}^{(n)}$ such that $\big(y^{n},s_{d}^{n},a^{n}(\tilde{m}_{1})\big)\in T_{\epsilon}^{(n)}(Y,S_{d},A)$. If successful, then set $\hat{m}_{1}=\tilde{m}_{1}$. Otherwise, set $\hat{m}_{1}=1$. Then, based on the known $a^{n}(\hat{m}_{1})$, the decoder looks for a pair $(\tilde{m}_{2},\tilde{j})$ with the smallest $\tilde{m}_{2} \in \mathcal{M}_{2}^{(n)}$ and $\tilde{j} \in \mathcal{J}^{(n)}$ such that $\big(y^{n},s_{d}^{n}, a^{n}(\hat{m}_{1}),x^{n}(\hat{m}_{1},\tilde{m}_{2},\tilde{j})\big)\in T_{\epsilon}^{(n)}(Y,S_{d},A,X)$.
If there exists such a pair, the decoded message is set to be $\hat{m} =(\hat{m}_{1},\tilde{m}_{2})$, and $\hat{x}^{n}=x^{n}(\hat{m}_{1},\tilde{m}_{2},\tilde{j})$. Otherwise, $\hat{m} =(1,1)$ and $\hat{x}^{n}=x^{n}(1,1,1)$.\footnotemark[3]

\footnotetext[3]{We note that although the simultaneous joint typicality decoding gives us different constraints on the individual rate as compared to the sequential two-stage decoding considered in this paper, it gives the same constraints on the total transmission rate in which we are interested.}

\textit{Analysis of Probability of Error}: Due to the symmetry of the random code construction, the error probability does not depend on which message was sent. Assuming that $M=(M_{1},M_{2})$ and $J$ were sent and chosen at the encoder. We define the error events as follows.
    \begin{align*}
      &\mathcal{E}_{1} = \{A^{n}(M_{1})\notin
    T_{\epsilon_{0}}^{(n)}(A)\} \\
      &\mathcal{E}_{2} = \big\{\big(S_{e}^{n},S_{d}^{n},A^{n}(M_{1})\big)\notin
    T_{\epsilon_{1}}^{(n)}(S_{e},S_{d},A)\big\} \\
      &\mathcal{E}_{3} = \big\{ \big(S_{e}^{n},A^{n}(M_{1}),X^{n}(M_{1},M_{2},j)\big)\notin
    T_{\epsilon_{1}}^{(n)}(S_{e},A,X)\ \mbox{for all}\ j \in \mathcal{J}^{(n)} \big\}\\
    &\mathcal{E}_{4a} = \big\{ \big(Y^{n},S_{d}^{n},A^{n}(M_{1})\big)\notin
    T_{\epsilon}^{(n)}(Y,S_{d},A)\big\}\\
     &\mathcal{E}_{4b} = \big\{ \big(Y^{n},S_{d}^{n},A^{n}(\tilde{m}_{1})\big)\in
    T_{\epsilon}^{(n)}(Y,S_{d},A)\ \mbox{for some} \ \tilde{m}_{1} \in \mathcal{M}_{1}^{(n)}, \tilde{m}_{1} \neq M_{1} \big\}\\
      &\mathcal{E}_{5a} = \big\{ \big(Y^{n},S_{d}^{n},A^{n}(M_{1}),X^{n}(M_{1},M_{2},J)\big)\notin
    T_{\epsilon}^{(n)}(Y,S_{d},A,X)\big\} \\
      &\mathcal{E}_{5b} = \big\{\big(Y^{n},S_{d}^{n},A^{n}(M_{1}),X^{n}(M_{1},\tilde{m}_{2},\tilde{j})\big)\in T_{\epsilon}^{(n)}(Y,S_{d},A,X)\ \mbox{for some}\ (\tilde{m}_{2},\tilde{j})\in \mathcal{M}_{2}^{(n)} \times  \mathcal{J}^{(n)},  (\tilde{m}_{2},\tilde{j}) \neq (M_{2},J)\big\}.
    \end{align*}

The probability of error events can be bounded by
\begin{align*}
    \mathrm{Pr}(\mathcal{E}) %&= \mathrm{Pr}\big((\hat{m},\hat{x}^{n})\neq (m,x^{n}) \big)\\
    &\leq \mathrm{Pr}(\mathcal{E}_{1}) + \mathrm{Pr}(\mathcal{E}_{2}\cap \mathcal{E}_{1}^{c}) + \mathrm{Pr}(\mathcal{E}_{3}\cap \mathcal{E}_{2}^{c}) + \mathrm{Pr}(\mathcal{E}_{4a}\cap \mathcal{E}_{3}^{c}) + \mathrm{Pr}(\mathcal{E}_{4b})+ \mathrm{Pr}(\mathcal{E}_{5a}\cap \mathcal{E}_{3}^{c})+\mathrm{Pr}(\mathcal{E}_{5b}),
\end{align*}
where $\mathcal{E}_{i}^{c}$ denotes the complement of event $\mathcal{E}_{i}$.

1) Since $A^{n}(M_{1})$ is i.i.d. according to $P_{A}$, by the LLN we have $\mathrm{Pr}(\mathcal{E}_{1})\rightarrow 0$ as $n\rightarrow \infty$.

2) Consider the event $\mathcal{E}_{1}^{c}$ where we have $A^{n}(M_{1})\in T_{\epsilon_{0}}^{(n)}(A)$.  Since $(S_{d}^{n},S_{e}^{n})$ is distributed according to $\prod_{i=1}^{n}P_{S_{e},S_{d}|A}\big(s_{e,i},s_{d,i}|a_{i}\big)$, by the conditional typicality lemma \cite{ElGamalKim}, we have that $\mathrm{Pr}\big(\mathcal{E}_{2} \cap \mathcal{E}_{1}^{c}\big)\rightarrow 0$ as $n\rightarrow \infty$.

3) By the covering lemma \cite{ElGamalKim} where $X^{n}$ is i.i.d. according to $\prod_{i=1}^{n}P_{X|A}(x_{i}|a_{i})$, we have $\mathrm{Pr}\big(\mathcal{E}_{3} \cap \mathcal{E}_{2}^{c}\big)\rightarrow 0$ as $n\rightarrow \infty$ if $\frac{1}{n}\log |\mathcal{J}^{(n)}| > I(X;S_{e}|A)+ \delta_{\epsilon_{1}}$, where $\delta_{\epsilon_{1}}\rightarrow 0$ as $\epsilon_{1}\rightarrow 0$.

4a) Consider the event $\mathcal{E}_{3}^{c}$ where we have $\big(S_{e}^{n},A^{n}(M_{1}),X^{n}(M_{1},M_{2},J)\big)\in T_{\epsilon_{1}}^{(n)}(S_{e},A,X)$. Since we have $S_{d}-(A,S_{e})-X$ forms a Markov chain and $S_{d}^{n}$ is distributed according to $\prod_{i=1}^{n}P_{S_{d}|A,S_{e}}\big(s_{d,i}|a_{i},s_{e,i})\big)$, we have that by the conditional typicality lemma \cite{ElGamalKim}, $\mathrm{Pr}\big((S_{d}^{n},S_{e}^{n},A^{n}(M_{1}),X^{n}(M_{1},M_{2},J))\in T_{\epsilon}^{(n)}(S_{d},S_{e},A,X) \big) \rightarrow 1$ as $n\rightarrow \infty$. And since we have the Markov chain $A-(X,S_{e},S_{d})-Y$ and $Y^{n}$ is distributed according to $\prod_{i=1}^{n}P_{Y|X,S_{e},S_{d}}(y_{i}|x_{i},s_{e,i},s_{d,i})$, by using once again the conditional typicality lemma, it follows that \[\mathrm{Pr}\big(\big(Y^{n}, S_{e}^{n},S_{d}^{n}, A^{n}(M_{1}),X^{n}(M_{1},M_{2},J)\big)\in T_{\epsilon}^{(n)}(Y,A,S_{e},S_{d},X)\big)\rightarrow 1 \ \mbox{as} \ n\rightarrow \infty.\]
This also implies that $\mathrm{Pr}(\mathcal{E}_{4a}\cap \mathcal{E}_{3}^{c})\rightarrow 0$ as $n\rightarrow \infty$.

4b) By the packing lemma \cite{ElGamalKim}, we have $\mathrm{Pr}\big(\mathcal{E}_{4b}\big)\rightarrow 0$ as $n\rightarrow \infty$ if $\frac{1}{n}\log |\mathcal{M}_{1}^{(n)}| < I(A;Y,S_{d})- \delta_{\epsilon}$, where $\delta_{\epsilon}\rightarrow 0$ as $\epsilon \rightarrow 0$.

5a) As in $E_{4a}$) we have $\mathrm{Pr}(E_{5a}\cap E_{3}^{c})\rightarrow 0$ as $n\rightarrow \infty$.

5b) Averaging over all $J=j$, by the packing lemma where $X^{n}$ is i.i.d. according to $\prod_{i=1}^{n}P_{X|A}(x_{i}|a_{i})$, we have $\mathrm{Pr}\big(E_{5b}\big)\rightarrow 0$ as $n\rightarrow \infty$ if $\frac{1}{n}\log |\mathcal{M}_{2}^{(n)}|+\frac{1}{n}\log |\mathcal{J}^{(n)}|< I(X;Y,S_{d}|A)- \delta_{\epsilon}$.

Finally, by combining the bounds on the code rates that make $\mathrm{Pr}\big(\mathcal{E}\big)\rightarrow 0$ as $n\rightarrow \infty$,
\begin{align*}
 \frac{1}{n}\log |\mathcal{J}^{(n)}|&> I(X;S_{e}|A)+ \delta_{\epsilon_{1}} \\
 \frac{1}{n}\log |\mathcal{M}_{1}^{(n)}|&< I(A;Y,S_{d})- \delta_{\epsilon} \\
 \frac{1}{n}\log |\mathcal{M}_{2}^{(n)}|+\frac{1}{n}\log |\mathcal{J}^{(n)}|&< I(X;Y,S_{d}|A)- \delta_{\epsilon},
\end{align*}
we have
\begin{align*}
 \frac{1}{n}\log |\mathcal{M}^{(n)}| =  \frac{1}{n}\log |\mathcal{M}_{1}^{(n)}||\mathcal{M}_{2}^{(n)}|&< I(A,X;Y,S_{d})-I(X;S_{e}|A)- \delta'_{\epsilon} \\
 \frac{1}{n}\log |\mathcal{M}_{2}^{(n)}| &< I(X;Y,S_{d}|A)-I(X;S_{e}|A)- \delta''_{\epsilon},
\end{align*}
where $\delta'_{\epsilon},\delta''_{\epsilon}\rightarrow 0$ as $\epsilon \rightarrow 0$.

Since, for any $\delta>0$, the achievable rate $R$ satisfies $\frac{1}{n}\log |\mathcal{M}^{(n)}| \geq R- \delta$, and we know that $\frac{1}{n}\log |\mathcal{M}_{2}^{(n)}| \geq 0$, then we get
\begin{align}
 R-\delta &\leq\frac{1}{n}\log |\mathcal{M}^{(n)}|< I(A,X;Y,S_{d})-I(X;S_{e}|A)- \delta'_{\epsilon} \nonumber \\
 \mbox{and}\quad 0 &\leq \frac{1}{n}\log |\mathcal{M}_{2}^{(n)}| < I(X;Y,S_{d}|A)-I(X;S_{e}|A)- \delta''_{\epsilon}.\label{eq:M2}
\end{align}

Since $\epsilon$ can be made arbitrarily small for increasing $n$, and by a standard random coding argument, we have that
\begin{align*}
 &R \leq I(A,X;Y,S_{d})-I(X;S_{e}|A)\\
 \mbox{and}\quad &0 < I(X;Y,S_{d}|A)-I(X;S_{e}|A),
\end{align*}
for some $P_{A}(a)P_{S_{e},S_{d}|A}(s_{e},s_{d}|a)P_{X|A,S_{e}}(x|a,s_{e})P_{Y|X,S_{e},S_{d}}(y|x,s_{e},s_{d})$ is achievable.

Note that the latter condition is for the two-stage coding to be successful, i.e., we can split the message into two parts with positive rates.
This concludes the achievability proof. \hfill$\blacksquare$

%%%%%%%%%%%%%%%%%%%%%%%%%%%%%%%%%%%%%%%%%%%%%%%%%%%%%%%%%%%%%%%%%%%%%%%%%%%%%%%%%%%%%%%%%%%%%%%%%%%%%%%%%%%%%%%
\subsection{Converse Proof of Theorem \ref{theoremCapa}}
We show that for any achievable rate $R$, it follows that $R \leq I(A,X;Y,S_{d})-I(X;S_{e}|A)$ and $0 \leq I(X;Y,S_{d}|A)-I(X;S_{e}|A)$ for some $P_{A}(a)P_{S_{e},S_{d}|A}(s_{e},s_{d}|a)P_{X|A,S_{e}}(x|a,s_{e})P_{Y|X,S_{e},S_{d}}(y|x,s_{e},s_{d})$. From the problem formulation, we can write the joint probability mass function,
\begin{align}
&P_{M,A^{n},S_{e}^{n},S_{d}^{n},X^{n},Y^{n},\hat{M},\hat{X}^{n}}(m,a^{n},s_{e}^{n},s_{d}^{n},x^{n},y^{n},\hat{m},\hat{x}^{n}) \nonumber \\
& = \frac{1_{\{f_{a}^{(n)}(m)=a^{n},f^{(n)}(m,s_{e}^{n})=x^{n},g_{x}^{(n)}(y^{n},s_{d}^{n})=\hat{x}^{n},g_{m}^{(n)}(y^{n},s_{d}^{n})=\hat{m}\}}}{|\mathcal{M}^{(n)}|}\cdot \prod_{i=1}^{n}P_{S_{e},S_{d}|A}(s_{e,i},s_{d,i}|a_{i})P_{Y|X,S_{e},S_{d}}(y_{i}|x_{i},s_{e,i},s_{d,i}),\label{eq:pmf1}
\end{align}
where $M$ is chosen uniformly at random from the set $\mathcal{M}^{(n)}=\{1,2,\ldots,|\mathcal{M}^{(n)}|\}$.
\begin{lemma}\label{lemma:markovchain1}
For the joint pmf in \eqref{eq:pmf1}$, S_{d,i}-(A_{i},S_{e,i},X_{i})-(M,A^{n\setminus i},S_{e,i+1}^{n},X^{n},Y^{i-1},S_{d}^{i-1})$ forms a Markov chain.
\end{lemma}
\begin{proof}
From \eqref{eq:pmf1}, we use the undirected graph as a tool to derive the Markov chain \cite{PermSteinWeiss}, \cite{Pearl}. Let
\begin{align*}
\mathcal{U}& \triangleq (M,A^{n\setminus i},S^{n\setminus i}_{e},X^{n},Y^{i-1},S^{n\setminus i}_{d}) \\
\mathcal{V} & \triangleq (A_{i},S_{e,i},X_{i})\\
\mathcal{W} & \triangleq (S_{d,i},Y_{i})
\end{align*}
We can draw the undirected graph associated with the marginal pmf derived from the joint pmf in \eqref{eq:pmf1} in Fig. \ref{fig:markovchain}.

 \begin{figure}[]
    \centering
    \psfrag{M}[][][0.7]{$M$}
    \psfrag{A^i-1}[][][0.7]{$A^{i-1}$}
    \psfrag{Ai}[][][0.7]{$A_{i}$}
    \psfrag{A_i+1^n}[][][0.7]{$A_{i+1}^{n}$}
    \psfrag{Se^i-1}[][][0.7]{$S_{e}^{i-1}$}
    \psfrag{Sei}[][][0.7]{$S_{e,i}$}
    \psfrag{Se_i+1^n}[][][0.7]{$S_{e,i+1}^{n}$}
    \psfrag{Sd^i-1}[][][0.7]{$S_{d}^{i-1}$}
    \psfrag{Sdi}[][][0.7]{$S_{d,i}$}
    \psfrag{Sd_i+1^n}[][][0.7]{$S_{d,i+1}^{n}$}
    \psfrag{X^i-1}[][][0.7]{$X^{i-1}$}
    \psfrag{Xi}[][][0.7]{$X_{i}$}
    \psfrag{X_i+1^n}[][][0.7]{$X_{i+1}^{n}$}
    \psfrag{Yi}[][][0.7]{$Y_{i}$}
    \psfrag{Y^i-1}[][][0.7]{$Y^{i-1}$}
    \includegraphics[width=8.5cm]{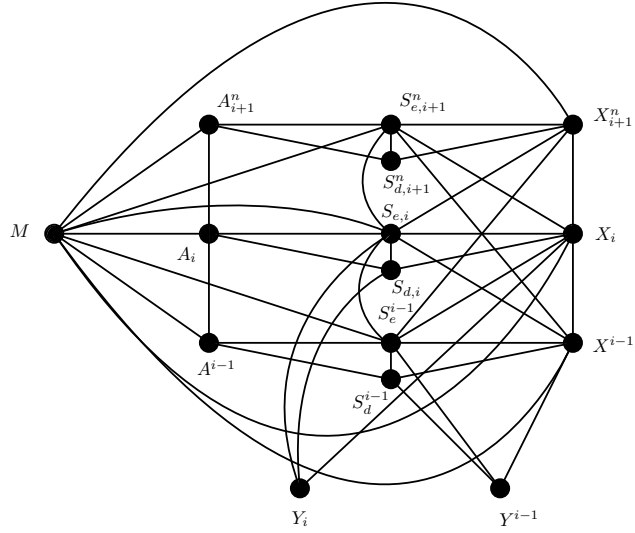}
    \caption{Graphical proof of the Markov chain $S_{d,i}-(A_{i},S_{e,i},X_{i})-(M,A^{n\setminus i},S_{e,i+1}^{n},X^{n},Y^{i-1},S_{d}^{i-1})$ with the marginal pmf derived from the joint pmf in \eqref{eq:pmf1} by summing out $(Y_{i+1}^{n},\hat{M},\hat{X}^{n}$).}\label{fig:markovchain}
\end{figure}

Since all paths in the graph from a node in $\mathcal{U}$ to a node in $\mathcal{W}$ pass through a node in $\mathcal{V}$, we have that $\mathcal{W}-\mathcal{V}-\mathcal{U}$ forms a Markov chain. Therefore, $S_{d,i}-(A_{i},S_{e,i},X_{i})-(M,A^{n\setminus i},S_{e,i+1}^{n},X^{n},Y^{i-1},S_{d}^{i-1})$ forms a Markov chain.
\end{proof}
Let us assume that a specific sequence of $(|\mathcal{M}^{(n)}|,n)$ codes exists such that the average error probabilities $P_{m,e}^{(n)}= \delta_{n}' \leq \delta_{n}$ and $P_{x,e}^{(n)}= \delta_{n}' \leq \delta_{n}$, and $ \log|\mathcal{M}^{(n)}| = n(R-\delta_{n}') \geq n(R-\delta_{n})$, with $ \lim_{n \rightarrow \infty} \delta_{n}= 0$.
Then standard properties of the entropy function give
\begin{align*}
  n(R-\delta_{n}) &\leq \log|\mathcal{M}^{(n)}|= H(M) \\
   %&= H(M)-H(X^{n},M|Y^{n})+ H(X^{n},M|Y^{n})\\
   &= H(M)-H(X^{n},M|Y^{n},S_{d}^{n})+ H(M|Y^{n},S_{d}^{n})+H(X^{n}|M,Y^{n},S_{d}^{n})\\
   &\leq H(M)-H(X^{n},M|Y^{n},S_{d}^{n})+ H(M|Y^{n},S_{d}^{n})+H(X^{n}|Y^{n},S_{d}^{n})
\end{align*}
Consider the last two terms in the above inequality. Similarly to \cite{Sumszyk2009}, by Fano's inequality, we get
\begin{align*}
H(M|Y^{n},S_{d}^{n}) &\leq h(\delta_{n}') + \delta_{n}'\cdot \log(2^{n(R-\delta_{n}')}-1) = n\epsilon^{(m)}_{n}, \\
H(X^{n}|Y^{n},S_{d}^{n}) &\leq h(\delta_{n}') + \delta_{n}'\cdot \log(|\mathcal{X}|^{n}-1) = n\epsilon^{(x)}_{n},
\end{align*}
where $h(\cdot)$ is the binary entropy function, and $\epsilon^{(m)}_{n}\rightarrow 0, \epsilon^{(x)}_{n} \rightarrow 0$ as $ \delta_{n} \rightarrow 0$.

Let $n\epsilon^{(m)}_{n} + n\epsilon^{(x)}_{n} \triangleq n\epsilon_{n}$, where $\epsilon_{n}$ satisfies $\lim_{n\rightarrow \infty}\epsilon_{n} =0$. Now we continue the chain of inequalities and get
\begin{align*}
n(R-\delta_{n})  &\leq H(M)-H(X^{n},M|Y^{n},S_{d}^{n})+n\epsilon_{n}\\
    &= H(M,S^{n}_{e})-H(S^{n}_{e}|M)-H(X^{n},M|Y^{n},S_{d}^{n})+n\epsilon_{n}\\
    &\overset{(a)}{=} H(M,S^{n}_{e},X^{n})-H(S^{n}_{e}|M,A^{n}) -H(X^{n},M|Y^{n},S_{d}^{n})+n\epsilon_{n}\\
    &\overset{(b)}{=} H(M,S^{n}_{e},X^{n})-H(S^{n}_{e}|A^{n})-H(X^{n},M|Y^{n},S_{d}^{n})+n\epsilon_{n}\\
    %&= H(X^{n},M)+H(S^{n}_{e}|X^{n},M)-H(S^{n}_{e}|A^{n})\\ &\qquad -H(X^{n},M|Y^{n})+n\epsilon_{n}\\
    &= I(X^{n},M;Y^{n},S_{d}^{n})-H(S^{n}_{e}|A^{n})+H(S^{n}_{e}|X^{n},M)+n\epsilon_{n}\\
    &\overset{(c)}{=} I(X^{n},M;Y^{n},S_{d}^{n})-H(S^{n}_{e}|A^{n})+H(S^{n}_{e}|X^{n},M,A^{n})+n\epsilon_{n}\\
    &= I(X^{n},M;Y^{n},S_{d}^{n})- I(X^{n},M;S^{n}_{e}|A^{n})+n\epsilon_{n},
\end{align*}
where $(a)$ follows from the fact that $X^{n}=f^{(n)}(M,S^{n}_{e})$ and $A^{n}=f^{(n)}_{a}(M)$, $(b)$ holds since $S_{e}^{n}$ is independent of $M$ given $A^{n}$, and $(c)$ from $A^{n}=f_{a}^{(n)}(M)$.

Continuing the chain of inequalities, we get
\begin{align*}
& n(R-\delta_{n}-\epsilon_{n})\\
& \leq\sum_{i=1}^{n} I(X^{n},M;Y_{i},S_{d,i}|Y^{i-1},S_{d}^{i-1})-I(X^{n},M;S_{e,i}|S_{e,i+1}^{n},A^{n})\\
& = \sum_{i=1}^{n} [I(X^{n},M,S_{e,i+1}^{n},A^{n};Y_{i},S_{d,i}|Y^{i-1},S_{d}^{i-1})  -I(S_{e,i+1}^{n},A^{n};Y_{i},S_{d,i}|X^{n},M,Y^{i-1},S_{d}^{i-1})]\\
& \qquad -[I(X^{n},M,Y^{i-1},S_{d}^{i-1};S_{e,i}|S_{e,i+1}^{n},A^{n}) -I(Y^{i-1},S_{d}^{i-1};S_{e,i}|X^{n},M,S_{e,i+1}^{n},A^{n})]\\
& \overset{(a)}{=} \sum_{i=1}^{n} I(X^{n},M,S_{e,i+1}^{n},A^{n};Y_{i},S_{d,i}|Y^{i-1},S_{d}^{i-1})- I(X^{n},M,Y^{i-1},S_{d}^{i-1};S_{e,i}|S_{e,i+1}^{n},A^{n})\\
& = \sum_{i=1}^{n} [H(Y_{i},S_{d,i}|Y^{i-1},S_{d}^{i-1})-H(Y_{i},S_{d,i}|Y^{i-1},S_{d}^{i-1},X^{n},M,S_{e,i+1}^{n},A^{n})]\\
& \qquad -[H(S_{e,i}|S_{e,i+1}^{n},A^{n}) -H(S_{e,i}|S_{e,i+1}^{n},A^{n},X^{n},M,Y^{i-1},S_{d}^{i-1})]\\
& \overset{(b)}{\leq} \sum_{i=1}^{n}H(Y_{i},S_{d,i})-H(Y_{i},S_{d,i}|Z_{i},A_{i})-H(S_{e,i}|A_{i}) +H(S_{e,i}|Z_{i},A_{i})\\
& = \sum_{i=1}^{n}I(A_{i},Z_{i};Y_{i},S_{d,i})-I(Z_{i};S_{e,i}|A_{i})\\
& \overset{(c)}{=} \sum_{i=1}^{n}I(A_{i},Z_{i},X_{i};Y_{i},S_{d,i})-I(Z_{i},X_{i};S_{e,i}|A_{i})\\
& = \sum_{i=1}^{n}I(A_{i};Y_{i},S_{d,i})+I(Z_{i},X_{i};Y_{i},S_{d,i}|A_{i})-I(Z_{i},X_{i};S_{e,i}|A_{i}),
\end{align*}
where
$(a)$ follows from the Csisz\'{a}r's sum identity in \cite{CsiszarBook}, $\sum_{i=1}^{n} I(S_{e,i+1}^{n},A^{n};Y_{i},S_{d,i}|X^{n},M,Y^{i-1},S_{d}^{i-1})-I(Y^{i-1},S_{d}^{i-1};S_{e,i}|X^{n},M,S_{e,i+1}^{n},A^{n})=0$ and additionally using $A^{n}=f^{(n)}_{a}(M)$,  $(b)$ follows from the fact that $(S_{e,i+1}^{n},A^{n\setminus i})-A_{i}-S_{e,i}$ forms a Markov chain and by defining $Z_{i}\triangleq (M,A^{n\setminus i},S_{e,i+1}^{n},X^{n},Y^{i-1},S_{d}^{i-1})$, and $(c)$ follows from the definition of $Z_{i}$.

Consider the sum of the last two terms,
\begin{align}
&  \sum_{i=1}^{n}I(Z_{i},X_{i};Y_{i},S_{d,i}|A_{i})-I(Z_{i},X_{i};S_{e,i}|A_{i})\nonumber \\
%& = \sum_{i=1}^{n} [I(Z_{i},X_{i};Y_{i},S_{e,i}|A_{i})-I(Z_{i},X_{i};S_{e,i}|A_{i},Y_{i})]\\ &\qquad -[I(Z_{i},X_{i};Y_{i},S_{e,i}|A_{i})-I(Z_{i},X_{i};Y_{i}|A_{i},S_{e,i})]\\
& \overset{(a)}{=} \sum_{i=1}^{n}I(Z_{i},X_{i};Y_{i},S_{d,i}|A_{i},S_{e,i})-I(Z_{i},X_{i};S_{e,i}|A_{i},Y_{i},S_{d,i})\nonumber\\
& = \sum_{i=1}^{n} H(Y_{i},S_{d,i}|A_{i},S_{e,i})-H(Y_{i},S_{d,i}|A_{i},S_{e,i},Z_{i},X_{i}) -I(X_{i};S_{e,i}|A_{i},Y_{i},S_{d,i})-I(Z_{i};S_{e,i}|A_{i},Y_{i},S_{d,i},X_{i})\nonumber\\
& \leq \sum_{i=1}^{n} H(Y_{i},S_{d,i}|A_{i},S_{e,i})-H(S_{d,i}|A_{i},S_{e,i},Z_{i},X_{i}) -H(Y_{i}|A_{i},S_{e,i},Z_{i},X_{i},S_{d,i})-I(X_{i};S_{e,i}|A_{i},Y_{i},S_{d,i})\nonumber\\
& \overset{(b)}{=} \sum_{i=1}^{n} H(Y_{i},S_{d,i}|A_{i},S_{e,i})-H(S_{d,i}|A_{i},S_{e,i},X_{i})-H(Y_{i}|A_{i},S_{e,i},X_{i},S_{d,i}) -I(X_{i};S_{e,i}|A_{i},Y_{i},S_{d,i})\nonumber\\
& = \sum_{i=1}^{n} I(Y_{i},S_{d,i};X_{i}|A_{i},S_{e,i})-I(X_{i};S_{e,i}|A_{i},Y_{i},S_{d,i})\nonumber\\
%& = \sum_{i=1}^{n} [I(Y_{i},S_{e,i};X_{i}|A_{i})-I(S_{e,i};X_{i}|A_{i})] \\ &\qquad -[I(X_{i};Y_{i},S_{e,i}|A_{i})-I(X_{i};Y_{i}|A_{i})]\\
& \overset{(c)}{=} \sum_{i=1}^{n} I(X_{i};Y_{i},S_{d,i}|A_{i})-I(X_{i};S_{e,i}|A_{i}), \label{eq:duplicate}
\end{align}
where $(a)$ follows by adding and subtracting the term $\sum_{i=1}^{n}I(Z_{i},X_{i};Y_{i},S_{d,i},S_{e,i}|A_{i})$, $(b)$ follows from the memoryless property of the channel where $(Z_{i},A_{i})-(X_{i},S_{e,i},S_{d,i})-Y_{i}$ forms the Markov chain and also the Markov chain $S_{d,i}-(A_{i},S_{e,i},X_{i})-Z_{i}$ (see Lemma \ref{lemma:markovchain1}), and $(c)$ follows by adding and subtracting the term $\sum_{i=1}^{n}I(Y_{i},S_{e,i},S_{d,i};X_{i}|A_{i})$.
Finally, we get
\begin{align}
n(R-\delta_{n}-\epsilon_{n}) &\leq \sum_{i=1}^{n} I(A_{i};Y_{i},S_{d,i})+I(X_{i};Y_{i},S_{d,i}|A_{i})-I(X_{i};S_{e,i}|A_{i}) \label{eq:converseR}
%& = \sum_{i=1}^{n} I(A_{i},X_{i};Y_{i})-I(X_{i};S_{e,i}|A_{i})\\
%& \leq n\cdot \max_{i}[I(A_{i},X_{i};Y_{i},S_{d,i})-I(X_{i};S_{e,i}|A_{i})].
\end{align}

%%%%%%%%%%%%%%%%%%%%%%%%%%%%%%%%%%%%%%%%%%%%%%%%%%%%%%%%%%%%%%%%%%%%%%%%%%%%%%%%%%%%%%%%%%%%%%%%%%%%%%%%%%%%%%%%%%%%
Next we prove the constraint which does not involve rate of the communication. It can be considered as the restriction imposed on the set of input distribution in a similar flavor as the dependence balance bound in \cite{Hekstra1989}. From the standard properties of the entropy function, we observe that
\begin{align*}
  0 &\leq H(M|A^{n}) \\
  &= H(M|A^{n})-H(X^{n},M|Y^{n},S_{d}^{n},A^{n})+ H(M|Y^{n},S_{d}^{n},A^{n}) +H(X^{n}|M,Y^{n},S_{d}^{n},A^{n})\\
  &\leq H(M|A^{n})-H(X^{n},M|Y^{n},S_{d}^{n},A^{n})+ H(M|Y^{n},S_{d}^{n}) +H(X^{n}|Y^{n},S_{d}^{n})
\end{align*}
Again, consider the last two terms in the above inequality. By Fano's inequality, we get
\begin{align*}
H(M|Y^{n},S_{d}^{n}) &\leq h(\delta_{n}') + \delta_{n}'\cdot \log(2^{n(R-\delta_{n}')}-1) = n\epsilon^{(m)}_{n}, \\
H(X^{n}|Y^{n},S_{d}^{n}) &\leq h(\delta_{n}') + \delta_{n}'\cdot \log(|\mathcal{X}|^{n}-1) = n\epsilon^{(x)}_{n}.
\end{align*}
Let $n\epsilon^{(m)}_{n} + n\epsilon^{(x)}_{n} \triangleq n\epsilon_{n}$, where $\epsilon_{n}$ satisfies $\lim_{n\rightarrow \infty}\epsilon_{n} =0$. Now we continue the chain of inequalities and get
\begin{align*}
-n\epsilon_{n}  &\leq H(M|A^{n})-H(X^{n},M|Y^{n},S_{d}^{n},A^{n}) \\
    &= H(M,S^{n}_{e}|A^{n})-H(S^{n}_{e}|M,A^{n})-H(X^{n},M|Y^{n},S_{d}^{n},A^{n})\\
    &\overset{(a)}{=} H(M,S^{n}_{e},X^{n}|A^{n})-H(S^{n}_{e}|M,A^{n}) -H(X^{n},M|Y^{n},S_{d}^{n},A^{n})\\
    &\overset{(b)}{=} H(M,S^{n}_{e},X^{n}|A^{n})-H(S^{n}_{e}|A^{n})-H(X^{n},M|Y^{n},S_{d}^{n},A^{n})\\
    %&= H(X^{n},M|A^{n})+H(S^{n}_{e}|X^{n},M,A^{n})-H(S^{n}_{e}|A^{n})\\ &\qquad -H(X^{n},M|Y^{n},A^{n})\\
    &= I(X^{n},M;Y^{n},S_{d}^{n}|A^{n})-I(X^{n},M;S^{n}_{e}|A^{n})\\
    & = \sum_{i=1}^{n} I(X^{n},M;Y_{i},S_{d,i}|Y^{i-1},S_{d}^{i-1},A^{n})-I(X^{n},M;S_{e,i}|S_{e,i+1}^{n},A^{n})\\
%& = \sum_{i=1}^{n} H(Y_{i}|Y^{i-1},A^{n})- H(Y_{i}|Y^{i-1},A^{n},X^{n},M) \\ & \qquad -I(X^{n},M;S_{e,i}|S_{e,i+1}^{n},A^{n})\\
& \overset{(c)}{\leq} \sum_{i=1}^{n} H(Y_{i},S_{d,i}|Y^{i-1},S_{d}^{i-1},A_{i})- H(Y_{i},S_{d,i}|Y^{i-1},S_{d}^{i-1},A_{i},X^{n},M)-I(X^{n},M;S_{e,i}|S_{e,i+1}^{n},A^{n})\\
& = \sum_{i=1}^{n} I(X^{n},M;Y_{i},S_{d,i}|Y^{i-1},S_{d}^{i-1},A_{i})-I(X^{n},M;S_{e,i}|S_{e,i+1}^{n},A^{n})\\
& \overset{(d)}{=}\sum_{i=1}^{n} [I(X^{n},M,S_{e,i+1}^{n},A^{n\setminus i};Y_{i},S_{d,i}|Y^{i-1},S_{d}^{i-1},A_{i}) -I(S_{e,i+1}^{n},A^{n\setminus i};Y_{i},S_{d,i}|X^{n},M,Y^{i-1},S_{d}^{i-1},A_{i})]\\
& \qquad -[I(X^{n},M,Y^{i-1},S_{d}^{i-1};S_{e,i}|S_{e,i+1}^{n},A^{n}) -I(Y^{i-1},S_{d}^{i-1};S_{e,i}|X^{n},M,S_{e,i+1}^{n},A^{n})] \\ & \qquad- I(A_{i};Y_{i},S_{d,i}|X^{n},M,Y^{i-1},S_{d}^{i-1})
\end{align*}
where $(a)$ follows from the fact that $X^{n}=f^{(n)}(M,S^{n}_{e})$, $(b)$ holds since $S_{e}^{n}$ is independent of $M$ given $A^{n}$, $(c)$ holds since $A^{n} = f_{a}(M)$, and $(d)$ follows from the fact that the last term is zero since $A^{n}=f_{a}^{(n)}(M)$.

Continuing the chain of inequalities, we have
\begin{align}
 -n\epsilon_{n}
%& \leq \sum_{i=1}^{n} I(X^{n},M;Y_{i}|Y^{i-1},A^{n})-I(X^{n},M;S_{e,i}|S_{e,i+1}^{n},A^{n})\\
%& = \sum_{i=1}^{n} H(Y_{i}|Y^{i-1},A^{n})- H(Y_{i}|Y^{i-1},A^{n},X^{n},M) \\ & \qquad -I(X^{n},M;S_{e,i}|S_{e,i+1}^{n},A^{n})\\
%& \overset{(a)}{\leq} \sum_{i=1}^{n} H(Y_{i}|Y^{i-1},A_{i})- H(Y_{i}|Y^{i-1},A_{i},X^{n},M)\\ & \qquad -I(X^{n},M;S_{e,i}|S_{e,i+1}^{n},A^{n})\\
%& = \sum_{i=1}^{n} I(X^{n},M;Y_{i}|Y^{i-1},A_{i})-I(X^{n},M;S_{e,i}|S_{e,i+1}^{n},A^{n})\\
%& \overset{(b)}{=}\sum_{i=1}^{n} [I(X^{n},M,S_{e,i+1}^{n},A^{n\setminus i};Y_{i}|Y^{i-1},A_{i}) \\ &\qquad -I(S_{e,i+1}^{n},A^{n\setminus i};Y_{i}|X^{n},M,Y^{i-1},A_{i})]\\
%& \qquad -[I(X^{n},M,Y^{i-1};S_{e,i}|S_{e,i+1}^{n},A^{n}) \\ &\qquad \qquad -I(Y^{i-1};S_{e,i}|X^{n},M,S_{e,i+1}^{n},A^{n})] \\ & \qquad- I(A_{i};Y_{i}|X^{n},M,Y^{i-1})\\
 &\leq \sum_{i=1}^{n} [I(X^{n},M,S_{e,i+1}^{n},A^{n\setminus i};Y_{i},S_{d,i}|Y^{i-1},S_{d}^{i-1},A_{i}) -I(S_{e,i+1}^{n},A^{n};Y_{i},S_{d,i}|X^{n},M,Y^{i-1},S_{d}^{i-1})]\nonumber\\
& \qquad -[I(X^{n},M,Y^{i-1},S_{d}^{i-1};S_{e,i}|S_{e,i+1}^{n},A^{n})  -I(Y^{i-1},S_{d}^{i-1};S_{e,i}|X^{n},M,S_{e,i+1}^{n},A^{n})]\nonumber \\
& \overset{(a)}{=} \sum_{i=1}^{n} I(X^{n},M,S_{e,i+1}^{n},A^{n\setminus i};Y_{i},S_{d,i}|Y^{i-1},S_{d}^{i-1},A_{i})- I(X^{n},M,Y^{i-1},S_{d}^{i-1};S_{e,i}|S_{e,i+1}^{n},A^{n})\nonumber\\
& = \sum_{i=1}^{n} [H(Y_{i},S_{d,i}|Y^{i-1},S_{d}^{i-1},A_{i})-H(Y_{i},S_{d,i}|Y^{i-1},S_{d}^{i-1},X^{n},M,S_{e,i+1}^{n},A^{n})]\nonumber\\
& \qquad -[H(S_{e,i}|S_{e,i+1}^{n},A^{n}) -H(S_{e,i}|S_{e,i+1}^{n},A^{n},X^{n},M,Y^{i-1},S_{d}^{i-1})]\nonumber\\
& \overset{(b)}{\leq} \sum_{i=1}^{n}H(Y_{i},S_{d,i}|A_{i})-H(Y_{i},S_{d,i}|Z_{i},A_{i})-H(S_{e,i}|A_{i}) +H(S_{e,i}|Z_{i},A_{i})\nonumber\\
%& = \sum_{i=1}^{n}I(Z_{i};Y_{i},S_{d,i}|A_{i})-I(Z_{i};S_{e,i}|A_{i})\nonumber\\
& \overset{(c)}{=} \sum_{i=1}^{n}I(Z_{i},X_{i};Y_{i},S_{d,i}|A_{i})-I(Z_{i},X_{i};S_{e,i}|A_{i})\nonumber\\
%& \overset{(d)}{\leq} \sum_{i=1}^{n}I(Z_{i},X_{i};Y_{i}|A_{i},S_{e,i})-I(Z_{i},X_{i};S_{e,i}|A_{i},Y_{i})\nonumber\\
%%& = \sum_{i=1}^{n} H(Y_{i}|A_{i},S_{e,i})-H(Y_{i}|A_{i},S_{e,i},X_{i},Z_{i})\\ &\qquad -I(X_{i};S_{e,i}|A_{i},Y_{i})-I(Z_{i};S_{e,i}|A_{i},Y_{i},X_{i})\\
%& \leq \sum_{i=1}^{n} H(Y_{i}|A_{i},S_{e,i})-H(Y_{i}|A_{i},S_{e,i},X_{i},Z_{i})\nonumber\\ &\qquad -I(X_{i};S_{e,i}|A_{i},Y_{i})\nonumber\\
%%& \overset{(g)}{=} \sum_{i=1}^{n} H(Y_{i}|A_{i},S_{e,i})-H(Y_{i}|A_{i},S_{e,i},X_{i})\\ &\qquad -I(X_{i};S_{e,i}|A_{i},Y_{i})\\
%&  \overset{(e)}{=} \sum_{i=1}^{n} I(Y_{i};X_{i}|A_{i},S_{e,i})-I(X_{i};S_{e,i}|A_{i},Y_{i})\nonumber\\
& \overset{(d)}{\leq} \sum_{i=1}^{n} I(X_{i};Y_{i},S_{d,i}|A_{i})-I(X_{i};S_{e,i}|A_{i}),\label{eq:converse0}
\end{align}
where %$(a)$ follows from $A^{n} = f_{a}(M)$,
%$(b)$ follows from the fact that the last term is zero,
$(a)$ follows by the Csisz\'{a}r's sum identity, $\sum_{i=1}^{n} I(S_{e,i+1}^{n},A^{n};Y_{i}|X^{n},M,Y^{i-1})-I(Y^{i-1};S_{e,i}|X^{n},M,S_{e,i+1}^{n},A^{n})=0$ and $A^{n}=f_{a}^{(n)}(M)$,
$(b)$ holds by using the Markov chain $(S_{e,i+1}^{n},A^{n\setminus i})-A_{i}-S_{e,i}$ and by defining $Z_{i} \triangleq (M,A^{n\setminus i},S_{e,i+1}^{n},X^{n},Y^{i-1},S_{d}^{i-1})$,
$(c)$ follows from the definition of $Z_{i}$,
%$(d)$ follows by adding and subtracting the term $\sum_{i=1}^{n}I(Z_{i},X_{i};Y_{i},S_{e,i}|A_{i})$,
%$(e)$ follows from the memoryless property of the channel where $(Z_{i},A_{i})-(X_{i},S_{e,i})-Y_{i}$ forms a Markov chain and so does $Z_{i}-(X_{i},S_{e,i},A_{i})-Y_{i}$,
and $(d)$ follows from the same steps as in obtaining \eqref{eq:duplicate}. %by adding and subtracting the term $\sum_{i=1}^{n}I(Y_{i},S_{e,i};X_{i}|A_{i})$.

Let $Q$ be a random variable uniformly distributed over $\{1, \ldots, n\}$ and independent of $(M,A^n,X^n,S_{e}^n,S_{d}^n,Y^{n})$,
we can rewrite \eqref{eq:converseR} and \eqref{eq:converse0} as
\begin{align*}
R &\leq \frac{1}{n}\sum_{i=1}^{n} I(A_{i},X_{i};Y_{i},S_{d,i}|Q=i)-I(X_{i};S_{e,i}|A_{i},Q=i)+\delta_{n}+\epsilon_{n}\\
  &=  I(A_{Q},X_{Q};Y_{Q},S_{d,Q}|Q)-I(X_{Q};S_{e,Q}|A_{Q},Q)+\delta_{n}+\epsilon_{n}
\end{align*}
and
\begin{align*}
0 &\leq \frac{1}{n}\sum_{i=1}^{n} I(X_{i};Y_{i},S_{d,i}|A_{i},Q=i)-I(X_{i};S_{e,i}|A_{i},Q=i)+\epsilon_{n}\\
  &= I(X_{Q};Y_{Q},S_{d,Q}|A_{Q},Q)-I(X_{Q};S_{e,Q}|A_{Q},Q)+\epsilon_{n}.
\end{align*}

Now since we have that $P_{S_{e,Q},S_{d,Q}|A_{Q}}=P_{S_{e},S_{d}|A}$, $P_{Y_{Q}|X_{Q},S_{e,Q},S_{d,Q}}=P_{Y|X,S_{e},S_{d}}$, and $A_{Q}-(X_{Q},S_{e,Q},S_{d,Q})-Y_{Q}$ forms a Markov chain, we identify $A \triangleq A_{Q}$, $S_{e} \triangleq S_{e,Q}$, $S_{d} \triangleq S_{d,Q}$, $X \triangleq X_{Q}$, and $Y \triangleq Y_{Q}$ to finally obtain
\begin{align*}
    R &\leq I(A,X;Y,S_{d}|Q)-I(X;S_{e}|A,Q)+\delta_{n}+\epsilon_{n}\\
    \mbox{and} \ \ 0 &\leq I(X;Y,S_{d}|A,Q)-I(X;S_{e}|A,Q)+\epsilon_{n},
\end{align*}
for some joint distribution
\begin{align}
&P_{Q}(q)P_{A|Q}(a|q)P_{S_{e},S_{d}|A}(s_{e},s_{d}|a)P_{X|A,S_{e},Q}(x|a,s_{e},q) P_{Y|X,S_{e},S_{d}}(y|x,s_{e},s_{d}). \label{eq:pdf2}
\end{align}
\begin{lemma}\label{lemma:markovchain2}
From the joint distribution in \eqref{eq:pdf2}, $(Y,S_{d})-(X,A,S_{e})-Q$ and $S_{e}-A-Q$ form Markov chains.
\end{lemma}
\begin{proof}
We use a (partial) list of properties satisfied by the Markov chain (the conditional independence relation) in \cite{Pearl}. As a quick reference, we restate it in the following. Let $W,X,Y,Z$ be the random variables, and $``\Longrightarrow"$ refer to ``imply",
 \begin{align*}
% \mbox{symmetry}&: X-Y-Z \Longrightarrow Z-Y-X \\
% \mbox{decomposition}&: X-Z-(W,Y) \Longrightarrow  X-Z-Y\\
 \mbox{weak union}&: X-Z-(W,Y) \Longrightarrow  X-(Z,W)-Y\\
 \mbox{contraction}&: X-Z-Y \ \mbox{and}\ X-(Z,Y)-W \Longrightarrow  X-Z-(Y,W).
 %\mbox{intersection}*&: X-(Z,Y)-W \mbox{and} X-(Z,W)-Y \Longrightarrow  X-Z-(Y,W)
 \end{align*}

From \eqref{eq:pdf2}, the following Markov chains are readily derived.
 \begin{align}
 &Q-A-(S_{e},S_{d}) \label{eq:chain1} \\
 &X-(A,S_{e},Q)-S_{d}\label{eq:chain2} \\
 &(A,Q)-(X,S_{e},S_{d})-Y \label{eq:chain3}
 \end{align}
By the weak union property, we can derive from \eqref{eq:chain1} the Markov chain
$Q-(A,S_{e})-S_{d}$. Then combining it with \eqref{eq:chain2}, by using the contraction property, we get the Markov chain
\begin{equation}\label{eq:chain5}
(X,Q)-(A,S_{e})-S_{d}
\end{equation}

Again using the weak union in \eqref{eq:chain3} and \eqref{eq:chain5}, we get
 \begin{align*}
 &Q-(X,A,S_{e},S_{d})-Y \\
 \mbox{and} \ \ &Q-(X,A,S_{e})-S_{d}
 \end{align*}
Combining these two Markov chains using the contraction property, we finally get the Markov chain $Q-(X,A,S_{e})-(Y,S_{d})$.
\end{proof}

To this end, we note that under any distribution of the form above, we have
\begin{align*}
    I(A,X;Y,S_{d}|Q)-I(X;S_{e}|A,Q) &= I(A,X,S_{e};Y,S_{d}|Q)-I(X,Y,S_{d};S_{e}|A,Q)\\
                                &= H(Y,S_{d}|Q)-H(Y,S_{d}|A,X,S_{e},Q)-H(S_{e}|A,Q) +H(S_{e}|X,Y,S_{d},A,Q)\\
                                & \leq H(Y,S_{d})-H(Y,S_{d}|A,X,S_{e},Q)-H(S_{e}|A,Q)+H(S_{e}|X,Y,S_{d},A)\\
                                & \overset{(*)}{=} H(Y,S_{d})-H(Y,S_{d}|A,X,S_{e})-H(S_{e}|A)+H(S_{e}|X,Y,S_{d},A)\\
                                &= I(A,X,S_{e};Y,S_{d})-I(X,Y;S_{e}|A)\\
                                &=I(A,X;Y,S_{d})-I(X;S_{e}|A),
\end{align*}
and
\begin{align*}
I(X;Y,S_{d}|A,Q)-I(X;S_{e}|A,Q) &= I(X,S_{e};Y,S_{d}|A,Q)-I(X,Y;S_{e}|A,Q)\\
                                &= H(Y,S_{d}|A,Q)-H(Y,S_{d}|A,X,S_{e},Q)-H(S_{e}|A,Q) +H(S_{e}|X,Y,S_{d},A,Q)\\
                                & \leq H(Y,S_{d}|A)-H(Y,S_{d}|A,X,S_{e},Q)-H(S_{e}|A,Q)+H(S_{e}|X,Y,S_{d},A)\\
                                & \overset{(*)}{=} H(Y,S_{d}|A)-H(Y,S_{d}|A,X,S_{e})-H(S_{e}|A)+H(S_{e}|X,Y,S_{d},A)\\
                                &= I(X,S_{e};Y,S_{d}|A)-I(X,Y;S_{e}|A)\\
                                &=I(X;Y,S_{d}|A)-I(X;S_{e}|A),
\end{align*}
where both equalities $(*)$ follows from the Markov chains $(Y,S_{d})-(X,A,S_{e})-Q$ and $S_{e}-A-Q$ (derived from \eqref{eq:pdf2}, see Lemma \ref{lemma:markovchain2}), and the joint distribution of $(A,S_{e},S_{d},X,Y)$ is of the form
\begin{align*}
&\sum_{q \in \mathcal{Q}}P_{Q}(q)P_{A|Q}(a|q)P_{S_{e},S_{d}|A}(s_{e},s_{d}|a)P_{X|A,S_{e},Q}(x|a,s_{e},q)P_{Y|X,S_{e},S_{d}}(y|x,s_{e},s_{d})\\ & \qquad=P_{A}(a)P_{S_{e},S_{d}|A}(s_{e},s_{d}|a)P_{X|A,S_{e}}(x|a,s_{e})P_{Y|X,S_{e},S_{d}}(y|x,s_{e},s_{d}).
\end{align*}

The proof is concluded by taking the limit $n \rightarrow \infty$. \hfill$\blacksquare$
%%%%%%%%%%%%%%%%%%%%%%%%%%%%%%%%%%%%%%%%%%%%%%%%%%%%%%%%%%%%%%%%%%%%%%%%%%%%%%%%%%%%%%%%%%%%%%%%%%%%%%%%%%%
\section{Proof of convexity of the region $\mathcal{R}_{\text{mod}}$ with dummy variable $\tilde{R}$}\label{sec:proof_convex_dummy}
Consider the achievable rate $0 \leq R \leq I(A,X;Y,S_{d})-I(X;S_{e}|A)$ for some $P_{A}(a), P_{X|A,S_{e}}(x|a,s_{e})$ such that $0 \leq I(X;Y,S_{d}|A)-I(X;S_{e}|A)$. We modify it by introducing a dummy variable $\tilde{R}$ which can take either positive or negative value, and we obtain the modified ``region". The modified region $\mathcal{R}_{\text{mod}}$ is the set
\begin{align*}
\mathcal{R}_{\text{mod}} = \{ (R,\tilde{R}): 0 \leq &R \leq I(A,X;Y,S_{d})-I(X;S_{e}|A)\\
&\tilde{R} \leq I(X;Y,S_{d}|A)-I(X;S_{e}|A)\\
\mbox{for some}\ P_{A}(a)&P_{S_{e},S_{d}|A}(s_{e},s_{d}|a)P_{X|A,S_{e}}(x|a,s_{e})P_{Y|X,S_{e},S_{d}}(y|x,s_{e},s_{d}). \}
\end{align*}

We will show that the region above is convex.
Assuming that any two arbitrary points $(R^{1},\tilde{R}^{1})$ and $(R^{2},\tilde{R}^{2}) \in \mathcal{R}_{\text{mod}}$. This implies that there exist distributions
\begin{align*}
P^{(1)}_{A,S_{e},S_{d},X,Y}(a,s_{e},s_{d},x,y) &= P_{A}^{(1)}(a)P_{S_{e},S_{d}|A}(s_{e},s_{d}|a)P_{X|A,S_{e}}^{(1)}(x|a,s_{e})P_{Y|X,S_{e},S_{d}}(y|x,s_{e},s_{d})\\
\mbox{and}\ P^{(2)}_{A,S_{e},S_{d},X,Y}(a,s_{e},s_{d},x,y) &= P_{A}^{(2)}(a)P_{S_{e},S_{d}|A}(s_{e},s_{d}|a)P_{X|A,S_{e}}^{(2)}(x|a,s_{e})P_{Y|X,S_{e},S_{d}}(y|x,s_{e},s_{d})
\end{align*}
such that
\begin{align}
0 \leq &R^{(1)} \leq I^{(1)}(A,X;Y,S_{d})-I^{(1)}(X;S_{e}|A) \nonumber \\
&\tilde{R}^{(1)} \leq I^{(1)}(X;Y,S_{d}|A)-I^{(1)}(X;S_{e}|A)\nonumber \\
\mbox{and}\ 0 \leq &R^{(2)} \leq I^{(2)}(A,X;Y,S_{d})-I^{(2)}(X;S_{e}|A)\nonumber\\
&\tilde{R}^{(2)} \leq I^{(2)}(X;Y,S_{d}|A)-I^{(2)}(X;S_{e}|A), \label{eq:RandR'}
\end{align}
where $I^{(i)}(\cdot)$ denotes the mutual information associated with $P^{(i)}_{A,S_{e},S_{d},X,Y}(a,s_{e},s_{d},x,y)$, $i=1,2$.

Now let $Q$ be an independent random variable taking value from $\{1,2\}$, where $\mathrm{Pr}(Q=1)=1-\mathrm{Pr}(Q=2)=\alpha, 0 \leq \alpha \leq 1$. %Furthermore we define $A \triangleq A^{Q}$, $S_{e} \triangleq S_{e}^{Q}$, $S_{d} \triangleq S_{d}^{Q}$, $X \triangleq X^{Q}$, $Y \triangleq Y^{Q}$.
Then we have the joint distribution
\begin{align}
P_{Q,A,S_{e},S_{d},X,Y}(q,a,s_{e},s_{d},x,y) &= P_{Q}(q)P_{A|Q}(a|q)P_{S_{e},S_{d}|A}(s_{e},s_{d}|a)P_{X|A,S_{e},Q}(x|a,s_{e},q)P_{Y|X,S_{e},S_{d}}(y|x,s_{e},s_{d}) \label{eq:joint_pdfQ}
\end{align}
where $P_{A|Q}(a|q) \triangleq P^{(q)}_{A}(a)$ and $P_{X|A,S_{e},Q}(x|a,s_{e},q) \triangleq P^{(q)}_{X|A,S_{e}}(x|a,s_{e})$ for $q=1,2$.

Consider now the marginal distribution (averaged over $Q$)
\begin{align*}
P_{A,S_{e},S_{d},X,Y}(a,s_{e},s_{d},x,y) &= \sum_{q=1,2}P_{Q}(q)P_{A|Q}(a|q)P_{S_{e},S_{d}|A}(s_{e},s_{d}|a)P_{X|A,S_{e},Q}(x|a,s_{e},q)P_{Y|X,S_{e},S_{d}}(y|x,s_{e},s_{d})
\end{align*}
which is associated with the mutual information terms $I(A,X;Y,S_{d})-I(X;S_{e}|A)$ and $I(X;Y,S_{d}|A)-I(X;S_{e}|A)$. It follows that
\begin{align}
&I(A,X;Y,S_{d})-I(X;S_{e}|A) = I(A,X,S_{e};Y,S_{d})-I(X,Y,S_{d};S_{e}|A)\nonumber \\
& \qquad \geq  I(A,X,S_{e};Y,S_{d}|Q)-I(X,Y,S_{d};S_{e}|A,Q)\nonumber \\
& \qquad= \alpha [I^{(1)}(A,X,S_{e};Y,S_{d})-I^{(1)}(X,Y,S_{d};S_{e}|A)] + (1-\alpha)[I^{(2)}(A,X,S_{e};Y,S_{d})-I^{(2)}(X,Y,S_{d};S_{e}|A)]\nonumber \\
& \qquad= \alpha [I^{(1)}(A,X;Y,S_{d})-I^{(1)}(X;S_{e}|A)] + (1-\alpha)[I^{(2)}(A,X;Y,S_{d})-I^{(2)}(X;S_{e}|A)], \label{eq:convex_comb0}
\end{align}
and
\begin{align}
&I(X;Y,S_{d}|A)-I(X;S_{e}|A) = I(X,S_{e};Y,S_{d}|A)-I(X,Y,S_{d};S_{e}|A)\nonumber \\
&\qquad\geq I(X,S_{e};Y,S_{d}|A,Q)-I(X,Y,S_{d};S_{e}|A,Q)\nonumber \\
& \qquad= \alpha [I^{(1)}(X,S_{e};Y,S_{d}|A)-I^{(1)}(X,Y,S_{d};S_{e}|A)] + (1-\alpha)[I^{(2)}(X,S_{e};Y,S_{d}|A)-I^{(2)}(X,Y,S_{d};S_{e}|A)]\nonumber\\
& \qquad= \alpha [I^{(1)}(X;Y,S_{d}|A)-I^{(1)}(X;S_{e}|A)] + (1-\alpha)[I^{(2)}(X;Y,S_{d}|A)-I^{(2)}(X;S_{e}|A)], \label{eq:convex_comb1}
\end{align}
where both inequalities follow from $(Y,S_{d})-(X,A,S_{e})-Q$ and $S_{e}-A-Q$ obtained in Lemma \ref{lemma:markovchain2}.

From \eqref{eq:RandR'}, \eqref{eq:convex_comb0}, and \eqref{eq:convex_comb1}, it follows that there exists a distribution $P_{Q,A,S_{e},S_{d},X,Y}(q,a,s_{e},s_{d},x,y)$ as in \eqref{eq:joint_pdfQ} with marginal factorized as $P_{A}(a)P_{S_{e},S_{d}|A}(s_{e},s_{d}|a)P_{X|A,S_{e}}(x|a,s_{e})P_{Y|X,S_{e},S_{d}}(y|x,s_{e},s_{d})$ such that
\begin{align}
I(A,X;Y,S_{d})-I(X;S_{e}|A) &\geq \alpha R^{(1)} + (1-\alpha) R^{(2)} \geq 0 \nonumber \\
I(X;Y,S_{d}|A)-I(X;S_{e}|A) &\geq \alpha \tilde{R}^{(1)} + (1-\alpha) \tilde{R}^{(2)} \label{eq:convex_comb2}
\end{align}

By the definition of $\mathcal{R}_{\text{mod}}$ and \eqref{eq:convex_comb2}, we have that
\begin{align*}
(\alpha R^{1} + (1-\alpha) R^{2},\tilde{R}^{1} + (1-\alpha) \tilde{R}^{2}) \in \mathcal{R}_{\text{mod}}.
\end{align*}
This implies that any convex combination of points $(R,\tilde{R}) \in
\mathcal{R}_{\text{mod}}$ is also in the set
$\mathcal{R}_{\text{mod}}$, and thus $\mathcal{R}_{\text{mod}}$ is
convex. \hfill$\blacksquare$

%%%%%%%%%%%%%%%%%%%%%%%%%%%%%%%%%%%%%%%%%%%%%%%%%%%%%%%%%%%%%%%%%%%%%%%%%%%%%%%%%%%%%%%%%%%%%%%%%%%%%%%%%%
\section{Proof of Proposition \ref{eq:C_Se}}\label{sec:proof_C_se}
\subsection{Achievability Proof of Proposition \ref{eq:C_Se}}
Similarly to the previous achievability proof, the proof follows from a standard random coding argument where we use the definition and properties of  $\epsilon$-typicality as in \cite{ElGamalKim}. We use the technique of rate splitting, i.e., the message $M$ of rate $R$ is split into two messages $M_{1}$ and $M_{2}$ of rates $R_{1}$ and $R_{2}$. Two-stage coding is then considered, i.e., a first stage for communicating the identity of the action sequence, and a second stage for communicating the identity of $S_{e}^{n}$ based on the known action sequence.

For given channels with transition probabilities
$P_{S_{e},S_{d}|A}(s_{e},s_{d}|a)$ and $P_{Y|X,S_{e},S_{d}}(y|x,s_{e},s_{d})$ we can assign the joint
probability to any random vector $(A,X,S_{e})$ by
\begin{align*}
P_{A,S_{e},S_{d},X,Y}(a,s_{e},s_{d},x,y)%&= P_{A,S_{e},X}(a,s_{e},x)P_{Y|X,S_{e}}(y|x,s_{e}) \\
  & =P_{A}(a)P_{S_{e},S_{d}|A}(s_{e},s_{d}|a)P_{X|A,S_{e}}(x|a,s_{e})P_{Y|X,S_{e},S_{d}}(y|x,s_{e},s_{d})
\end{align*}

\textit{Codebook Generation}: Fix $P_{A}$ and $P_{X|A,S_{e}}$. Let $\mathcal{M}_{1}^{(n)}=\{1,2,\ldots,|\mathcal{M}_{1}^{(n)}|\}$, $\mathcal{M}_{2}^{(n)}=\{1,2,\ldots,|\mathcal{M}_{2}^{(n)}|\}$ and $\mathcal{J}^{(n)}=\{1,2,\ldots,|\mathcal{J}^{(n)}|\}$. For all $m_{1}\in \mathcal{M}_{1}^{(n)}$, generate $a^{n}(m_{1})$ i.i.d. according to $\prod_{i=1}^{n}P_{A}(a_i)$. Then for each $m_{1}$, generate $|\mathcal{M}_{2}^{(n)}||\mathcal{J}^{(n)}|$ codewords $\{\check{s}_{e}^{n}(m_{1},m_{2},j)\}_{m_{2} \in \mathcal{M}_{2}^{(n)}, j \in \mathcal{J}^{(n)}}$ i.i.d. each according to $\prod_{i=1}^{n}P_{S_{e}|A}(\check{s}_{e,i}|a_{i}(m_1))$. Finally, for each  $(a^{n},\check{s}_{e}^{n})$ pair, generate $x^{n}$ i.i.d. according to $\prod_{i=1}^{n}P_{X|S_{e},A}(x_i|\check{s}_{e,i},a_i(m_1))$. Then the codebooks are revealed to the action encoder, the channel encoder, and the decoder. Let $0 < \epsilon_{0} <\epsilon_{1}<\epsilon < 1$.

\textit{Encoding}: Given the message $m=(m_{1},m_{2})\in \mathcal{M}^{(n)}$, the action codeword $a^{n}(m_{1})$ is chosen and the channel state information $(s_{e}^{n},s_{d}^{n})$ is generated as an output of the memoryless channel, $P_{S_{e}^{n},S_{d}^{n}|A^{n}}(s_{e}^{n},s_{d}^{n}|a^{n})=\prod_{i=1}^{n}P_{S_{e},S_{d}|A}(s_{e,i},s_{d,i}|a_{i})$.
The encoder looks for the smallest value of $j\in \mathcal{J}^{(n)}$ such that $\check{s}_{e}^{n}(m_{1},m_{2},j)=s_{e}^{n}$. The channel input sequence is then chosen to be $x^{n}(m_{1},m_{2},j)$. If no such $j$ exists, set $j=1$.

\textit{Decoding}:  Upon receiving $y^{n}$ and $s_{d}^{n}$, the decoder in the first step looks for the smallest $\tilde{m}_{1} \in \mathcal{M}_{1}^{(n)}$ such that $\big(y^{n},s_{d}^{n},a^{n}(\tilde{m}_{1})\big)\in T_{\epsilon}^{(n)}(Y,S_{d},A)$. If successful, then set $\hat{m}_{1}=\tilde{m}_{1}$. Otherwise, set $\hat{m}_{1}=1$. Then, based on the known $a^{n}(\hat{m}_{1})$, the decoder looks for a pair $(\tilde{m}_{2},\tilde{j})$ with the smallest $\tilde{m}_{2} \in \mathcal{M}_{2}^{(n)}$ and $\tilde{j} \in \mathcal{J}^{(n)}$ such that $\big(y^{n},s_{d}^{n}, a^{n}(\hat{m}_{1}),\check{s}_{e}^{n}(\hat{m}_{1},\tilde{m}_{2},\tilde{j}),x^{n}(\hat{m}_{1},\tilde{m}_{2},\tilde{j})\big)\in T_{\epsilon}^{(n)}(Y,S_{d},A,S_{e},X)$.
If there exists such a pair, the decoded message is set to be $\hat{m} =(\hat{m}_{1},\tilde{m}_{2})$, and the decoded state $\hat{s}_{e}^{n}=\check{s}_{e}^{n}(\hat{m}_{1},\tilde{m}_{2},\tilde{j})$. Otherwise, $\hat{m} =(1,1)$ and $\hat{s}_{e}^{n}=\check{s}_{e}^{n}(1,1,1)$.\footnotemark[4]

\footnotetext[4]{We note that although the simultaneous joint typicality decoding gives us different constraints on the individual rate as compared to the sequential two-stage decoding considered in this paper, it gives the same constraints on the total transmission rate in which we are interested.}

\textit{Analysis of Probability of Error}: Due to the symmetry of the random code construction, the error probability does not depend on which message was sent. Assuming that $M=(M_{1},M_{2})$ and $J$ were sent and chosen at the encoder. We define the error events as follows.
    \begin{align*}
      &\mathcal{E}_{1} = \{A^{n}(M_{1})\notin
    T_{\epsilon_{0}}^{(n)}(A)\} \\
      &\mathcal{E}_{2} = \big\{\big(S_{e}^{n},S_{d}^{n},A^{n}(M_{1})\big)\notin
    T_{\epsilon_{1}}^{(n)}(S_{e},S_{d},A)\big\} \\
      &\mathcal{E}_{3a} = \big\{S_{e}^{n} \neq  \check{S}_{e}^{n}(M_{1},M_{2},j) \ \mbox{for all}\ j \in \mathcal{J}^{(n)}  \big\}\\
      &\mathcal{E}_{3b} = \big\{ \big(S_{e}^{n},A^{n}(M_{1}),X^{n}(M_{1},M_{2},J)\big)\notin
    T_{\epsilon_{1}}^{(n)}(S_{e},A,X)\big\}\\
    &\mathcal{E}_{4a} = \big\{ \big(Y^{n},S_{d}^{n},A^{n}(M_{1})\big)\notin
    T_{\epsilon}^{(n)}(Y,S_{d},A)\big\}\\
     &\mathcal{E}_{4b} = \big\{ \big(Y^{n},S_{d}^{n},A^{n}(\tilde{m}_{1})\big)\in
    T_{\epsilon}^{(n)}(Y,S_{d},A)\ \mbox{for some} \ \tilde{m}_{1} \in \mathcal{M}_{1}^{(n)}, \tilde{m}_{1} \neq M_{1} \big\}\\
      &\mathcal{E}_{5a} = \big\{ \big(Y^{n},S_{d}^{n},A^{n}(M_{1}),\check{S}_{e}^{n}(M_{1},M_{2},J), X^{n}(M_{1},M_{2},J)\big)\notin
    T_{\epsilon}^{(n)}(Y,S_{d},A,S_{e},X)\big\} \\
      &\mathcal{E}_{5b} = \big\{\big(Y^{n},S_{d}^{n},A^{n}(M_{1}),\check{S}_{e}^{n}(M_{1},\tilde{m}_{2},\tilde{j}), X^{n}(M_{1},\tilde{m}_{2},\tilde{j})\big)\in T_{\epsilon}^{(n)}(Y,S_{d},A,S_{e},X)\\
      &\qquad \qquad \mbox{for some}\ (\tilde{m}_{2},\tilde{j}) \in \mathcal{M}_{2}^{(n)} \times \mathcal{J}^{(n)}, (\tilde{m}_{2},\tilde{j}) \neq (M_{2},J)\big\}.
    \end{align*}

The probability of error events can be bounded by
\begin{align*}
    \mathrm{Pr}(\mathcal{E}) %&= \mathrm{Pr}\big((\hat{m},\hat{x}^{n})\neq (m,x^{n}) \big)\\
    &\leq \mathrm{Pr}(\mathcal{E}_{1}) + \mathrm{Pr}(\mathcal{E}_{2}\cap \mathcal{E}_{1}^{c}) + \mathrm{Pr}(\mathcal{E}_{3a}\cap \mathcal{E}_{2}^{c}) + \mathrm{Pr}(\mathcal{E}_{3b}\cap \mathcal{E}_{3a}^{c}\cap \mathcal{E}_{2}^{c}) + \mathrm{Pr}(\mathcal{E}_{4a}\cap \mathcal{E}_{3}^{c}) + \mathrm{Pr}(\mathcal{E}_{4b})\\
    &\qquad + \mathrm{Pr}(\mathcal{E}_{5a}\cap \mathcal{E}_{3}^{c})+\mathrm{Pr}(\mathcal{E}_{5b}),
\end{align*}
where $\mathcal{E}_{3} \triangleq \mathcal{E}_{3a} \cup \mathcal{E}_{3b}$ and $\mathcal{E}_{i}^{c}$ denotes the complement of event $\mathcal{E}_{i}$.

$1)$ Since $A^{n}(M_{1})$ is i.i.d. according to $P_{A}$, by the LLN we have $\mathrm{Pr}(\mathcal{E}_{1})\rightarrow 0$ as $n\rightarrow \infty$.

$2$) Consider the event $E_{1}^{c}$ where we have $A^{n}(M_{1})\in T_{\epsilon_{0}}^{(n)}(A)$.  Since $(S_{d}^{n},S_{e}^{n})$ is distributed according to $\prod_{i=1}^{n}P_{S_{e},S_{d}|A}\big(s_{e,i},s_{d,i}|a_{i}\big)$, by the conditional typicality lemma \cite{ElGamalKim}, we have that $\mathrm{Pr}\big(\mathcal{E}_{2} \cap \mathcal{E}_{1}^{c}\big)\rightarrow 0$ as $n\rightarrow \infty$.

$3a$) Consider the event $\mathcal{E}_{2}^{c}$ where we have $\big(S_{e}^{n},S_{d}^{n},A^{n}(M_{1})\big)\in
    T_{\epsilon_{1}}^{(n)}(S_{e},S_{d},A)$. It follows from the property of typical sequences \cite{ElGamalKim} that $P_{S^{n}_{e}|A}(s^{n}_{e}|a^{n}) \geq 2^{-n[H(S_{e}|A)+ \delta{\epsilon_{1}}]}$.  Since both $S_{e}^{n}$ and $\check{S}_{e}^{n}$ are i.i.d. according to $P_{S_{e}|A}$, we have $\mathrm{Pr}\big(\mathcal{E}_{31} \cap \mathcal{E}_{2}^{c}\big)\rightarrow 0$ as $n\rightarrow \infty$ if $\frac{1}{n}\log |\mathcal{J}^{(n)}| > H(S_{e}|A)+ \delta_{\epsilon_{1}}$, where $\delta_{\epsilon_{1}}\rightarrow 0$ as $\epsilon_{1}\rightarrow 0$.

$3b$) Consider the event $\mathcal{E}_{3a}^{c}$ where $J$ is selected and $S_{e}^{n} = \check{S}_{e}^{n}(M_{1},M_{2},J)$. Since $X^{n}$ is i.i.d. according to $\prod_{i=1}^{n}P_{X|S_{e},A}(x_{i}|s_{e,i},a_i)$, by the conditional typicality lemma, we have that $\mathrm{Pr}\big(\mathcal{E}_{3b} \cap \mathcal{E}_{3a}^{c}\cap \mathcal{E}_{2}^{c}\big)\rightarrow 0$ as $n\rightarrow \infty$.

$4a$) Consider the event $\mathcal{E}_{3}^{c}$ where we have $\big(S_{e}^{n},A^{n}(M_{1}),X^{n}(M_{1},M_{2},J)\big)\in T_{\epsilon_{1}}^{(n)}(S_{e},A,X)$. Since we have $S_{d}-(A,S_{e})-X$ forms a Markov chain and $S_{d}^{n}$ is distributed according to $\prod_{i=1}^{n}P_{S_{d}|A,S_{e}}\big(s_{d,i}|a_{i},s_{e,i}\big)$, by the conditional typicality lemma, we have that $\mathrm{Pr}\big((S_{d}^{n},S_{e}^{n},A^{n},X^{n})\in T_{\epsilon}^{(n)}(S_{d},S_{e},A,X) \big) \rightarrow 1$ as $n\rightarrow \infty$. And since we have the Markov chain $A-(X,S_{e},S_{d})-Y$ and $Y^{n}$ is distributed according to $\prod_{i=1}^{n}P_{Y|X,S_{e},S_{d}}(y_i|x_i,s_{e,i},s_{d,i})$, by using once again the conditional typicality lemma, it follows that $\mathrm{Pr}\big(\big(Y^{n}, S_{e}^{n},S_{d}^{n}, A^{n}(M_{1}),X^{n}(M_{1},M_{2},J)\big)\in T_{\epsilon}^{(n)}(Y,A,S_{e},S_{d},X)\big)\rightarrow 1$ as $n\rightarrow \infty$. This also implies that $\mathrm{Pr}(\mathcal{E}_{4a}\cap \mathcal{E}_{3}^{c})\rightarrow 0$ as $n\rightarrow \infty$.

$4b$) By the packing lemma \cite{ElGamalKim}, we have $\mathrm{Pr}\big(\mathcal{E}_{4b} \big)\rightarrow 0$ as $n\rightarrow \infty$ if $\frac{1}{n}\log |\mathcal{M}_{1}^{(n)}|< I(A;Y,S_{d})- \delta_{\epsilon}$, where $\delta_{\epsilon}\rightarrow 0$ as $\epsilon \rightarrow 0$.

$5a$) As in $\mathcal{E}_{4a}$) we have $\mathrm{Pr}(\mathcal{E}_{5a}\cap \mathcal{E}_{3}^{c})\rightarrow 0$ as $n\rightarrow \infty$.

$5b$) Averaging over all $J=j$, by the packing lemma where $\check{S}_{e}^{n}$ is i.i.d. according to $\prod_{i=1}^{n}P_{S_{e}|A}(\check{s}_{e,i}|a_{i})$ and $X^{n}$ is i.i.d. according to $\prod_{i=1}^{n}P_{X|S_{e},A}(x_i|\check{s}_{e,i},a_{i})$, we have $\mathrm{Pr}\big(\mathcal{E}_{5b} \big)\rightarrow 0$ as $n\rightarrow \infty$ if $\frac{1}{n}\log |\mathcal{M}_{2}^{(n)}|+\frac{1}{n}\log |\mathcal{J}^{(n)}|< I(S_{e},X;Y,S_{d}|A)- \delta_{\epsilon}$.

Finally, by combining the bounds on the code rates,
\begin{align*}
 \frac{1}{n}\log |\mathcal{J}^{(n)}| &> H(S_{e}|A)+ \delta_{\epsilon_{1}} \\
 \frac{1}{n}\log |\mathcal{M}_{1}^{(n)}|&< I(A;Y,S_{d})- \delta_{\epsilon} \\
 \frac{1}{n}\log |\mathcal{M}_{2}^{(n)}|+\frac{1}{n}\log |\mathcal{J}^{(n)}|&< I(S_{e},X;Y,S_{d}|A)- \delta_{\epsilon},
\end{align*}
where $\epsilon >0$ can be made arbitrarily small with increasing block length $n$, we have shown that, for any $\delta>0$,
%so that for sufficiently large $n$, $P_{e}\leq 6\epsilon < \delta$.
%Since $\epsilon >0$ can be made arbitrarily small with increasing block length $n$, we have $\delta_{\epsilon}\rightarrow 0$ as $n\rightarrow \infty$. Therefore,
with $n$ sufficiently large, $\mathrm{Pr}(\mathcal{E})< \delta$ when $R \leq I(A;Y,S_{d})+I(S_{e},X;Y,S_{d}|A)-H(S_{e}|A)$ and $I(S_{e},X;Y,S_{d}|A)-H(S_{e}|A)>0$ for some $P_{A}(a)P_{S_{e},S_{d}|A}(s_{e},s_{d}|a)P_{X|A,S_{e}}(x|a,s_{e})P_{Y|X,S_{e},S_{d}}(y|x,s_{e},s_{d})$. Again, we note that the latter condition is for the successful two-stage coding, i.e., we can split the message into two parts with positive rates.
This together with a random coding argument concludes the achievability proof. \hfill$\blacksquare$

%%%%%%%%%%%%%%%%%%%%%%%%%%%%%%%%%%%%%%%%%%%%%%%%%%%%%%%%%%%%%%%%%%%%%%%%%%%%%%%%%%%%%%%%%%%%%%%%%%%%%%%%%%%%%%%
\subsection{Converse Proof of Proposition \ref{eq:C_Se}}
We show that, for any achievable rate $R$, it follows that $R \leq I(A,X,S_{e};Y,S_{d})-H(S_{e}|A)$ and $0 \leq I(S_{e},X;Y,S_{d}|A)-H(S_{e}|A)$ for some $P_{A}(a)P_{S_{e},S_{d}|A}(s_{e},s_{d}|a)P_{X|A,S_{e}}(x|a,s_{e})P_{Y|X,S_{e},S_{d}}(y|x,s_{e},s_{d})$. From the problem formulation, we can write the joint probability mass function,
\begin{align*}
&P_{M,A^{n},S_{e}^{n},S_{d}^{n},X^{n},Y^{n},\hat{M},\hat{S}_{e}^{n}}(m,a^{n},s_{e}^{n},s_{d}^{n},x^{n},y^{n},\hat{m},\hat{s}_{e}^{n}) \nonumber \\ & = \frac{1_{\{f_{a}^{(n)}(m)=a^{n},f^{(n)}(m,s_{e}^{n})=x^{n},g_{s_{e}}^{(n)}(y^{n},s_{d}^{n})=\hat{s}_{e}^{n},g_{m}^{(n)}(y^{n},s_{d}^{n})=\hat{m}\}}}{|\mathcal{M}^{(n)}|}\cdot \prod_{i=1}^{n}P_{S_{e},S_{d}|A}(s_{e,i},s_{d,i}|a_{i})P_{Y|X,S_{e},S_{d}}(y_{i}|x_{i},s_{e,i},s_{d,i}),
\end{align*}
where $M$ is chosen uniformly at random from the set $\mathcal{M}^{(n)}=\{1,2,\ldots,|\mathcal{M}^{(n)}|\}$.

Let us assume that a specific sequence of $(|\mathcal{M}^{(n)}|,n)$ codes exists such that the average error probabilities $P_{m,e}^{(n)}=\delta_{n}' \leq \delta_{n}$, $P_{s_{e},e}^{(n)}=\delta_{n}' \leq \delta_{n}$, and $ \log|\mathcal{M}^{(n)}| = n(R-\delta_{n}') \geq n(R-\delta_{n})$, with $\lim_{n \rightarrow \infty}\delta_{n} =0$.
Then standard properties of the entropy function give
\begin{align*}
  n(R-\delta_{n}) &\leq \log|\mathcal{M}^{(n)}|= H(M) \\
   %&= H(M)-H(X^{n},M|Y^{n})+ H(X^{n},M|Y^{n})\\
   &= H(M)-H(X^{n},S_{e}^{n},M|Y^{n},S_{d}^{n})+ H(M|Y^{n},S_{d}^{n})+H(S_{e}^{n}|M,Y^{n},S_{d}^{n})+ H(X^{n}|M,S_{e}^{n},Y^{n},S_{d}^{n})\\
   &\overset{(*)}{\leq} H(M)-H(X^{n},S_{e}^{n},M|Y^{n},S_{d}^{n})+ H(M|Y^{n},S_{d}^{n})+H(S_{e}^{n}|Y^{n},S_{d}^{n})
\end{align*}
where $(*)$ holds since $X^{n}=f^{(n)}(M,S^{n}_{e})$ and $f^{(n)}(\cdot)$ is a deterministic function.

Consider the last two terms in the above inequality. Similarly to \cite{Sumszyk2009}, by Fano's inequality, we get
\begin{align*}
H(M|Y^{n},S_{d}^{n}) &\leq h(\delta_{n}') + \delta_{n}'\cdot \log(2^{n(R-\delta_{n}')}-1) = n\epsilon^{(m)}_{n}, \\
H(S_{e}^{n}|Y^{n},S_{d}^{n}) &\leq h(\delta_{n}') + \delta_{n}'\cdot \log(|\mathcal{S}_{e}|^{n}-1) = n\epsilon^{(s_{e})}_{n},
%H(X^{n}|Y^{n},S_{d}^{n}) &\leq h(P_{x,e}^{(n)}) + P_{x,e}^{(n)}\cdot \log(|\mathcal{X}|^{n}-1) = n\delta^{(x)}_{n},
\end{align*}
where $h(\cdot)$ is the binary entropy function, and $\epsilon^{(m)}_{n}\rightarrow 0, \epsilon^{(s_{e})}_{n} \rightarrow 0$ as $n \rightarrow \infty$.

Let $n\epsilon^{(m)}_{n} +n \epsilon^{(s_{e})}_{n} \triangleq n\epsilon_{n}$, where $\epsilon_{n}$ satisfies $\lim_{n\rightarrow \infty}\epsilon_{n} =0$. Now we continue the chain of inequalities and get
\begin{align*}
n(R-\delta_{n}) & \leq H(M)-H(X^{n},S_{e}^{n},M|Y^{n},S_{d}^{n})+n\epsilon_{n}\\
    &= H(M,S^{n}_{e})-H(S^{n}_{e}|M)-H(X^{n},S_{e}^{n},M|Y^{n},S_{d}^{n})+n\epsilon_{n}\\
    &\overset{(a)}{=} H(M,S^{n}_{e},X^{n})-H(S^{n}_{e}|M,A^{n}) -H(X^{n},S_{e}^{n},M|Y^{n},S_{d}^{n})+n\epsilon_{n}\\
    &\overset{(b)}{=} H(M,S^{n}_{e},X^{n})-H(S^{n}_{e}|A^{n})-H(X^{n},S_{e}^{n},M|Y^{n},S_{d}^{n})+n\epsilon_{n}\\
    %&= H(X^{n},M)+H(S^{n}_{e}|X^{n},M)-H(S^{n}_{e}|A^{n})\\ &\qquad -H(X^{n},M|Y^{n})+n\epsilon_{n}\\
    &= I(X^{n},S_{e}^{n},M;Y^{n},S_{d}^{n})-H(S^{n}_{e}|A^{n})+n\epsilon_{n}
\end{align*}
where $(a)$ follows from the fact that $X^{n}=f^{(n)}(M,S^{n}_{e})$ and $A^{n}=f^{(n)}_{a}(M)$, $(b)$ follows since $S_{e}^{n}$ is independent of $M$ given $A^{n}$.

Continuing the chain of inequalities, we get
\begin{align}
& n(R-\delta_{n}-\epsilon_{n}) \nonumber\\
& \overset{(a)}{\leq} \sum_{i=1}^{n} I(X^{n},S_{e}^{n},M;Y_{i},S_{d,i}|Y^{i-1},S_{d}^{i-1})-H(S_{e,i}|A_{i}) \nonumber \\
& \overset{(b)}{=} \sum_{i=1}^{n} H(Y_{i},S_{d,i}|Y^{i-1},S_{d}^{i-1})-H(Y_{i},S_{d,i}|Y^{i-1},S_{d}^{i-1},X^{n},S_{e}^{n},M,A_{i}) -H(S_{e,i}|A_{i})\nonumber \\
& \overset{(c)}{\leq} \sum_{i=1}^{n}H(Y_{i},S_{d,i})-H(Y_{i},S_{d,i}|X_{i},S_{e,i},A_{i})-H(S_{e,i}|A_{i})\nonumber  \\
& = \sum_{i=1}^{n}I(A_{i},S_{e,i},X_{i};Y_{i},S_{d,i})-H(S_{e,i}|A_{i}), \label{eq:converseR_se}
\end{align}
where
$(a)$ follows from the memoryless property of the channel $P_{S_{e}|A}$, $(b)$ follows from the fact that $A^{n}=f^{(n)}_{a}(M)$, and $(c)$ follows from the Markov chain $(Y^{i-1},S_{d}^{i-1},X^{n\setminus i},S_{e}^{n\setminus i},M)-(X_{i},S_{e,i},A_{i})-(Y_{i},S_{d,i})$ and that conditioning reduces entropy.

%%%%%%%%%%%%%%%%%%%%%%%%%%%%%%%%%%%%%%%%%%%%%%%%%%%%%%%%%%%%%%%%%%%%%%%%%%%%%%%%%%%%%%%%%%%%%%%%%%%%%%%%%%%%%%%%%%%%
Next we prove the constraint which does not involve rate of the communication. From the standard properties of the entropy function, we observe that
\begin{align*}
  &0 \leq H(M|A^{n}) \\
  &\overset{(*)}{=} H(M|A^{n})-H(X^{n},S_{e}^{n},M|Y^{n},S_{d}^{n},A^{n})+H(M|Y^{n},S_{d}^{n},A^{n})+ H(S_{e}^{n}|M,Y^{n},S_{d}^{n},A^{n})\\
  & \leq H(M|A^{n})-H(X^{n},S^{n}_{e},M|Y^{n},S_{d}^{n},A^{n})+ H(M|Y^{n},S_{d}^{n})+ H(S_{e}^{n}|Y^{n},S_{d}^{n}),
\end{align*}
where $(*)$ follows from the fact that $X^{n}=f^{(n)}(M,S^{n}_{e})$, and $f^{(n)}(\cdot)$ is a deterministic function.

Again, applying Fano's inequality to last two terms in the above inequality, we get
\begin{align}
-n\epsilon_{n}  &\leq H(M|A^{n})-H(X^{n},S^{n}_{e},M|Y^{n},S_{d}^{n},A^{n}) \nonumber \\
    &= H(M,S^{n}_{e}|A^{n})-H(S^{n}_{e}|M,A^{n})-H(X^{n},S^{n}_{e},M|Y^{n},S_{d}^{n},A^{n}) \nonumber \\
    &\overset{(a)}{=} H(M,S^{n}_{e},X^{n}|A^{n})-H(S^{n}_{e}|M,A^{n}) -H(X^{n},S^{n}_{e},M|Y^{n},S_{d}^{n},A^{n}) \nonumber \\
    &\overset{(b)}{=} H(M,S^{n}_{e},X^{n}|A^{n})-H(S^{n}_{e}|A^{n})-H(X^{n},S^{n}_{e},M|Y^{n},S_{d}^{n},A^{n}) \nonumber \\
    %&= H(X^{n},M|A^{n})+H(S^{n}_{e}|X^{n},M,A^{n})-H(S^{n}_{e}|A^{n})\\ &\qquad -H(X^{n},M|Y^{n},A^{n})\\
    &= I(X^{n},S^{n}_{e},M;Y^{n},S_{d}^{n}|A^{n})-H(S^{n}_{e}|A^{n})\nonumber \\
    & \overset{(c)}{=} \sum_{i=1}^{n} I(X^{n},S^{n}_{e},M;Y_{i},S_{d,i}|Y^{i-1},S_{d}^{i-1},A^{n})-H(S_{e,i}|A_{i})\nonumber \\
    & \overset{(d)}{\leq} \sum_{i=1}^{n}H(Y_{i},S_{d,i}|A_{i})-H(Y_{i},S_{d,i}|X_{i},S_{e,i},A_{i})-H(S_{e,i}|A_{i}) \nonumber \\
    & = \sum_{i=1}^{n}I(S_{e,i},X_{i};Y_{i},S_{d,i}|A_{i})-H(S_{e,i}|A_{i}), \label{eq:converse0_se}
\end{align}
where $(a)$ follows from the fact that $X^{n}=f^{(n)}(M,S^{n}_{e})$, $(b)$ holds since $S_{e}^{n}$ is independent of $M$ given $A^{n}$, $(c)$ follows from the memoryless property of the channel $P_{S_{e}|A}$, and $(d)$ follows from the Markov chain\\ $(Y^{i-1},S_{d}^{i-1},X^{n\setminus i},S_{e}^{n\setminus i},A^{n\setminus i},M)-(X_{i},S_{e,i},A_{i})-(Y_{i},S_{d,i})$ and that conditioning reduces entropy.

Let $Q$ be a random variable uniformly distributed over $\{1, \ldots, n\}$ and independent of $(M,A^n,X^n,S_{e}^n,S_{d}^n,Y^{n})$,
we can rewrite \eqref{eq:converseR_se} and \eqref{eq:converse0_se} as
\begin{align*}
R &\leq \frac{1}{n}\sum_{i=1}^{n} I(A_{i},S_{e,i},X_{i};Y_{i},S_{d,i}|Q=i)-H(S_{e,i}|A_{i},Q=i)+\delta_{n}+\epsilon_{n}\\
  &=  I(A_{Q},S_{e,Q},X_{Q};Y_{Q},S_{d,Q}|Q)-H(S_{e,Q}|A_{Q},Q)+\delta_{n}+\epsilon_{n}
\end{align*}
and
\begin{align*}
0 &\leq \frac{1}{n}\sum_{i=1}^{n} I(S_{e,i},X_{i};Y_{i},S_{d,i}|A_{i},Q=i)-H(S_{e,i}|A_{i},Q=i)+\epsilon_{n}\\
  &= I(S_{e,Q},X_{Q};Y_{Q},S_{d,Q}|A_{Q},Q)-H(S_{e,Q}|A_{Q},Q)+\epsilon_{n}.
\end{align*}
Now since we have that $P_{S_{e,Q},S_{d,Q}|A_{Q}}=P_{S_{e},S_{d}|A}$, $P_{Y_{Q}|X_{Q},S_{e,Q},S_{d,Q}}=P_{Y|X,S_{e},S_{d}}$, and $A_{Q}-(X_{Q},S_{e,Q},S_{d,Q})-Y_{Q}$ forms a Markov chain, we identify $A \triangleq A_{Q}$, $S_{e} \triangleq S_{e,Q}$, $S_{d} \triangleq S_{d,Q}$, $X \triangleq X_{Q}$, and $Y \triangleq Y_{Q}$ to finally obtain
\begin{align*}
    R &\leq I(A,S_{e},X;Y,S_{d}|Q)-H(S_{e}|A,Q)+\delta_{n}+\epsilon_{n}\\
    \mbox{and} \ \ 0 &\leq I(S_{e},X;Y,S_{d}|A,Q)-H(S_{e}|A,Q)+\epsilon_{n},
\end{align*}
for some joint distribution
\begin{align}
&P_{Q}(q)P_{A|Q}(a|q)P_{S_{e},S_{d}|A}(s_{e},s_{d}|a)P_{X|A,S_{e},Q}(x|a,s_{e},q) P_{Y|X,S_{e},S_{d}}(y|x,s_{e},s_{d}). \label{eq:pdf3}
\end{align}
From the joint distribution in \eqref{eq:pdf3} and the derivation of \eqref{eq:chain1}-\eqref{eq:chain5} (see Lemma \ref{lemma:markovchain2}), we have the Markov chains $Q-A-(S_{e},S_{d})$ and $Q-(X,A,S_{e})-(Y,S_{d})$.

To this end, we note that under any distribution of the form above, we have
\begin{align*}
    I(A,S_{e},X;Y,S_{d}|Q)-H(S_{e}|A,Q)
                                &= H(Y,S_{d}|Q)-H(Y,S_{d}|A,X,S_{e},Q)-H(S_{e}|A,Q) \\
                                & \leq H(Y,S_{d})-H(Y,S_{d}|A,X,S_{e},Q)-H(S_{e}|A,Q)\\
                                & \overset{(*)}{=} H(Y,S_{d})-H(Y,S_{d}|A,X,S_{e})-H(S_{e}|A)\\
                                &=I(A,S_{e},X;Y,S_{d})-H(S_{e}|A),
\end{align*}
and
\begin{align*}
I(S_{e},X;Y,S_{d}|A,Q)-H(S_{e}|A,Q)
                                &= H(Y,S_{d}|A,Q)-H(Y,S_{d}|A,X,S_{e},Q)-H(S_{e}|A,Q) \\
                                & \leq H(Y,S_{d}|A)-H(Y,S_{d}|A,X,S_{e},Q)-H(S_{e}|A,Q)\\
                                & \overset{(*)}{=} H(Y,S_{d}|A)-H(Y,S_{d}|A,X,S_{e})-H(S_{e}|A)\\
                                &=I(S_{e},X;Y,S_{d}|A)-H(S_{e}|A),
\end{align*}
where both inequalities $(*)$ follows from the Markov chains $(Y,S_{d})-(X,A,S_{e})-Q$ and $S_{e}-A-Q$, and the joint distribution of $(A,S_{e},S_{d},X,Y)$ is of the form

\begin{align*}
&\sum_{q \in \mathcal{Q}}P_{Q}(q)P_{A|Q}(a|q)P_{S_{e},S_{d}|A}(s_{e},s_{d}|a)P_{X|A,S_{e},Q}(x|a,s_{e},q)P_{Y|X,S_{e},S_{d}}(y|x,s_{e},s_{d})\\ & \qquad=P_{A}(a)P_{S_{e},S_{d}|A}(s_{e},s_{d}|a)P_{X|A,S_{e}}(x|a,s_{e})P_{Y|X,S_{e},S_{d}}(y|x,s_{e},s_{d}).
\end{align*}
The proof is concluded by taking the limit $n \rightarrow \infty$. \hfill$\blacksquare$

\bibliographystyle{IEEEtran}
\bibliography{IEEEabrv,bibliography}

\end{document}